%% file: TailBinaryPanel.tex
\documentclass[12pt]{article}
\usepackage{xr-hyper}
\input{TailBinaryPanel_header.tex}
\begin{document}

\title{\ifsubmission\else\vspace{-2.5em}\fi
    Binary Outcome Models with Extreme Covariates:\\
Estimation and Prediction\thanks{
    \ifsubmission
Liu: Department of Economics, University of Pittsburgh, 4527 Posvar Hall, 230 S.\ Bouquet St., Pittsburgh, PA 15260, email: \texttt{laura.liu@pitt.edu}. Wang: Department of Economics, Syracuse University, 110 Eggers Hall, Syracuse, NY, 13244, email: \texttt{ywang402@syr.edu}.
\else \texttt{laura.liu@pitt.edu} (Liu) and \texttt{ywang402@syr.edu} (Wang). We thank Frank Diebold, Christian Haefke, Bo Honor\'{e}, Roger Klein, Ulrich M\"{u}ller, Frank Schorfheide, and seminar participants at the FRB Chicago, Monash University, University of Melbourne, Reserve Bank of Australia, University of Sydney, FRB Philadelphia, UC Irvine, UCSD, and Penn State, as well as conference participants at the Applied Time Series Econometrics Workshop at the FRB St.\ Louis, Dolomiti Macro Meetings, AiE Conference and Festschrift in Honor of Joon Y.\ Park, Midwest Econometrics Group Annual Meeting, Greater New York Econometrics Colloquium, Women in Macroeconomics Workshop II, NBER Summer Institute, and NBER-NSF Time Series Conference for helpful comments and discussions. The authors are solely responsible for any remaining errors.\fi}}
\author{Laura Liu \\
\textit{University of Pittsburgh} \and Yulong Wang \\
\textit{Syracuse University}}
\date{First version: July 7, 2024\\
This version: \today}
\maketitle
\ifsubmission\else\vspace{-2em}\fi
\begin{abstract}
    This paper presents a novel semiparametric method to study the effects of extreme events on binary outcomes and subsequently forecast future outcomes. Our approach, based on Bayes' theorem and regularly varying (RV) functions, facilitates a Pareto approximation in the tail without imposing parametric assumptions beyond the tail. We analyze cross-sectional as well as static and dynamic panel data models, incorporate additional covariates, and accommodate the unobserved unit-specific tail thickness and RV functions in panel data. We establish consistency and asymptotic normality of our tail estimator, and show that our objective function converges to that of a panel Logit regression on tail observations with the log extreme covariate as a regressor, thereby simplifying implementation. The empirical application assesses whether small banks become riskier when local housing prices sharply decline, a crucial channel in the 2007--2008 financial crisis.
% This paper presents a novel semiparametric method to study the effects of extreme events on binary outcomes and forecast future outcomes, which is particularly relevant given the recent extreme events. Our approach, based on Bayes' theorem and regularly varying (RV) functions, facilitates a Pareto approximation in the tail while imposing no parametric assumptions on the relationship between covariates and outcomes beyond the tail. We analyze cross-sectional as well as static and dynamic panel data models, incorporate additional covariates in both setups, and accommodate the unobserved unit-specific tail thickness and RV functions in the panel setup. We establish consistency and asymptotic normality of the proposed tail estimator. We demonstrate that, under regularity conditions, our objective function converges to that of a panel Logit regression on tail observations with the log extreme covariate as a regressor, which simplifies implementation and facilitates empirical research. We evaluate the finite-sample properties of the proposed tail estimator in Monte Carlo simulations. In the empirical application, we examine a panel of small banks to assess whether they become riskier when local housing prices experience a significant decline, a crucial channel in the 2007--2008 financial crisis.


\end{abstract}

\newpage

\section{Introduction}
Binary outcome models are widely used in both empirical microeconomics and empirical macroeconomics research. For example, microeconomic studies may be interested in the determinants of individuals' labor force participation decisions, and macroeconomic analyses may seek to forecast the probability of recessions or country defaults. Recently, extreme events, such as the Covid-19 pandemic and its aftermath, sustained periods of high inflation, and increasingly frequent extreme weather, have highlighted the importance of studying the effect of extreme covariates on binary outcomes.

As an illustration, consider a simple example with cross-sectional data, where we have a random sample of $\{Y,X\}$ with binary outcome $Y\in \{0,1\}$ and continuous covariate $X\in \mathbb{R}$. Here and throughout, we follow the convention that uppercase letters denote random variables, while lowercase letters represent their realized values.
Depending on the empirical context, one may be interested in the conditional probability 
\begin{equation}
\pi(x)=\P\left(Y=1|X=x\right),\label{eq:c-s-simple-prob}
\end{equation}
the partial effect $\partial \pi(x) /\partial x$, and the elasticity, when $x$ takes extreme values. Given a random sample with $N$ observations, we characterize the extremeness by letting $x\rightarrow\infty$ as $N\rightarrow\infty$.\footnote{Here we focus on the right tail without loss of generality. In practice, one can conduct a similar analysis for the left tail, and we allow for different tail thickness in each tail.}\textsuperscript{,}\footnote{\label{fn:moderate-large-x}In finite sample, our approach still shows significant improvement for moderately large $x$, roughly below the 20th or above 80th percentile of the $X$ distribution, and the corresponding $\pi(x)$ may not be very close to 0 or 1 either: see Figure \ref{fig:sim-exp1-PY}, for example.}  
In many cases, $\pi(x)\rightarrow 0$ or 1 as $x\rightarrow\infty$, so we combine both scenarios and define a unified measure of extreme elasticity as
\begin{equation}
\delta(x)=\frac{\partial \left(\pi(x)\left(1-\pi(x)\right)\right) }{\partial x}\frac{x}{\pi(x)\left(1-\pi(x)\right)}.\label{eq:c-s-simple-elas}
\end{equation}
Further details are provided in Proposition \ref{prop:elas} and Remark \ref{rm:elas}.
We will use sovereign debt and country default as a running example for intuitive explanations. In this context, the outcome $Y$ indicates whether the country defaults or not, the covariate $X$ is the debt-to-GDP ratio, and the conditional probability $\pi(x)$ represents the probability of default when the debt-to-GDP ratio is particularly high, which can be viewed as a counterfactual probability or a predictive probability.

Existing methods face limitations when dealing with extreme values, a common feature in economic data that manifests as heavy tails \citep{gabaix2009,Gabaix2016JEP}. Parametric methods, such as Logit and Probit, assume threshold-crossing models with thin-tailed error distributions, which can lead to significant misspecification bias, particularly in the tails. Conversely, nonparametric methods, such as kernel, sieve, and spline estimators, may encounter difficulties due to limited information in the tail, resulting in highly inefficient estimators with large variances. This inefficiency can lead to imprecise forecasts, especially for extreme values.

To address these challenges, we propose a novel semiparametric approach based on Bayes' theorem and RV functions, which offer a flexible framework that encompasses a wide range of distributions. For heavy-tailed distributions, the RV condition naturally leads to a Pareto approximation in the tail (see Section \ref{sec:c-s-simple} for details), while making no parametric assumptions on the relationship between covariates and outcomes beyond the tail region.\footnote{\label{com:sp-kernel}To see the semiparametric nature of our tail estimator, note that traditional kernel methods focus on local observations close to a specific $x$ value of interest, leaving the model far from $x$ unspecified. Similarly, our tail estimator focuses on a specific local region, namely the extreme tail where $x\rightarrow\infty$, leaving the model for the middle part unspecified.} This is crucial because patterns in the tail and middle can differ substantially, as illustrated in the log-log plots in Figures \ref{fig:sim-exp1-loglog} and \ref{fig:app-loglog}. Specifically, by Bayes' rule, we have that 
\begin{equation}
\P\left(Y=1|X=x\right) =\frac{f_{X|Y}\left(x|1\right) \P
\left(Y=1\right) }{f_{X|Y}\left(x|1\right) \P\left(Y=1\right)
+f_{X|Y}\left( x|0\right) \P\left( Y=0\right)},\label{eq:c-s-simple-bayes}
\end{equation}
where $f_{X|Y}\left(\cdot|y\right) $ denotes the conditional density
function of $X$ given $Y=y$. When $x$ takes extreme values, the RV properties allow us to approximate $f_{X|Y}\left( \cdot |y\right) $ by $x^{-\alpha^\py-1}$ up to some constant, where $\alpha^\py$ is the Pareto exponent for $y\in\{0,1\}$. Then, the extreme elasticity $\delta(x)$ converges to $-\left|\alpha^\pl-\alpha^\po\right|$, as $x\rightarrow \infty $. We establish formal asymptotic results in a more general setting with additional non-extreme covariates $Z=z$, where we hold $z$ fixed and let $x\rightarrow\infty$. Note that given the limited data available in the tail for extreme $x$ values, the convergence rate of the tail estimator is slower than the standard $\sqrt N$-rate. 

Furthermore, our method not only addresses the limited information in the tail but also effectively handles unobserved unit-specific tail parameters and unit-specific RV functions in panel data. The heterogeneity in the RV functions naturally accommodates unit-specific scale parameters as well. We show that, under regularity conditions, our objective function is asymptotically equivalent to that of a panel Logit regression on tail observations, using the log extreme covariate as a regressor. This equivalence simplifies the estimation procedure and makes the method more convenient for empirical studies.

Here we consider panel data with large $N$ and either small or large $T$. For small $T$, we can use the conditional MLE to eliminate unit-specific tail effects. The proof nontrivially extends \citet{WangTsai2009} to panel data with independent but not necessarily identically distributed (i.n.i.d.) observations. For large $T$, we can directly estimate the unit-specific parameters using existing methods for  large-$T$ panel Logit and further provide unit-specific forecasts.  Finally, we expand our discussion to dynamic panel data models, incorporating lagged outcomes into the analysis.

We assess the finite-sample performance of our tail estimator through Monte Carlo simulations under various specifications. Experiment 1 focuses on estimation accuracy in cross-sectional data, and Experiment 2 examines out-of-sample forecasts in panel data with large $N$ and large $T$. Results show that our tail estimator outperforms parametric and nonparametric alternatives, particularly in capturing heavy-tailed behavior, reducing misspecification bias, and providing more accurate forecasts.

In our empirical application, we analyze the impact of local housing price declines on the riskiness of small banks, which played an important role in the 2007--2008 financial crisis. We construct a panel dataset of loan charge-off rates of small banks, local housing prices, and unemployment rates, and again find that the tail estimator yields the most accurate pseudo out-of-sample forecasts among all methods considered. In addition, this empirical analysis also helps reveal interesting heterogeneity patterns in riskiness among small banks, varying across geographic regions and bank characteristics.

\paragraph{Related literature.}
Our work draws on a wide range of econometric literature, including binary outcome models, panel data, extreme value theory, and forecasting. 

First, our work builds on a large literature on the binary outcome models, especially those for panel data. However, we focus on the analysis of extreme events instead of mid-sample properties as in the existing literature.

In binary cross-sectional models, various methods have been well studied in the literature, including parametric approaches (e.g., Logit and Probit), nonparametric approaches (e.g., \citet{matzkin1992nonparametric}), and semiparametric approaches, such as the single-index model \citep{klein1993efficient} and the maximum score estimator \citep{manski1985semiparametric,horowitz1992smoothed}. Also see \citet{horowitz2001binary} for a review. Our approach falls within the semiparametric framework, approximating the heavy tails using RV functions while making no parametric assumptions in the middle. Notably, much of this existing literature assumes threshold-crossing models, whereas our approach does not require such a structure: see Section \ref{sec:threshold-xing} for a detailed comparison.

In contrast to common methods that directly work with $\P(Y=1|X=x)$, we employ Bayes' theorem to ``reverse'' the conditioning set. Similar transformations based on Bayes' theorem have also been applied in the discriminant analysis \citep{amemiya1985advanced} with Gaussian $X|Y=y$, as well as in the semiparametric single-index model in \citet{klein1993efficient}. In our context of extreme events, this transformation offers a novel way to disentangle tail behavior. Unlike those previous works focusing on mid-sample properties, we analyze tail properties with only tail observations, necessitating the development of nonstandard asymptotic theory.

In binary panel data models, the unobserved unit-specific heterogeneity can be challenging to handle due to the incidental parameter problem in the presence of nonlinear model structures. In addition to parametric, nonparametric, and semiparametric approaches, the binary panel literature can also be categorized by assumptions on unit-specific heterogeneity. See Chapter 15 in \citet{Wooldridge2010} and Chapter 7 in \citet{hsiao2022analysis} for textbook discussions. In a random effects setup, the unit-specific heterogeneity is assumed to be drawn from a common underlying distribution. In a correlated random effects setup, this distribution could depend on observed covariates. In a fixed effects setup, the unit-specific heterogeneity is treated as fixed parameters unique to each unit without imposing any distributional assumptions, implicitly allowing for an arbitrary correlation between the heterogeneity and covariates. For logistic errors, \citet{Chamberlain1980}, among others, employs the conditional MLE to eliminate the unit-specific intercept and estimate the common parameters; for general errors, \citet{Manski1987} extends the maximum score estimator to panel data.

Since we allow for heterogeneity in the RV functions without explicitly modeling it, our approach aligns with a fixed effects framework, where unit-specific tail thickness and RV functions can flexibly depend on the additional covariates $Z$. We show that our objective function is asymptotically equivalent to a panel Logit regression on tail observations with the log extreme covariate as a regressor. This equivalence facilitates a comparable conditional MLE analysis to estimate common parameters in small-$T$ panels. For large-$T$ panels, bias correction methods, such as those developed by \citet{fernandez2018fixed} and \citet{stammann2016estimating}, can be utilized to further estimate unit-specific parameters. Additionally, we extend our approach to dynamic binary panel data models, where the conditional MLE construction is related to the work of \citet{honore2000panel}.

Second, our work contributes to the heavy tail and extreme value theory literature. For a comprehensive review, please refer to \citet{de2006extreme} and \citet{gabaix2009,Gabaix2016JEP}. The literature typically focuses on continuous outcomes with heavy tails, such as the classic \citet{hill1975} estimator, as well as estimators proposed by \citet{smith1987MLE} and \citet{gabaix2011rank}. See recent reviews by \citet{GomesGuillou2015review} and \citet{Fedotenkov2020review} for over a hundred estimators. \citet{WangTsai2009} incorporate covariates and develop a tail index regression model, and \citet{Nicolau2023} further extend this framework to time-series data under strongly mixing conditions. In contrast to the existing literature, we focus on binary outcomes and utilize Bayes' theorem to construct a novel estimator. To incorporate covariates, we build upon the tail index regression model and significantly extend the theoretical framework to handle panel data and i.n.i.d.\ random variables. 

Third, our work further contributes to the recent literature on unit-specific forecasts in panel data setups. To the best of our knowledge, this paper is the first investigation into unit-specific forecasts in panel data settings with binary outcomes and extreme covariates. In linear models, \citet{liu2020forecasting} and \citet{liu2023density} develop empirical Bayes and full Bayesian approaches for point and density forecasts, respectively. In nonlinear models, \citet{christensen2020robust} propose efficient robust forecasts for discrete outcome models, and \citet{liu2023forecasting} construct set and density forecasts for Tobit models with censored outcomes. Also see \citet{Baltagi2013} for the best linear unbiased predictor, \citet{giacomini2023robust} for robust forecasts, and \citet{qu2023comparing} for forecasting comparison. Most existing methods focus on continuous outcomes (except for \citet{christensen2020robust}) and mid-sample properties, and thus are not suitable for our purpose. Also note that panel data are particularly valuable for analyses involving extreme events, as the limited tail information makes it even more beneficial to combine information across cross-sectional units.

Finally, our work is also related to the growing interest in studying extreme events and their economic consequences, partly motivated by the 2007--2008 financial crisis and the Covid-19 pandemic. Existing literature addresses extreme events in various ways: some exclude them (e.g., \citet{schorfheide2021real}), some adapt models to accommodate both extreme and ordinary observations within a unified framework (e.g., \citet{carriero2022addressing}, and \citet{lenza2022estimate}), and some employ quantile regressions (e.g., \citet{tobias2016covar}, \citet{plagborg2020growth}, and \citet{adrian2022term}). Our work differs in several ways. First, their approaches are specific for continuous outcomes, not for binary ones. Second, our approach focuses on cross-sectional and panel data rather than time-series data, allowing us to exploit information across units.\footnote{We also consider dynamic panels under conditional stationarity, and our framework could be extended to accommodate strongly mixing conditions in the time-series dynamics as in \citet{Nicolau2023}, though this extension is beyond the scope of this paper.} Third, we incorporate a semiparametric approach that models extreme events separately from the middle of the data, acknowledging potentially distinct patterns in normal and extreme environments. While the value-at-risk and quantile regression papers also consider the second and third points, our approach is specifically designed for binary outcomes, with a different model structure and estimation procedure.

The remainder of this paper is organized as follows. Sections \ref{sec:c-s} and \ref{sec:panel} specify our methodology for cross-sectional and panel data, respectively, and derive their asymptotic properties. Section \ref{sec:extensions} extends our estimator to various contexts, such as a dynamic panel data model. Section \ref{sec:sim} conducts Monte Carlo experiments to examine the finite-sample properties of our estimators. Section \ref{sec:app} employs our panel data estimator to analyze how local housing price declines affected the riskiness of small banks during the 2007--2008 financial crisis. Finally, Section \ref{sec:conclusion} concludes. Appendix \ref{sec:proofs} provides the proofs for all propositions and theorems, and Appendix \ref{sec:app-tab-fig} contains additional tables and figures.

\section{Cross-sectional data}\label{sec:c-s}

In this section, we continue the discussion from the introduction for cross-sectional data. Section \ref{sec:c-s-simple} focuses on intuition and illustrates our main idea without additional covariates. It also compares our method with the classic threshold-crossing model. Section \ref{sec:c-s-asym} details the estimator and establishes its asymptotic properties in a general setup with additional covariates.

\subsection{Baseline models and RV functions}\label{sec:c-s-simple}
Recall that in the introduction, we presented a simple cross-sectional model with binary outcome $Y$ and continuous covariate $X$. Without loss of generality, let $X\in\mathbb{R}^{+}$, and $x\rightarrow\infty $ as $N\rightarrow
\infty$.\footnote{Directly incorporating multidimensional extreme covariates poses theoretical challenges. However, a practical solution is to replace $X$ with $v(X;\gamma)$, where $v(\cdot;\gamma)$ is a known function with unknown finite-dimensional parameters $\gamma$, and $\gamma$ can be consistently estimated with sufficiently fast convergence rate. For simplicity, we will treat $X$ as a scalar in the rest of the paper, but most of our discussions can be extended to the case of $v(X;\gamma)$.} We focus on three potential objects of interest: the conditional probability $\pi(x)$, the partial effect $\partial \pi(x) /\partial x$, and the extreme elasticity $\delta(x)$, as defined in \eqref{eq:c-s-simple-prob} and \eqref{eq:c-s-simple-elas}. 

Heavy-tailed distributions are well-documented in economic and financial data, e.g., \citet{gabaix2009,Gabaix2016JEP}, and the literature has suggested various methods to assess their presence. First, a log-log plot provides a visual assessment of tail behavior by plotting the threshold $x$ against the probability of exceeding $x$, both on a log scale, highlighting the relative prevalence of large values in the data. If tail observations align around a downward-sloping line, this suggests a heavy tail pattern with the slope being approximately $-\alpha$, where $\alpha$ is the Pareto exponent or tail index. In contrast, a vertical alignment indicates a relatively thin tail. See Figures \ref{fig:sim-exp1-loglog} and \ref{fig:app-loglog} for examples with simulated and empirical data. Second, a more rigorous evaluation involves estimating the tail index and conducting statistical tests: see for example, \citet{clauset2009power}. Finally, for a data-driven evaluation, techniques like cross-validation and pseudo out-of-sample forecasting can be employed as well, as demonstrated in our Monte Carlo Experiment 2 and empirical example.

Existing methods encounter difficulties in handling extreme values. Parametric approaches, such as Logit and probit, may incur substantial misspecification bias in the tail, while nonparametric methods, such as kernel, sieve, and spline, may be highly inefficient due to a limited number of tail observations. To overcome these challenges, we propose a semiparametric approach based on Bayes' theorem and RV functions as follows.

First, we observe extreme values in the covariate $X$ and seek to understand the behavior of $Y$ under extreme $X$, so we essentially aim to analyze their comovement in the tail. To facilitate this tail analysis, we use Bayes' theorem to ``reverse'' the conditioning set in the conditional probability \eqref{eq:c-s-simple-prob}, resulting in equation \eqref{eq:c-s-simple-bayes}.
Note that the Bayes' theorem representation is simply an alternative characterization of the data, so we are agnostic about the causal direction between $X$ and $Y$. A similar Bayes' theorem representation has been employed in the discriminant analysis (Chapter 9.2.8 in \citet{amemiya1985advanced}), where normal distributions of $X|Y=y$ leads to a quadratic Logit form of $\P(Y=1|X=x)$. \citet{klein1993efficient} also use a similar Bayes' theorem transformation in their semiparametric single-index estimator. 

Next, as $x\rightarrow \infty $, the conditional pdfs $f_{X|Y}\left( x|y\right)$ for $y\in\{0,1\}$ are dominated by their tail behaviors. If $X|Y=y$ exhibits a heavy tail, its distribution can be well approximated by a Pareto distribution \citep{smith1987MLE}. Therefore, we adopt a more general concept of RV functions that leads to a Pareto approximation in the tail.\footnote{Our method can be adapted to the generalized Pareto distribution, which also encompasses thin-tailed and bounded-support distributions. However, these distributions are often easily distinguishable from heavy-tailed distributions given empirical data. Moreover, the simplicity of our Pareto-based estimator offers a practical advantage over the more involved generalized Pareto estimator.}

We first introduce the following definitions. Let $X\in \mathbb{R}^{+}$ be a generic random variable with a heavy-tailed distribution. We say the upper tail probability (or survival function) $1-F_X(x)$ is \textit{Regularly Varying (RV)} at infinity with index $-\alpha$ for some $\alpha>0$, if for all $x>0$, as $\eta \rightarrow \infty $,
\begin{equation}
\frac{1-F_X\left( \eta x\right) }{1-F_X\left( \eta \right) }\rightarrow x^{-\alpha },
\label{eq:RV}
\end{equation}
which is denoted as $1-F_X\in RV_{-\alpha }$. Equivalently, by Karamata's characterization theorem, we can write
\begin{equation*}
1-F_X(x) =x^{-\alpha }\L(x),
%\label{eq:slow-varying}
\end{equation*}
where $\L \left( \cdot \right) $ is a \textit{slowly varying function} such
that for all $x>0$, 
$
\frac{\L \left( \eta x\right) }{\L \left( \eta \right) }\rightarrow 1,
$
as $\eta \rightarrow \infty $.

Let us highlight a few key points. First, the RV condition \eqref{eq:RV} implies that as $\x\rightarrow\infty$, for $x\ge\x$,
\begin{equation}
1-F_{X}(x)=\left(1-F_{X}(\x)\right)\left(\frac x {\x}\right)^{-\alpha}(1+o(1)).\label{eq:rv-cdf}
\end{equation}
Thus, the tail distribution of $X$ is asymptotically proportional to a Pareto distribution as the first-order approximation. Second, the parameter $\alpha $ is also referred to as the \textit{Pareto exponent}, which characterizes the tail thickness of $1-F_X\left(x \right)$. In particular, a smaller $\alpha $ indicates a heavier tail, meaning that the upper tail probability decays to zero more slowly; as $\alpha\rightarrow\infty$, the distribution becomes thin-tailed, such as logistic and normal distributions.%\footnote{While the literature has considered cases with $\alpha=0$ (slowly varying) and $\alpha<0$ (rapidly varying) (see Appendix B in \citet{de2006extreme}), we focus on $\alpha>0$ here for heavy tail distributions.} 
Third, the RV condition is relatively mild and satisfied by many commonly used heavy-tailed distributions, including the Student-$t$, $F$, and Cauchy distributions. For example, in a Student-$t$ distribution, $\alpha $ is equal to the degrees of freedom.\label{com:t-dist} A list of distributions and their corresponding values of $\alpha $ can be found in \citet{Gabaix2016JEP}. Fourth, the moments of $X$, $\E\left[ \left| X\right| ^{r}\right]$, are finite up to order $r<\alpha$. Note that we can handle cases with very heavy tails where $\alpha\in(0,1)$ and no moments exist. Finally, following from Proposition B.1.9(11) in \citet{de2006extreme}, if the pdf $f_{X}(x)$ is non-increasing for sufficiently large $x$, we have the following approximation that as $x\rightarrow\infty$,
\begin{align}
\frac{xf_{X}\left( x \right) }{1-F_{X}\left( x \right) }\rightarrow
\alpha. \label{eq:pdf}
\end{align}

Back to our binary outcome model, we assume that the distribution of 
$X$ conditional on $Y=y$ satisfies the RV condition \eqref{eq:RV}, that is, for $y\in\{0,1\}$,
\(
1-F_{X|Y}\left(\cdot|y\right)\in RV_{-\alpha^\py}.
\)
The Pareto exponent $\alpha^\py$ is indexed by $y$, allowing for different tail thickness for $y\in\{0,1\}$.
If $Y$ were continuous, estimating $\alpha^\py$ as a function of $y$ would be challenging without parametric assumptions. The binary outcome simplifies this, as $\alpha^\py$ takes only two values, $\alpha ^\po$ and $\alpha ^\pl$, which can be easily estimated using methods such as the classic \citet{hill1975} estimator in this simple model. Further estimation details are provided in Section \ref{sec:c-s-asym}. Also note that our approach is semiparametric, as we impose no parametric assumptions on the relationship between $X$ and $Y$ outside the tail region.

After estimating $\alpha^\py$, we can proceed with the conditional probability $\pi(x)$ as well as other objects of interest. 
Based on Bayes' theorem \eqref{eq:c-s-simple-bayes} and the pdf approximation \eqref{eq:pdf}, it follows that
\begin{align}
\pi(x) 
&\sim \frac{1}{1+\frac{\P\left( Y=0\right) }{\P\left(
Y=1\right) }\frac{\alpha ^\po}{\alpha ^\pl}\frac{1-F_{X|Y}\left(
x|0\right) }{1-F_{X|Y}\left( x|1\right) }}, \label{eq:c-s-simple-prob1}
\end{align}
where $f(x)\sim g(x)$ denotes $\lim_{x\rightarrow\infty} f(x)/g(x)=1$ for generic functions $f(x)$ and $g(x)$.
By the RV condition \eqref{eq:rv-cdf}, given some large $\x^\py$, the term in the denominator becomes
\begin{align}
\frac{\P\left( Y=0\right) }{\P\left(
Y=1\right) }\frac{\alpha ^\po}{\alpha ^\pl}\frac{1-F_{X|Y}\left(
x|0\right) }{1-F_{X|Y}\left( x|1\right)}
\sim\frac{\P\left( Y=0,X\ge\x^\po\right) }{\P\left(
Y=1,X\ge\x^\pl\right) }\frac{\alpha ^\po}{\alpha ^\pl}\frac{\left(\x^\po\right)^{\alpha^\po}}{\left(\x^\pl\right)^{\alpha^\pl}} x^{\alpha^\pl-\alpha^\po}.\label{eq:c-s-simple-prob2}
\end{align}
Then, $\pi(x)$ can be consistently estimated by plugging in $\hat\alpha^\py$ and the sample analog of $\frac{\P\left( Y=0,X\ge\x^\po\right) }{\P\left(
Y=1,X\ge\x^\pl\right) }\approx\frac {N^\po}{N^\pl}$, where $N^\py=\sumi \1\left\{ X_i^\py\ge \x^\py, Y_i=y\right\} $ is the number of tail observations in the subsample with $Y_i=y$.\footnote{The consistency result was established in a previous version of this paper and is available upon request.}

\begin{remark}\label{rm:c-s-Logit}
\normalfont{
Based on \eqref{eq:c-s-simple-prob1} and \eqref{eq:c-s-simple-prob2}, we define $A=\frac{\P\left( Y=0,X\ge\x^\po\right) }{\P\left(Y=1,X\ge\x^\pl\right) }\frac{\alpha ^\po}{\alpha ^\pl}\frac{\left(\x^\po\right)^{\alpha^\po}}{\left(\x^\pl\right)^{\alpha^\pl}}$, a constant not changing over $x$. Then, $\pi(x)$ is simplified to:
\begin{align}
\pi(x)\sim \frac{1}{1+A\cdot x^{\alpha ^\pl-\alpha ^\po}},\label{eq:c-s-A}
\end{align}
and we can directly estimate $A$ and $\alpha^*=\alpha ^\pl-\alpha ^\po$ via the MLE. 
However, for the cross-sectional case, we slightly prefer the estimator based on \eqref{eq:c-s-simple-prob1} and \eqref{eq:c-s-simple-prob2} due to its easier implementation in the simple setup; moreover, when incorporating additional covariates $Z$ (see Section \ref{sec:c-s-asym}), the parameter $A$ could be a complicated function of $Z$ and difficult to estimate.

Furthermore, letting $\tilde A = -\log A$ yields $\pi(x)\sim \frac{1}{1+\exp\left(-\tilde A+\alpha^*\log x\right)},$
so our objective function asymptotically resembles that of a Logit regression on tail observations, using the log extreme covariate as the regressor. This asymptotic equivalence is particularly relevant in panel data cases with unobserved unit-specific heterogeneity: see Remark \ref{rm:panel-Logit}.
}
\end{remark}

\begin{remark}\label{rm:c-s-01}
\normalfont{
From \eqref{eq:c-s-A}, we see that between the two subsamples with $Y_i=y\in\{0,1\}$, the one with the heavier tail (i.e., smaller Pareto exponent $\alpha^\py$) will ultimately dominate the conditional probability, that is, as $x\rightarrow\infty$, 
\begin{align*}
\pi(x)\sim \frac{1}{1+A\cdot x^{\alpha ^\pl-\alpha ^\po}}\rightarrow\begin{cases}
1,& \text{ if } \alpha ^\pl<\alpha ^\po,\\
0,& \text{ if } \alpha ^\pl>\alpha ^\po,\\
\frac 1 {1+A},& \text{ if } \alpha ^\pl=\alpha ^\po.
\end{cases}
\end{align*}
This Pareto approximation captures the first-order behavior, and higher-order terms, as detailed in Assumption \ref{assn:c-s-tail}, can further refine $\pi(x)$. Also note that, in finite samples, our method still provides significant improvements for moderately large $x$ with $\pi(x)$ not very close to 0 or 1, as discussed in footnote \ref{fn:moderate-large-x}.
}
\end{remark}

After estimating $\pi(x)$, the partial effect $\partial \pi(x) /\partial x$ can be obtained as a by-product through either analytical or numerical differentiation.  It may also be of interest to consider averages of the partial effects over certain covariate values, which relate to commonly used measures, such as the average partial effects (APE) and the average marginal effects (AME) in panel data models: see \citet{AbrevayaHsu2021} for a survey of various partial effects in panels, and \citet{DaveziesDHaultfoeuilleLaage2021} for recent partial identification results of average causal effects in short-$T$ panel Logit models with fixed effects. 
Proposition \ref{prop:avg-partial-effects} in the Appendix establishes the existence of the tail average of partial effects.

For extreme elasticity \eqref{eq:c-s-simple-elas}, given the Pareto approximation in the tail, the extreme elasticity is solely determined by the difference in tail indices, $\alpha^*=\alpha ^\pl-\alpha ^\po$.
\begin{proposition} [Cross-sectional data: extreme elasticity]
\label{prop:elas}Suppose we have: (a) $1-F_{X|Y}\left(\cdot|y\right)\in RV_{-\alpha^\py}$, for $y\in\{0,1\}$; 
(b) $f_{X|Y}\left( x|y\right) $ and $f_{X|Y}^{\prime
}\left( x|y\right) $ are non-increasing in $x\ge\x$, for some $\x>0$; 
and (c) $0<\P(Y=1)<1$.
Then, as $x\rightarrow \infty $, $\delta(x) \rightarrow -\left|\alpha ^*\right|.$
\end{proposition}
\begin{remark}\label{rm:elas}
\normalfont{
From Remark \ref{rm:c-s-01}, we see that $\pi(x)$ approaches either 0 or 1 as $x\rightarrow\infty$ when $\alpha ^\pl\neq\alpha ^\po$. The extreme elasticity defined in \eqref{eq:c-s-simple-elas} provides a unified characterization of both cases. Rewriting \eqref{eq:c-s-simple-elas}, we obtain that 
\begin{equation}
\delta(x)=\frac{\partial \pi (x)}{\partial x}
\frac{x}{\pi (x)} + \frac{\partial(1-\pi (x))}{\partial x}
\frac{x}{ 1-\pi (x)} =\delta_{\pi}(x)+\delta_{1-\pi}(x).\label{eq:elas-decomp}
\end{equation}
For example, when $\alpha ^\pl<\alpha ^\po$, we have $\pi(x)\rightarrow1$, $\delta_{\pi}(x)\rightarrow0$, and thus $\delta(x)\sim\delta_{1-\pi}(x)\rightarrow\alpha ^\pl-\alpha ^\po$.
Note that the extreme elasticity is inherently non-positive due to the RV structure. Moreover, as $\pi (x)(1-\pi (x))=\V(Y|X=x)$, the extreme elasticity can also be interpreted as the elasticity of risk in the tail.
}
\end{remark}

\subsubsection{Comparison with a threshold-crossing model}\label{sec:threshold-xing}
We now compare our approach with the classic threshold-crossing model, which has been extensively studied in the literature:
\begin{equation*}
Y=\1\left\{X-\varepsilon \ge  0\right\},
\end{equation*}
where $\varepsilon$ is the error term. Note that our method does not require a threshold-crossing structure. However, the following proposition reveals a direct relationship between our key parameters $\alpha ^\py$ and the Pareto exponents of $X$ and $\varepsilon$ in the threshold-crossing model.
\begin{proposition} [Threshold-crossing model]
\label{prop:threshold-xing}Suppose we have:
(a) $1-F_{X}\in RV_{-\alpha_X }$, and $1-F_{\varepsilon
}\in RV_{-\alpha_{\varepsilon} }$, for some $\alpha_X,\,\alpha_{\varepsilon}>0$;
(b) $f_{X}(x) $ is non-increasing in $x\ge\x$, for some $\x>0$;
and (c)  $\varepsilon\perp X$.
Then, for $y\in\{0,1\}$, $1-F_{X|Y}\left(\cdot|y\right)\in RV_{-\alpha^\py}$ with
\begin{align*}
\alpha ^\po =\alpha_X +\alpha_{\varepsilon}, \text{ and }
\alpha ^\pl =\alpha_X.
\end{align*}
Furthermore, $\alpha_{\varepsilon} =\alpha ^\po-\alpha ^\pl.$
\end{proposition}
There are three points worth noting. First, the threshold-crossing structure and the \textit{unconditional} RV tails provide a sufficient condition for our \textit{conditional} RV tails. The comovement between $X$ and $\varepsilon$ in the tail is captured by the comovement between $X$ and $Y$ in the tail.

Second, the difference in the conditional Pareto exponents, $\alpha ^\po-\alpha ^\pl$, equals the Pareto exponent of the unobserved error term. As the tail of the error term becomes thinner with larger $\alpha_{\varepsilon}$, this difference increases. The extreme case of a very thin-tailed error with $\alpha_{\varepsilon}\rightarrow\infty$  (e.g., close to the Logit)  corresponds to $\alpha ^\po\rightarrow\infty$, $\alpha ^\pl=\alpha_X$, and $\pi(x)\rightarrow 1$: see also the first case in Remark \ref{rm:c-s-01}.

Finally, directly estimating the threshold-crossing model with RV errors $\varepsilon$ could be challenging due to two related issues: first, $\varepsilon$ is not observable, making it difficult to determine which observations are in the tail; second, using the full data could lead to significant bias, as the middle part may deviate substantially from the Pareto distribution. In contrast, our method circumvents this issue by using cutoffs of observed $X$, and offers an asymptotic equivalent estimator that is relatively easy to implement and scalable to more complicated setups, such as those involving additional covariates and panel data.

\subsection{Estimation and asymptotic properties}\label{sec:c-s-asym}
In this subsection, we first extend the simple cross-sectional model in Section \ref{sec:c-s-simple} to incorporate additional covariates, which is more empirically relevant. We then derive the asymptotic properties of our proposed estimator. 

Let $Z$ be a $d_Z$-dimensional vector of non-extreme covariates in addition to $X$. $Z$ can be either discrete or continuous. For example, consider the context of sovereign debts and country defaults, where institutional characteristics could be included as covariates given their potential impact on default risk. Now we rewrite our conditional probability and the Bayes' theorem representation by adding $Z$ to all conditioning sets: 
\begin{align*}
\pi \left( x,z\right)
&=\P\left( Y=1|X=x,Z=z\right) \\
&=\frac{f_{X|Y,Z}\left( x|1,z\right) \P\left( Y=1|Z=z\right) }{
f_{X|Y,Z}\left( x|1,z\right) \P\left( Y=1|Z=z\right)
+f_{X|Y,Z}\left( x|0,z\right) \P\left( Y=0|Z=z\right) }.
\end{align*}
Note that here we hold $z$ constant and let $x$ tend to infinity. 
Accordingly, we assume that 
$
1-F_{X|Y,Z}\left(\cdot|y,z\right)\in RV_{-\alpha^\py\left( z\right) },
$
where the Pareto exponent $\alpha^\py\left( z\right) $ is a function of the additional covariates $z$.

While $\alpha^\py(\cdot)$ could depend on $z$ in a complicated nonlinear way, it is difficult to nonparametrically estimate it due to the data limitation that only large values of $X$ are used for tail analysis. To sidestep this issue, we assume a pseudo-linear structure: 
\begin{equation}
\alpha^\py\left( z;\theta^\py\right) =z'\theta^\py,
\label{eq:index}
\end{equation}
where $\theta^\py$ denotes the pseudo-parameters for $y\in\{0,1\}$, and $\alpha^\py\left( z\right)=\alpha^\py\left( z;\theta^\py_0\right)$ represents the tail index function evaluated at the pseudo-true parameter values. This simplification is also in line with the literature, such as \citet{WangTsai2009}.\footnote{We adopt a pseudo-linear approximation to better align with the panel data case in Section \ref{sec:panel}, rather than an exponential approximation in \citet{WangTsai2009}.} Additionally, we require that $Z_i'\theta^\py>0$ almost surely, with sufficient conditions discussed after Assumption \ref{assn:c-s-tail}.

From the RV condition, for $y\in\{0,1\}$, given some large threshold $\x^\py_N$, if the conditional pdf is non-increasing for sufficiently large $x$, it is asymptotically equivalent to a Pareto distribution,
\[f_{X|Y,Z}\left( x\left|y,z,x\ge\x^\py_N\right.\right)\sim\alpha^\py\left( z\right)\left(\frac x {\x^\py_N}\right)^{-\alpha^\py\left( z\right)-1}.\]

Let $\Xi^\py_{N,i}=\left\{ X_i^\py\ge\x^\py_N, Y_i=y\right\}$, then $N^\py=\sumi \1\left\{\Xi^\py_{N,i}\right\} $ is the total number of tail observations in the subsample with $Y_i=y$. The asymptotic Pareto distribution above leads to the following pseudo-MLE
\begin{equation}
\hat\theta^\py=\arg \max_{\theta^\py\in\Theta^\py }\sumi\left( \log Z_i'\theta^\py-Z_i'\theta^\py \log\frac{X_i}{\x^\py_N} \right)\1\left\{\Xi^\py_{N,i}\right\}, \label{eq:mle}
\end{equation}
where $\Theta^\py\subset\mathbb R^{d_Z}$ is a convex cone as discussed after Assumption \ref{assn:c-s-tail}, and the pseudo-true parameter value $\theta^\py_0\in \intr\left(\Theta^\py\right)$, the interior of $\Theta^\py$. Since the objective function is concave and the domain is convex, this MLE is a convex programming problem that can be easily solved. In the simple model without $Z$, the FOC of \eqref{eq:mle} results in 
\(
\hat\alpha^\py=\left[\frac{1}{N^\py}\sumi\left( \log X_i -\log \x^\py_N \right)\1\left\{\Xi^\py_{N,i}\right\}\right] ^{-1},
\)
which coincides with the classic \citet{hill1975} estimator, as the indicator function $\1\left\{\Xi^\py_{N,i}\right\}$ effectively selects the largest order statistics.

To establish the asymptotic properties, we impose the following assumptions.
\begin{assumption}[Cross-sectional data: model assumptions] \label{assn:c-s-model}
Suppose we have:
\begin{enumerate}[label=(\alph*)]
\item $\{Y_i,X_i,Z_i\}$ is i.i.d.
\item $0<\P\left(Y_i=1|Z_i\right) <1$ almost surely in $Z_i$.
\end{enumerate}
\end{assumption}
Condition (b) ensures non-degenerate probabilities, thus guaranteeing a non-trivial Bayes' theorem representation. 

\begin{assumption}[Cross-sectional data: tail approximation]\label{assn:c-s-tail}
    Let $\P_{Z}$ be the probability measure induced by $Z$. For $y\in\{0,1\} $, and for $\P_{Z}$-almost all $z\in supp(Z)$:
\begin{enumerate}[label=(\alph*)]
\item The conditional cdf $F_{X|Y,Z}\left( x|y,z\right) $ satisfies that
\begin{equation*}
1-F_{X|Y,Z}\left( x|y,z\right) =C^\py\left( z\right) x^{-\alpha^\py\left( z\right) }\left( 1+D^\py\left( z\right) x^{-\beta^\py\left( z\right) }+r^\py\left(x,z\right)\right),
\end{equation*}
as $x\rightarrow \infty$, where $C^\py\left( z\right) >0,\,\left|D^\py\left(z\right) \right|\le\overline D<\infty,\,\alpha^\py\left( z\right) >0$ satisfying \eqref{eq:index}, and $\beta^\py\left( z\right) \ge\underline \beta>0$, for some constants $\overline D$ and $\underline \beta$.
\item The remainder term satisfies that $\left|x^{\beta ^\py\left( z\right) }r^\py\left(x,z\right)\right|\le\bar r(x)$, where $\bar r(x)\rightarrow0$ as $x\rightarrow \infty$.
\item The conditional pdf $f_{X|Y,Z}\left( x|y,z\right)$ is non-increasing in $x\ge\x$, for some $\x>0$.
\end{enumerate}
\end{assumption}
First, the expression in condition (a) implies the RV condition and resembles equation (2.2) in \citet{WangTsai2009}. In the absence of additional covariates $Z$, it simplifies to the well-studied second-order condition: see for example, \citet{hall1982} and Chapter 2 in \citet{de2006extreme}. The first-order term $C^\py\left( z\right) x^{-\alpha^\py\left( z\right) }$ is proportional to a Pareto distribution, while the second-order term $D^\py\left( z\right) x^{-\beta ^\py\left( z\right) }$ captures deviations from the ideal Pareto form and is essential for characterizing the asymptotic distribution. 
Second, to guarantee the existence of $\theta^\py$ such that $\alpha^\py(Z_i)=Z_i'\theta^\py>0$ almost surely (i.e., the existence of the convex cone), a sufficient condition is that each element of $Z_i$ has a domain strictly above or below zero. This can be achieved through a monotonic transformation of $Z_i$, such as the exponential or probability integral transforms. For example, if $Z_i\in\mathbb R^{d_Z}_{++}$ (the set of vectors with strictly positive elements) almost surely, then we can set $\Theta^\py=\mathbb R^{d_Z}_{+}\setminus{\{0\}}$ (the set of vectors with non-negative elements, excluding the zero vector).
Third, the bounds on $D^\py\left(z\right)$ and $\beta^\py\left(z\right)$, as well as the supremum condition (b) on the remainder term, allow us to ignore higher order terms in the asymptotic analysis. 
Finally, the relatively mild condition (c) of a non-increasing conditional probability density function ensures that the tail behavior of the pdf can be inferred from the corresponding cdf, according to Proposition B.1.9(11) in \citet{de2006extreme}.

To derive the asymptotic properties, note that $N^\py$ is a random variable, so we need to carefully account for its randomness when applying the LLN and CLT. We further define 
\begin{align}\xi^\py_N = \E\left[C^\py(Z_i)\left(\x^\py_N\right)^{-\alpha^\py(Z_i)}\right]\P(Y_i=y),\label{eq:c-s-xi-def}\end{align}
a non-random sequence representing the asymptotic proportion of tail observations for $Y_i=y$, and utilize it to characterize the asymptotic behavior of our estimator.
\begin{assumption}[Cross-sectional data: estimation]\label{assn:c-s-est} For $y\in\{0,1\}$, suppose that:
\begin{enumerate}[label=(\alph*)]
\item $N\xi^\py_N\rightarrow\infty$, as $N\rightarrow \infty $.
\item $\sqrt{\frac{N}{\xi_N}} \E\left[ Z_i \frac{\beta^\py(Z_i)C^\py\left( Z_i\right) D^\py\left( Z_i\right)}{\alpha^\py(Z_i)\left(\alpha^\py(Z_i)+\beta^\py(Z_i)\right)} \left(\x^\py_N\right)^{-\alpha ^\py\left(Z_i\right) -\beta ^\py\left( Z_i\right) } \right] \rightarrow 0$, as $N\rightarrow \infty $.
\item $H_{N0}^\py =\E\left[\left.\frac {Z_iZ_i'}{\left(\alpha^\py(Z_i)\right)^2}\right|\Xi^\py_{N,i}\right]$ is finite and positive definite, for sufficiently large $N$.
\end{enumerate}
\end{assumption}
Conditions (a) and (b) jointly impose upper and lower bounds on the rate at which the tuning parameter $\x^\py_N$ tends to infinity. It implies that $N^\py$ goes to infinity at a slower rate than $N$. Also note that in condition (b), we select a larger $\x^\py_N$ to eliminate the asymptotic bias, albeit at the expense of a slower convergence rate. This is close in spirit to choosing an undersmoothing bandwidth in kernel regressions. Condition (c) is a mild regularity condition that ensures the invertibility of the Hessian matrix.  

The following theorem establishes the asymptotic result, building on \citet{WangTsai2009}.
\begin{theorem}[Cross-sectional data: parameter estimation]
\label{thm:c-s} Suppose $\theta^\py_0\in \intr\left(\Theta^\py\right)$, where $\Theta^\py\subset\mathbb R^{d_Z}$ is a convex cone. 
Suppose Assumptions \ref{assn:c-s-model}--\ref{assn:c-s-est} hold. Then, for $y\in\{0,1\}$, we have that
\begin{equation*}
\sqrt{N\xi^\py_N}\left(H_{N0}^\py\right)^{1/2}\left( \hat\theta^\py-\theta^\py_0\right) \overset{d}{\rightarrow }\mathcal{N}\left( 0,\I_{d_Z}\right), 
\end{equation*}
as $N\rightarrow \infty $.
In addition, $\hat\theta^\pl$ and $\hat\theta^\po$ are
asymptotically independent.
\end{theorem}
The convergence rate is slower than the standard parametric $\sqrt N$-rate, since the tail estimator only accounts for observations in the tail and the number of tail observations $N^\py=N\xi^\py_N\left(1+o_p(1)\right)$ increases at a slower rate than the total number of observations $N$.
The convergence rate depends on the magnitude of the second-order term: larger $\beta^\py(z)$ allows for a smaller cutoff $\x^\py_N$ according to Assumption \ref{assn:c-s-est}(b), and thus a larger $\xi^\py_N$ based on equation \eqref{eq:c-s-xi-def}, leading to a faster convergence rate.  Especially, when $\beta^\py(z)\rightarrow\infty$, the distribution approaches the ideal Pareto case, and the convergence rate is close to the $\sqrt N$-rate.

For implementation, we adopt a common practice to select the threshold $\x^\py_N$: taking empirical quantiles (e.g., 90th and 95th percentiles) as potential thresholds, and checking them through graphic diagnostics like log-log plots. While there are various threshold estimation methods in the literature,\footnote{For example, \citet{guillou2001diagnostic} choose the threshold by minimizing the asymptotic MSE, \citet{clauset2009power} select the threshold by maximizing the marginal likelihood or minimizing the distance between the power-law distribution and empirical data, and \citet{WangTsai2009} further consider a discrepancy measure for a tail index regression model with covariates.} our simulations and empirical analysis suggest that parameter estimates are relatively robust within a range of threshold values. Intuitively, in a log-log plot, such as Figure \ref{fig:sim-exp1-loglog}, the threshold $\x^\py_N$ acts as the starting point for slope estimation, and small variations in the threshold within the downward-sloping region would only minimally affect the estimated slope. 

Given $\hat\theta^\py$, we can obtain the conditional probability $\pi \left( x,z\right)$ similar to \eqref{eq:c-s-simple-prob1} and \eqref{eq:c-s-simple-prob2} in the simple model.
It is more practical to estimate $\P\left( Y=y,X\ge\x^\py_N|Z=z\right)$ parametrically, due to limited tail data. 
Analogous to Proposition \ref{prop:elas}, if we further assume that $f_{X|Y,Z}'\left( x|y,z\right) $ is non-increasing for sufficiently large $x$, then as $x\rightarrow \infty $, the extreme elasticity is now
\begin{equation*}
\delta(x,z) = \frac{\left(\partial \pi (x,z)(1-\pi (x,z))\right)}{
\partial x}\frac{x}{\pi (x,z)\left( 1-\pi (x,z)\right) }\rightarrow-\left|z'\left(\theta^\pl-\theta^\po\right)\right|.
\end{equation*}
A consistent estimator is obtained by substituting $\hat\theta^\py$ in place of $\theta^\py$.

\section{Panel data}\label{sec:panel}
Our method is particularly useful in panel data analysis, specifically in addressing unobserved individual heterogeneity. For instance, in the context of country defaults, different countries could have various cultural and historical backgrounds that might not be fully captured by observed data. In the analysis of extreme events, this unobserved heterogeneity could manifest as unobserved unit-specific tail thickness and RV functions.\footnote{For panel data with observed heterogeneity only, a pooled estimator with observed covariates can be employed, similar to the cross-sectional case in Section \ref{sec:c-s-asym}.}

This section focuses on panel data with large $N$ and small $T$. To build intuition, we begin in Section \ref{sec:panel-simple} by considering a simple case without additional covariates $Z_i$. We then incorporate these covariates in Section \ref{sec:panel-asym} and derive the asymptotic for the common parameters and extreme elasticity. Later on, we will extend our discussions to models with large $T$ in Section \ref{sec:panel-largeT}, time-varying additional covariates in Section \ref{sec:time-var-z}, and dynamic panel data in Section \ref{sec:dyn-panel}.

\subsection{Baseline model and conditional MLE}\label{sec:panel-simple}
Suppose we observe a panel dataset $\{Y_{it},X_{it}\}$ for $i=1,...,N$ and $t=1,...,T$. For illustrative purposes, let $T=2$ (though our method is applicable to any $T \ge 2$), and $X_{it}$ be a scalar extreme covariate.

As $X_{it}$ approaches infinity, unobserved individual heterogeneity could reflect in unit-specific tail thickness and RV functions. Assume the unit-specific tail indices take the following additive form
\begin{equation}
\tilde\alpha_i^\py=\alpha^\py+\lambda_i, \label{eq:alpha_hetero}
\end{equation}
where $\alpha^\py$ represents the common component depending on $y\in\{0,1\}$, and $ \lambda_i$ denotes the unit-specific component. For example, the tail indices of stock market returns exhibit heterogeneity across regions \citep{Jondeau2003}, as do Covid-19 cases and deaths across countries \citep{Einmahl2023}. The additive form helps cancel out the unobserved unit-specific tail thickness, as will soon be demonstrated. 

Let $\P_i$ be the probability measure given unit-specific quantities. Specifically, $\P_i(\cdot)=\P\left(\cdot\;;\;\lambda_i,\left\{\L^\py_i\right\}\right)$, where $\lambda_i$ is the unit-specific tail thickness as defined above, and $\L^\py_i(\cdot)$ is the unit-specific slowly varying function as specified below. Any functions subscripted with $i$ are implicitly conditioned on these unit-specific quantities. Note that we are working within a fixed effects framework where $\left\{\lambda_i,\left\{\L^\py_i\right\}\right\}$ are considered as fixed for each unit $i$ and can be arbitrarily correlated with the covariates of $i$, making the panel data analysis richer and more challenging than the cross-sectional case.

Based on a simplified version of Assumption \ref{assn:panel-model} in Section \ref{sec:panel-asym} on the model setup, we assume that $\{Y_{it},X_{it}\}$ are independent across units and stationary across time. The stationarity condition enables us to eliminate unit-specific tail effects, and the conditional independence condition ensures that for $y_1,y_2\in \{0,1\}$,
\begin{align}
\P_i\left( Y_{i1}=y_{1},Y_{i2}=y_2|X_{i1},X_{i2}\right) 
=\P_i\left( Y_{i1}=y_1|X_{i1}\right) \cdot \P_i\left(
Y_{i2}=y_2|X_{i2}\right). \label{eq:model}
\end{align}

We rewrite the conditional probability of unit $i$ using Bayes' theorem similar to the cross-sectional case. 
\begin{align*}
\pi _i(x) &= \P_i\left(
Y_{it}=1|X_{it}=x\right) \\
&=\frac{f_{i,X_{it}|Y_{it}}(x|1) \P_i\left(
Y_{it}=1\right) }{f_{i,X_{it}|Y_{it}}(x|1) \P_i
\left( Y_{it}=1\right) +f_{i,X_{it}|Y_{it}}(x|0) 
\P_i\left( Y_{it}=0\right) }.
\end{align*}
Here the subscript $i$ in $\pi _i(x)$ indicates potential heterogeneity across $i$, and the absence of a $t$ index reflects the stationarity condition. 

To proceed, we again apply the Pareto approximation at the tail, assuming that 
$1-F_{i,X_{it}|Y_{it}}(\cdot|y)\in RV_{-\tilde\alpha_i^\py},$
or equivalently,
$1-F_{i,X_{it}|Y_{it}}(x|y) =x^{-\tilde\alpha_i^\py}
\L_i^\py(x),$
for some slowly varying function $\L_i^\py(x) $.
Given the stationarity condition, $\L_i^\py(x) $ does not change over time. From the RV condition and the pdf approximation in \eqref{eq:pdf}, we have that
\begin{align*}
\pi_i(x) \sim \frac{1}{1+\frac{\P_i\left(
Y_{it}=0\right) }{\P_i\left( Y_{it}=1\right) }\frac{\tilde{
\alpha}_i^\po}{\tilde\alpha_i^\pl}\frac{
1-F_{i,X_{it}|Y_{it}}(x|0) }{1-F_{i,X_{it}|Y_{it}}(x|1)}} 
=\frac{1}{1+\frac{\P_i\left(
Y_{it}=0\right) }{\P_i\left( Y_{it}=1\right) }\frac{\tilde{
\alpha}_i^\po}{\tilde\alpha_i^\pl}\frac{\L 
_i^\po(x) }{\L _i^\pl(x) }\frac{
x^{-\tilde\alpha_i^\po}}{x^{-\tilde\alpha_i^\pl}}}. 
\end{align*}

Let us look into each term one by one. First, note that $\frac{x^{-\tilde\alpha_i^\po}}{x^{-\tilde\alpha_i^\pl}}=x^{\alpha ^\pl-\alpha ^\po}$, since $\lambda_i$ enters additively into the tail indices \eqref{eq:alpha_hetero} and is thus canceled out in the ratio. Second, we have that $\frac{\L _i^\po(x) }{\L _i^\pl(x) }\rightarrow \mathcal L^*_i$, which does not depend on $x$ asymptotically due to the slowly varying nature of $\L_i^\py(x)$. Third, the term $\frac{\P_i\left(Y_{it}=0\right) }{\P_i\left( Y_{it}=1\right) }\frac{\tilde\alpha_i^\po}{\tilde\alpha_i^\pl}$ could depend on $\left\{\lambda_i,\left\{\L^\py_i\right\}\right\}$ in a complicated, possibly nonlinear manner, as they enter into the conditioning sets. Combining all three terms, let $A_i=\frac{\P_i\left(Y_{it}=0\right) }{\P_i\left( Y_{it}=1\right) }\frac{\tilde\alpha_i^\po}{\tilde\alpha_i^\pl}\mathcal L^*_i,$ and the conditional probability admits the following asymptotic approximation
\begin{align}
\pi _i(x) &\sim 
\frac{1}{1+A_i \cdot  x ^{\alpha^\pl-\alpha ^\po}}.\label{eq:panel-prob}
\end{align}

Without further assumptions, $A_i$ cannot be consistently estimated in small-$T$ panels due to the incidental parameter problem. However, we can eliminate $A_i$ by conditioning on the event $Y_{i1} + Y_{i2} = 1$. Following from the conditional independence \eqref{eq:model}, we obtain that as $x_{1},x_2\rightarrow \infty $,
\begin{align}
&\P_i\left( Y_{i1}=1|Y_{i1}+Y_{i2}=1,X_{i1}=x_{1},X_{i2}=x_2\right)
\label{eq:panel-cond} \\
&=\frac{1}{1+\frac{\P_i\left( Y_{i1}=0|X_{i1}=x_{1}\right) 
\P_i\left( Y_{i2}=1|X_{i2}=x_2\right) }{\P_i\left(
Y_{i1}=1|X_{i1}=x_{1}\right) \P_i\left(
Y_{i2}=0|X_{i2}=x_2\right) }} =\frac1{1+\frac{\pi_i\left( x_{2}\right) (1-\pi _i\left( x_{1}\right)) }{\pi_i\left( x_1\right)(1-\pi_i\left( x_{2}\right)) } } \sim\frac{1}{1+\left( \frac{x_{1}}{x_2}\right) ^{\alpha^\pl - \alpha^\po}}.
\notag
\end{align}
The key idea here parallels the panel data Logit model, which we will further compare in Remark \ref{rm:panel-Logit} below. 

Let $\alpha^*=\alpha^\pl - \alpha^\po$,
then the extreme elasticity is given by \[\delta_i(x)=\frac{\partial \left(\pi_i(x)\left(1-\pi_i(x)\right)\right) }{\partial x}\frac{x}{\pi_i(x)\left(1-\pi_i(x)\right)}\rightarrow -\left|\alpha^*\right|,\] which implies that, under the RV condition together with the additive form \eqref{eq:alpha_hetero}, all units share the same extreme elasticity, despite having unit-specific tail thickness and RV functions.
We can estimate $\alpha^*$ via the conditional MLE on tail observations
\begin{align*}
\hat\alpha^* &=\arg \max_{\alpha^*}\sumi\left[ \alpha^*(1-Y_{i1})\log\frac{X_{i1}}{X_{i2}}-\log\left(1+\left(\frac{X_{i1}}{X_{i2}}\right)^{\alpha^*}\right) \right]  \1\{\Xi_{N,i}\},
\end{align*}
where event $\Xi_{N,i}=\left\{Y_{i1}+Y_{i2}=1, X_{it}\ge  \x
_{N}\text{, for }t=1,2\right\},$ and the tail threshold $\x_{N}\rightarrow\infty$ as $N\rightarrow\infty$. While in principle the tail threshold could be specific to both the period $t$ and the value of $y$, we set it to a common $\x_N$ for simplicity in both theory and implementation.

\begin{remark}\label{rm:panel-Logit}
\normalfont{
We now compare our setup with the classic panel Logit model. The following discussion is also in line with Remark \ref{rm:c-s-Logit} for the cross-sectional case. For instance, equation (1) in \citet{honore2000panel}, with a slight change in notation, can be represented as:
\begin{equation}
\P\left( Y_{it}=1|X_{it}=x;C_i\right) =\frac{1}{1+\exp \left(- (x\beta +C_i)\right) },\label{eq:logit}
\end{equation}
and accordingly
\begin{equation}
\P\left( Y_{i1}=1|X_{i1}=x_1,X_{i2}=x_2,Y_{i1}+Y_{i2}=1\right) =\frac{1}{
1+\exp\left(-(x_1-x_2)\beta\right)}. \label{eq:logit-cond}
\end{equation}
Comparing the expressions in \eqref{eq:panel-prob} and \eqref{eq:panel-cond} with those in \eqref{eq:logit} and \eqref{eq:logit-cond} leads to the following useful observations.

First, the classic Logit model is characterized by the exponential function of $x$, while ours is by a power function of $x$. However, by applying a log transformation to $x$, our method essentially adopts the same conditional likelihood function as the panel Logit model, but on tail observations. Our parameter of interest $\alpha ^*$ corresponds to $-\beta $ in the panel Logit model, and as shown in Section \ref{sec:panel-largeT} for large $T$, our $-\log A_i$ corresponds to their $C_i$ as well. 

Our derivation thus provides a theoretical justification for, and clarifies the underlying assumptions behind, the intuitive practice of taking the log of $X_{it}$ when handling extreme observations. In practice, we can estimate $\alpha^*$ by conducting the conditional MLE for panel Logit models on tail observations, using $\log X_{it}$ as a regressor, where the negative of its coefficient provides the estimate for $\alpha^*$. 

Second, as discussed in Section \ref{sec:c-s-simple}, our method is a semiparametric approach, focusing on tail behavior.  
Especially, our approach avoids imposing parametric assumptions on the entire error distribution, and thus is more robust to potential misspecification. Of course, this robustness comes at the cost of reduced sample size, as we only utilize the large $X_{it}$ observations.
}
\end{remark}

\subsection{Estimation and asymptotics}\label{sec:panel-asym}
In this subsection, we derive the asymptotic normality of our estimator in the panel data setup. We now extend the baseline model in Section \ref{sec:panel-simple} and include additional covariates $Z_i$ to capture potential observed heterogeneity. The unit-specific tail thickness becomes 
\begin{equation}
\tilde\alpha_i^\py\left( z;\theta^\py\right) =z'\theta^\py+\lambda_i>0, \label{eq:heter_alpha_panel}
\end{equation}
and $\tilde\alpha_i^\py\left( z\right)=\tilde\alpha_i^\py\left( z;\theta^\py_0\right)$ denotes the tail index functions evaluated at the pseudo-true parameter values. The corresponding RV condition is $1-F_{i,X_{it}|Y_{it},Z_i}(\cdot|y,z)\in RV_{-\tilde\alpha_i^\py\left(z\right) }.$
We focus on the time-invariant $Z_i$ in this subsection, and the case with time-varying $Z_{it}$ is discussed in Section \ref{sec:time-var-z}. 

Let $\C_i$ denote the unit-specific collection of functions $\left\{\beta^\py_i(\cdot),\,C^\py_i(\cdot),\,D^\py_i(\cdot)\right\}$ in the unit-specific distribution $F_{i,X_{it}|Y_{it},Z_i}$: see Assumption \ref{assn:panel-tail} for more details. As $\left\{\lambda_i,\C_i\right\}$ are fixed for each $i$, we denote $\P_i$ as the probability measure given $\left\{\lambda_i,\C_i\right\}$, that is, $\P_i(\cdot) = \P(\cdot\;;\;\lambda_i,\C_i)$. Similarly, we denote $\E_i[\cdot] = \E[\cdot\;;\;\lambda_i,\C_i]$. 

First, we adopt the following assumptions on the model setup. 
\begin{assumption}[Panel data: model assumptions]\label{assn:panel-model}
Suppose we have:
\begin{enumerate}[label=(\alph*)]
\item $\{Y_{i1},Y_{i2},X_{i1},X_{i2},Z_i\}$ are independent across $i$.
\item For each $i$, $\{Y_{it},X_{it}\}$ are stationary across $t=1,2$.
\item For each $i$, given $\left\{\lambda_i,\C_i\right\}$, $Y_{i1}\perp Y_{i2}|X_{i1},X_{i2},Z_i$.
\item For each $i$, $\P_i\left( Y_{it}=1|X_{i1},X_{i2},Z_i\right) =\P_i
\left( Y_{it}=1|X_{it},Z_i\right) $ for $t = 1,2$.
\item For each $i$, $0<\P_i(Y_{it}=1|Z_i)<1$ almost surely in $Z_i$.
\end{enumerate}
\end{assumption}
In condition (a), the covariates and outcomes are i.n.i.d.\ across units due to the unobserved unit-specific heterogeneity as fixed effects. Condition (b) ensures stationarity that helps eliminate the unit-specific tail effects. Conditions (c) and (d) imply conditional independence, thus for $y_1,y_2\in \{0,1\}$,
\begin{align}
\P_i\left( Y_{i1}=y_1,Y_{i2}=y_2|X_{i1},X_{i2},Z_i\right) =\P_i\left( Y_{i1}=y_1|X_{i1},Z_i\right) \cdot \P_i
\left( Y_{i2}=y_2|X_{i2},Z_i\right). \label{eq:modelz}
\end{align}
Condition (e) guarantees non-degenerate probabilities for all units.

Second, we assume the following tail conditions. 
\begin{assumption}[Panel data: tail approximation]\label{assn:panel-tail} For $i=1,\cdots,N$, for $y\in\{0,1\}$, and for $\P_{Z_i}$-almost all $z\in supp(Z_i)$:
\begin{enumerate}[label=(\alph*)]
\item The conditional cdf $F_{i,X_{it}|Y_{it},Z_i}\left(x|y,z\right) $ satisfies that
\begin{equation*}
1-F_{i,X_{it}|Y_{it},Z_i}\left( x|y,z\right) = C^\py_i\left( z\right)
x^{-\tilde\alpha_i^\py\left( z\right) }\left( 1+D_i^\py\left( z\right)
x^{-\beta^\py_i\left( z\right)}+r^\py_i\left(x,z\right) \right),
\end{equation*}
as $x\rightarrow \infty$, where $C^\py_i\left( z\right) >0$, $\left|D^\py_i\left(z\right) \right|\le\overline D<\infty$, $\tilde\alpha_i^\py\left( z\right)\ge\underline\alpha>0$, and $0<\underline\beta\le\beta^\py_i(z)\le\bar\beta<\infty$, for some constants $\overline D$, $\underline\alpha$, $\underline \beta$, and $\bar\beta$. 
\item The remainder term satisfies that $\left|x^{\beta ^\py_i\left( z\right)+k }\frac{\partial^k r^\py_i\left(x,z\right)}{\partial x^k}\right|\le\bar r(x)$, where $\bar r(x)\rightarrow0$ as $x\rightarrow \infty$, for $k=0,1$.
\item The conditional pdf $f_{i,X_{it}|Y_{it},Z_i}\left( x|y,z\right)$
is non-increasing in $x\ge\x$, for some $\x>0$.
\end{enumerate}
\end{assumption}
Please refer to the discussion after Assumption \ref{assn:c-s-tail} for a detailed explanation. Compared to Assumption \ref{assn:c-s-tail}, we now have all these functions as unit-specific, indexed by $i$, and condition (b) further bounds the derivative of the remainder term for the i.n.i.d.\ case. Especially, we can accommodate a unit-specific scale parameter, given by $\left[C^\py_i\left( Z_i\right)\right]^{1/\tilde\alpha^\py(Z_i)}$, in the Pareto tail approximation. These unit-specific functions will be absorbed into $A_i$ in \eqref{eq:Ai} below, and then differenced out for small $T$ or estimated for large $T$.

Then, $\L^\py_i(x,z)=C^\py_i\left( z\right)\left( 1+D_i^\py\left( z\right)
x^{-\beta^\py_i\left( z\right) }+r^\py_i\left(x,z\right) \right)$ is the corresponding slowly varying function,
and $\frac{\L^\po_i(x,z)}{\L^\pl_i(x,z)}\rightarrow\L^*_i(z)$.
Define $\theta^*= \theta^\pl-\theta^\po$, and $A_i=\frac{\P_i\left(Y_{it}=0|Z_i\right) }{\P_i\left( Y_{it}=1|Z_i\right) }\frac{\tilde\alpha_i^\po}{\tilde\alpha_i^\pl}\mathcal L^*_i(Z_i)$. As $x\rightarrow\infty$, 
\begin{align}
\pi_i(x,z)&=\P_i\left(Y_{it}=1|X_{it}=x,Z_i\right) 
\sim\frac{1}{1+A_i\cdot x ^{Z_i'\theta^*}}. \label{eq:Ai}
\end{align}
Note that as $\P_i(\cdot) = \P(\cdot\;;\;\lambda_i,\C_i)$, $A_i$ can be viewed as the fixed effects that can potentially depend on the observed heterogeneity $Z_i$ and unobserved heterogeneity $\left\{\lambda_i,\C_i\right\}$ in an arbitrary way.

For small $T$, we again eliminate $A_i$ by conditioning on $Y_{i1}+Y_{i2}=1$ and construct the conditional MLE as follows: 
\begin{align}
\hat\theta^* &=\arg \max_{\theta^*\in\Theta^*}
\sumi\left[ Z_i'\theta^* (1-Y_{i1})\log\frac{X_{i1}}{X_{i2}}-\log\left(1+\left(\frac{X_{i1}}{X_{i2}}\right)^{Z_i'\theta^* }\right) \right]
\1\{\Xi_{N,i}\}, \label{eq:mle panel} 
\end{align}
where $\Theta^*\subset\mathbb R^{d_Z}$ is a convex cone and the pseudo-true parameter value $\theta^*_0\in \intr\left(\Theta^*\right)$. Again, this can be implemented via the conditional MLE for panel Logit models, applied to tail observations and using $Z_i\log X_{it}$ as regressors.

Finally, analogous to the cross-sectional case, the number of units contributing to the conditional likelihood, $N_{\Xi}=\sumi \1 \{\Xi_i\} $, is a random variable, and we define the following non-random sequence to capture the asymptotic proportion of these units: 
\begin{align}
\xi_N=\frac 1 N \sumi \xi_{N,i},\text{ where }\xi_{N,i}=2\E_i\left[\prod_{y\in\{0,1\}}\x_N^{-\tilde\alpha^\py_i\left( Z_i\right) }C^\py_i\left( Z_i\right) \P_i(Y_{it}=y|Z_i) \right].\label{eq:panel-xi-def}
\end{align}
Also define the Hessian terms $H_{N0,i}=\E_i\left[\left.
\frac{\left( \frac{X_{i1}}{X_{i2}}\right) ^{Z_i'\theta^*_0}}{\left( 1+\left( \frac{X_{i1}}{X_{i2}}\right) ^{Z_i^{\prime
}\theta^*_0}\right) ^2}\left( \log \frac{X_{i1}}{X_{i2}} \right) ^2Z_iZ_i'\right|\Xi_{N,i}\right]$ and \(H_{N0} = {\sumi H_{N0,i}\xi_{N,i}}/{\sumi\xi_{N,i}}.\) Now we assume the following conditions for estimation.
\begin{assumption}[Panel data: estimation]\label{assn:panel-est} For $y\in\{0,1\}$, suppose we have: 
\begin{enumerate}[label=(\alph*)] 
\item $N\xi_N \rightarrow\infty$ and $N\xi_N=o\left(\left(\underline x_N\right)^{\underline\beta\frac{2+\kappa}{1+\kappa
    }}\right)$ as $N\rightarrow\infty$, for some $\kappa\ge0$.
\end{enumerate}
\noindent For sufficiently large $N$ and all $i$:
\begin{enumerate}[label=(\alph*),resume] 
\item $\lambda_{\min}\left(H_{N0,i}\right)\ge\underline\sigma^2_N$, where $\underline{\sigma}^2_N\left(N\xi_N\right)^{\kappa/\left(2+\kappa\right)}\to\infty$ as $N\rightarrow\infty$.
\item $\E_i\left[\left.\left| \log\frac{X_{i1}}{X_{i2}}\right|^4\|Z_i\|^4\right|\Xi_{N,i}\right] <M$ for some $M<\infty$.
\end{enumerate}
\end{assumption}
This assumption is in line with Assumption \ref{assn:c-s-est} for the cross-sectional case and accommodates i.n.i.d.\ observations across $i$. Condition (a) ensures that the number of observations contributing to the likelihood increases with the cross-sectional sample size, while employing undersmoothing to eliminate asymptotic bias. Conditions (b) and (c) guarantee the validity of the LLN and Lindeberg-Feller CLT for i.n.i.d.\ observations.

Under these assumptions, we establish the asymptotic normality of our estimator for panel data models with i.n.i.d.\ random variables. As the number of tail observations $N_{\Xi}=N\xi_N\left(1+o_p(1)\right)$ increases at a slower rate than $N$, the convergence rate is slower than the parametric rate.
\begin{theorem}[Panel data: common parameters]
\label{thm:panel} Suppose that $\theta^*_0\in \intr\left(\Theta^*\right)$, where $\Theta^*\subset\mathbb R^{d_Z}$ is a convex cone. 
Suppose Assumptions \ref{assn:panel-model}--\ref{assn:panel-est} hold. Then, as $N\rightarrow\infty$,
\begin{equation*}
\sqrt{N\xi_N}\left(H_{N0}\right)^{1/2}\left( \hat\theta^*-\theta^*_0\right) \overset{d}{
\rightarrow }\mathcal{N}\left( 0,\I_{d_Z}\right).
\end{equation*} 
\end{theorem}

Moreover, the unit-specific extreme elasticity depends only on observed heterogeneity, 
\begin{equation*}
\delta_i \left( x,z\right) = \frac{\partial \left(\pi_i (x,z)(1-\pi_i (x,z))\right)}{
\partial x}\frac{x}{\pi_i (x,z)\left( 1-\pi_i (x,z)\right) } \rightarrow-\left|z'\theta^*\right|,
\end{equation*}
provided that $f_{i,X_{it}|Y_{it},Z_i}'\left( x|y,z\right)$ is non-increasing for sufficiently large $x$.
In practice, we can consistently estimate this elasticity by substituting $\hat\theta^*$ for $\theta^*$.

\section{Extensions}\label{sec:extensions}

\subsection{Panel data with large $T$}\label{sec:panel-largeT}

Our previous panel data analysis has focused on small-$T$ panels. In such cases, unit-specific parameters cannot be consistently estimated, and hence we eliminate them by conditioning on the sum of $Y_{it}$ over time, which is a sufficient statistic for the individual parameters. Now, with a large $T$, we are able to estimate these unit-specific parameters. 

Based on \eqref{eq:Ai} in Section \ref{sec:panel-asym}, let $\tilde A_i = -\log A_i$, and we have that, as $x\rightarrow\infty$, 
\begin{align}
\P_i\left(Y_{it}=1|X_{it}=x,Z_i\right) 
\sim\frac{1}{1+A_i\cdot x ^{Z_i'\theta^*}} 
=\frac{1}{1+\exp \left( -\tilde A_i+\log x\cdot Z_i'\theta^* \right)}.\label{eq:longT}
\end{align}
Given this asymptotic equivalence, we can resort to panel Logit estimators for large $N$ and large $T$. For example, as shown in Example 2 in \citet{fernandez2018fixed}, the
MLE estimates of $\{\theta^*,\{\tilde A_i\} \}$ are jointly consistent and $\hat\theta^*-\theta^*_0$ is asymptotically distributed as $\mathcal{N}\left( \frac{B}{N},\frac{H^{-1}}{NT}\right),$
as $N,T\rightarrow \infty $. $B$ denotes the asymptotic bias
due to the incidental parameter, and analytical or jackknife methods can be applied for bias correction. See also \citet{stammann2016estimating}. Once $\{\theta^*,\{\tilde A_i\} \}$ are consistently estimated, we can proceed to estimate the APE and provide unit-specific forecasts of $\P_i\left(Y_{i,T+1}=1|X_{i,T+1}=x,Z_i\right)$ based on \eqref{eq:longT}.\footnote{Introducing time-varying thickness $\theta^\py_t$ could be possible with additional parametric structure, such as combining our estimator with the autoregressive conditional density models in \citet{hansen1994autoregressive}. However, this changes $A_i$ to time-varying $A_{it}=\frac{\P_i\left(Y_{it}=0|Z_{i}\right) }{\P_i\left( Y_{it}=1|Z_{i}\right) }\frac{\tilde\alpha_{it}^\po(Z_{i})}{\tilde\alpha_{it}^\pl(Z_{i})}{\L_i^*(Z_{i})}$, which can no longer be differenced out even if $\left\{\lambda_i,\C_i\right\}$ does not change over time. To address this, we need to impose further parametric assumptions on $A_{it}$ or employ a local likelihood estimator as in Section \ref{sec:time-var-z}.}

\subsection{Panels with time-varying covariates $Z_{it}$} \label{sec:time-var-z}
In empirical studies, some additional covariates $Z$ may vary over time. For instance, in the context of country defaults, factors such as global business cycles could play a significant role. Then, the unit-specific tail indices in \eqref{eq:heter_alpha_panel} becomes $\tilde\alpha_i^\py\left( z_{it}\right) =z_{it}'\theta^\py+\lambda_i$, and $A_i$ in \eqref{eq:Ai} becomes $A_{it}=\frac{\P_i\left(Y_{it}=0|Z_{it}\right) }{\P_i\left( Y_{it}=1|Z_{it}\right) }\frac{\tilde\alpha_i^\po(Z_{it})}{\tilde\alpha_i^\pl(Z_{it})}{\L _i^*(Z_{it})} $, so we cannot difference out $A_{it}$ as in \eqref{eq:Ai}. 

However, under regularity conditions, $\theta^*$ can be estimated using a local (conditional) likelihood estimator \citep{tibshirani1987local, honore2000panel}. Intuitively, it introduces an additional kernel weight term that controls the distances across $Z_{it}$. Assuming $T=2$ for simplicity, the estimator for $\theta^*$ is given by
\[
\hat\theta^*=\arg \max_{\theta^*\in\Theta^*}\sum_{i=1}^{N}k^{(d_Z)}_h\left({Z_{i1},Z_{i2}}\right)\left\{(1-Y_{i1})\log\frac{X_{i1}
}{X_{i2}}\cdot Z_{i1}'\theta^*+\log\left(1+\left( \frac{
X_{i1}}{X_{i2}}\right) ^{Z_{i1}'\theta^*}\right)\right\}\1\{\Xi_{N,i}\}
\]
where $k^{(d_Z)}_h({z_{i1},z_{i2}})=\frac 1 {h^{d_Z}} \prod_{d=1}^{d_Z}k\left(\frac{z_{i1,d}-z_{i2,d}} h\right)$ is a multidimensional kernel with $k(\cdot)$ being the kernel function and $h$ being the bandwidth.

\subsection{Dynamic panel data model}\label{sec:dyn-panel}
In this subsection, we extend our analysis to dynamic panel data, incorporating predetermined variables such as lagged outcomes to capture potential persistence. 

For small $T$, w.l.o.g., we assume the first order Markov property: for each $i$,
\[X_{it},Y_{it} \;|\;X_{i,1:t-1},Y_{i,1:t-1},Z_{i}=X_{it},Y_{it} \;|\;Y_{i,t-1},Z_i,\] 
given $\left\{\lambda_i,\C_i\right\}$, where $\C_i$ comprises unit-specific functions in $F_{i,X_{it}|Y_{it},Y_{i,t-1},Z_i}$, similar to the tail approximation in Assumption \ref{assn:c-s-tail}. %Likewise, $\P_i$ indicates the corresponding probability measure given $\left\{\lambda_i,\C_i\right\}$. 
The first order Markov property implies the stationarity of the conditional joint distribution $Y_{it},Y_{i,t-1},X_{it}\;|\;Z_i$. 

Recall that previously we applied Bayes' theorem by partitioning the data into two subsets based on $Y_{it}=0$ and 1. With dynamics, we now further partition the data according to transition dynamics. In the first order Markov case, there are four transition patterns: $0\rightarrow0$, $0\rightarrow1$, $1\rightarrow0$, and $1\rightarrow1$. Therefore, we partition the data into four subsets by $\left( Y_{it},Y_{i,t-1}\right) =(y,y_{-})$, and characterize Bayes' theorem as follows 
\begin{align*}
&\P_i\left(
Y_{it}=y|X_{it}=x,Y_{i,t-1}=y_{-},Z_i=z\right) \\
&=\frac{f_{i,X_{it}|Y_{it},Y_{i,t-1},Z_i}\left(x|y,y_{-},z\right) 
\P_i\left( Y_{it}=y|Y_{i,t-1}=y_{-},Z_i=z\right) }{
\sum_{y,y_{-}}f_{i,X_{it}|Y_{it},Y_{i,t-1},Z_i}\left(x|y,y_{-},z\right) \P_i\left( Y_{it}=y|Y_{i,t-1}=y_{-},Z_i=z\right) }.
\end{align*}
Similar to \eqref{eq:heter_alpha_panel}, we introduce unit-specific tail thickness as \begin{align*}
&\tilde\alpha_i^{(yy_{-})}(z)=z'\theta^{(yy_{-})}+\lambda_i.
\end{align*}
Then, we have $1-F_{i,X_{it}|Y_{it},Y_{i,t-1},Z_i}\left(\cdot|y,y_{-},z\right) \in RV_{-\tilde\alpha_i^{(yy_{-})}(z)}.$

In the previous static panel analysis, our exercise was equivalent to normalizing $\theta^\po=\mathbf0$ and estimating $\theta^*=\theta^\pl$. Here, we similary normalize $\theta^{(00)}=0$ to ensure identification. With four transition patterns, a minimum of five periods of data $\mathbf{Y}_i=(Y_{i1,}Y_{i2,},Y_{i3,},Y_{i4},Y_{i5})'$ are required. The following four events cover all transition patterns and help construct the conditional likelihood:
\begin{align*}
E_{1} :\mathbf{Y}_i=\left( 0,0,1,1,0\right) ', \quad
E_{2} :\mathbf{Y}_i=\left( 0,1,1,0,0\right) ', \\
E_{3} :\mathbf{Y}_i=\left( 1,1,0,0,1\right) ',\quad 
E_{4} :\mathbf{Y}_i=\left( 1,0,0,1,1\right) '.
\end{align*}
Given the stationarity of the conditional joint distribution $Y_{it},Y_{i,t-1},X_{it}\;|\;Z_i$ implied by the first order Markov assumption, we have that as $x_t\rightarrow\infty$ for $t=1,\cdots,5$,
\begin{align*}
&\P_i\left( E_{1}\middle|\cup_{e=1}^4 E_{e},\left\{X_{it}=x_t\right\}_{t=1}^5,Z_i=z\right) \\
&\quad \sim \frac{x_{3}^{-z'\theta^{(01)}}x_{4}^{-z'\theta^{(11)}}x_{5}^{-z'\theta^{(10)}}}
{\left(
\begin{aligned}
&x_{3}^{-z'\theta^{(01)}}x_{4}^{-z'\theta^{(11)}}x_{5}^{-z'\theta^{(10)}}+x_2^{-z'\theta^{(01)}}x_{3}^{-z'\theta^{(11)}}x_{4}^{-z'\theta^{(10)}}\\
&\quad +x_2^{-z'\theta^{(11)}}x_{3}^{-z'\theta^{(10)}}x_{5}^{-z'\theta^{(01)}}+x_2^{-z'\theta^{(10)}}x_{4}^{-z'\theta^{(01)}}x_{5}^{-z'\theta^{(11)}}
\end{aligned}
\right)}.
\end{align*}
Then, the conditional MLE for $\left(\theta^{(01)},\theta^{(10)},\theta^{(11)}\right) $ can be constructed similarly to the static panel case in Section \ref{sec:panel-asym}. Note that unit $i$ contributes to the conditional likelihood only if $X_{it}$ appears in the tail for at least five periods. This data requirement may be challenging, so the method would be more suitable for datasets with a larger $N$, and a fixed but slightly larger $T$.

For large $T$, the bias correction in \citet{fernandez2018fixed} remains applicable to dynamic panel data models, so the estimator in Section \ref{sec:panel-largeT} remains valid.

\section{Monte Carlo simulations}\label{sec:sim}
We conduct two sets of Monte Carlo simulation experiments. Experiment 1 examines cross-sectional data and focuses on estimation performance, providing intuitions into when and how the proposed estimator outperforms the alternatives. Experiment 2 investigates panel data with large $N$ and large $T$,\footnote{We also conducted Monte Carlo simulations in panel data with large $N$ and small $T$ in a previous version of this paper. Results from these simulations are available upon request.} and focuses on pseudo out-of-sample forecasting performance, aligning more closely with the empirical example of bank loan charge-off rates. 

\subsection{Alternative estimators}\label{sec:alt-est}
We compare the proposed estimator with four alternatives: a Logit estimator using all observations (Logit, all $X$), a Logit estimator using only tail observations (Logit, tail $X$), a local linear estimator, and a local Logit estimator, where the first two are parametric estimators and the last two are nonparametric ones.

For cross-sectional data, let $\beta = (\beta_0,\,\beta_1)'$. First, the Logit estimator is defined as \[Y_i = \1(\beta_0+\beta_1X_i-\varepsilon_i\ge0),\quad\varepsilon_i\sim\text{ standard logistic.}\] Second, the local linear estimator is given by
\begin{align*}
&\hat \beta(x) = \arg\underset{\beta}{\min} \sumi  k_h(X_i - x) \left[Y_i - \beta_0 - \beta_1 (X_i - x)\right]^2,\text{ and }\hat{\E}[Y_i|X_i=x] =\hat \beta_0(x),
\end{align*}
where $k_h(\cdot)$ is a kernel function with bandwidth $h$. We employ a Gaussian kernel here, %i.e., $k_h(\cdot) = \frac 1 h \phi(\frac{\cdot}h)$, where $\phi(\cdot)$ is the standard normal pdf, 
and choose the bandwidth based on Silverman's rule of thumb $h \approx 1.06 \hat{\sigma} N^{-1/5}$, where $\hat{\sigma}$ is the standard deviation of the $X_i$s. The results are robust with respect to a range of bandwidth choices. Finally, the local Logit estimator is a flexible nonparametric estimator that specifically accounts for binary outcomes $Y_i$. Let $p_i(\beta,x) = \frac{1}{1 + \exp\left[-(\beta_0 + \beta_1 (X_i-x))\right]}$. The local Logit maximizes the locally weighted log-likelihood function
\begin{align*}
&\hat \beta(x)= \arg\underset{\beta}{\max} \sumi  k_h(X_i - x) \left[ Y_i \log p_i(\beta,x) + (1 - Y_i) \log\left(1 - p_i(\beta,x)\right) \right],\\
&\hat{\E}[Y_i|X_i=x] =p_i(\hat\beta(x),x).
\end{align*}

For panel data, the Logit estimator is given by
\[Y_{it} = \1(\beta_0+\beta_1X_{it}+C_i-\varepsilon_{it}\ge0),\quad \varepsilon_{it}\sim \text{standard logistic},\] where $C_i$ captures unobserved individual heterogeneity. In panels with large $N$ and large $T$, $\{\beta_0,\beta_1,\{C_i\}\}$ can again be jointly estimated using a fixed effects estimator with bias corrections \citep{fernandez2018fixed, stammann2016estimating}.
For the local linear and local Logit estimators, the setup is more flexible \[Y_{it}=g\left(X_{it},C_i,\varepsilon_{it}\right),\] where the function $g$ and the distributions of $C_i$ and $\varepsilon_{it}$ could be unknown. We incorporate a correlated random effects structure, which allows for the unobserved individual heterogeneity to be correlated with the covariates $X_{it}$ (and $Z_i)$ and thus may help enhance the performance of these alternatives. More specifically, suppose there is a sufficient statistic $V_i$ which could be multidimensional, such that $C_i|X_{i,1:T}=C_i|V_i$. One commonly used example of $V_i$ is the time sum of $X_{it}$, i.e., $V_i = \sum_t X_{it}.$ Then, $V_i$ can be included in the conditioning set for the local linear and local Logit estimators, so we essentially run a nonparametric regression to estimate the conditional mean $\E[Y_i|X_i,V_i]$: see details in \citet{liu2021identification} for general nonlinear panel data models. Without loss of generality, let $V_i$ be a scalar for notation simplicity. Now $\beta = (\beta_0,\,\beta_1,\,\beta_2)'$, $p_{it}(\beta,x) = \frac{1}{1 + \exp\left[-(\beta_0 + \beta_1 (X_{it}-x)+\beta_2(V_i-v))\right]}$, and
\begin{align*}
\text{Local linear: }& \underset{\beta}{\min} \sumi  k_h(X_{it} - x) k_h(V_i - v)\left[Y_{it} - \beta_0 - \beta_1 (X_{it} - x)-\beta_2(V_i-v)\right]^2,\\
\text{Local Logit: }&\underset{\beta}{\max} \sumi  k_h(X_{it} - x) k_h(V_i - v) \left[ Y_{it} \log p_{it}(\beta,x) + (1 - Y_{it}) \log\left(1 - p_{it}(\beta,x)\right) \right].
\end{align*} Furthermore, for all these estimators, we can also incorporate additional covariates $Z_i$, and the formulas are similar to those above.

\subsection{Experiment 1: cross-sectional data}

\begin{table}[t]
\caption{Monte Carlo design - Experiment 1}
\label{tab:sim-exp1-dgp}
\begin{center}
    \vspace{-1em}
\begin{tabular}{ll} \hline \hline
Model:& $Y_i = \1\left(X_i-\varepsilon_i \ge \text{med}_X-\text{med}_\varepsilon \right)$ \\
Covariate:& $X_i\sim |t_{\alpha_X}|$, $\alpha_X=0.5,1,1.5,2$ \\
Error term:& $\varepsilon_i\sim |t_{\alpha_{\varepsilon}}| $, $\alpha_{\varepsilon}=0.5,1,1.5,2$ \\
Sample Size:& $N=10000$ \\
\# Repetitions:& $N_{sim}=1000$ \\ \hline
\end{tabular}
\vspace{-1em}
\end{center}
\end{table}

Experiment 1 is based on cross-sectional data, where we focus on comparing estimation performance across estimators. 

The Monte Carlo design is summarized in Table \ref{tab:sim-exp1-dgp}. The data are generated from a threshold-crossing model with both $X_i$ and $\varepsilon_i$ following Student-$t$ distributions.\footnote{We use the difference in medians $\left(\text{med}_X-\text{med}_{\varepsilon}\right)$ as the threshold to keep the samples more balanced between $Y_i=0$ and 1.} Here we consider a range of $\alpha$ values from 0.5 to 2. As elaborated in Proposition \ref{prop:threshold-xing}, in the tail, $\alpha^\po=\alpha_X+\alpha_{\varepsilon}$ ranges from 1 to 4, $\alpha^\pl=\alpha_X$ varies from 0.5 to 2, and the extreme elasticity $-\left|\alpha^\pl-\alpha^\po\right|=-\alpha_{\varepsilon}$ spans from $-2$ to $-0.5$.\footnote{In many economic datasets, the tail index $\alpha$ typically falls between 1 and 2. For example, in a review paper, \citet{gabaix2009} mentions that ``it seems that the tail exponent of wealth is rather stable, perhaps around 1.5,'' referencing \citet{klass2006forbes}. Also, \citet{clauset2009power} remark that the tail index usually lies between 1 and 2 (note that the $\alpha$ in their notation corresponds to $\alpha-1$ in ours).}  The sample size $N=10000$ is directly comparable with our empirical data sets on bank loan charge-off rates, which comprises 8538 banks in the baseline sample. For each experimental setup, we execute 1000 Monte Carlo simulations.

Our tail estimator for $\alpha^\py$ is based on the ``rank-1/2'' estimator in \citet{gabaix2011rank}, which can be viewed as a refinement of the classic Hill estimator and often performs well in finite samples. Then, we estimate the conditional probability $\pi(x)$ using the sample analog of \eqref{eq:c-s-simple-prob1} and \eqref{eq:c-s-simple-prob2}. We also compare with the alternative estimators described in Section \ref{sec:alt-est}.
For both the tail estimator and the ``Logit, tail $X$'' estimator, we set the cutoffs $\x_N^\py$ at the 97.5th percentile of the distributions of $X$ for $Y=0$ and 1 separately.\footnote{\label{com:cutoff}The tail estimator is robust with respect to a range of cutoffs $\x_N^\py$. As evident from the log-log plot in Figure \ref{fig:sim-exp1-loglog}, the tail exhibits a pronounced downward pattern, with the slope remaining stable across small variations in $\x_N^\py$.} In the main text, we plot the comparisons regarding estimated parameters, conditional probability, and extreme elasticity, for the specification with $\alpha_X = 1$ and $\alpha_{\varepsilon}=1$. For detailed results across all model specifications, please refer to Tables  \ref{tab:sim-exp1-param}--\ref{tab:sim-exp1-ape} in the Appendix. The main messages are similar across specifications with different tail index values.

\begin{figure}[t]
\caption{Log-log plot - Experiment 1, $\alpha_X = 1$, $\alpha_{\varepsilon}=1$}
\label{fig:sim-exp1-loglog}
\begin{center}
    \vspace{-1em}
\includegraphics[width=.8\textwidth]{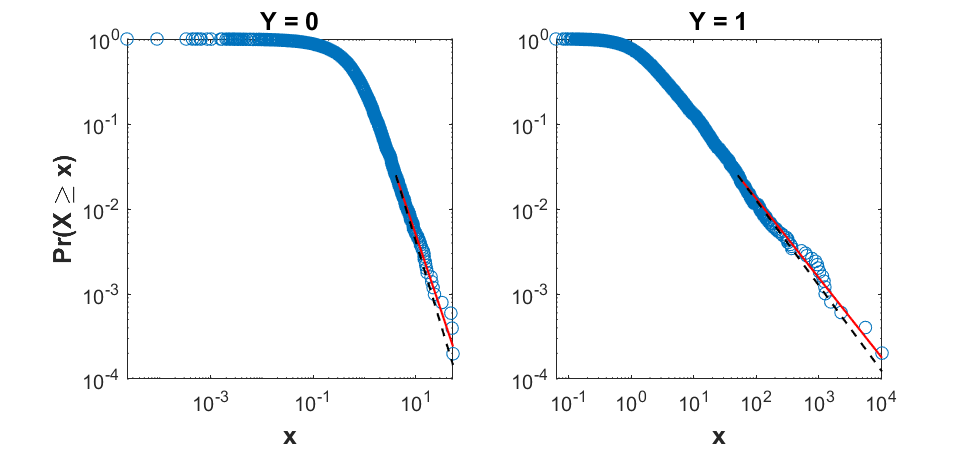}
    \vspace{-1em}
\end{center}
{\footnotesize {\em Notes:} %Based on one sample from the specification with $\alpha_X = 1$ and $\alpha_{\varepsilon}=1$. 
In the tail, $\alpha^\po=2$ and $\alpha^\pl=1$ are indicated by the black dash lines. The estimated $\hat\alpha^\po$ and $\hat\alpha^\pl$ are indicated by the red solid lines.}\setlength{\baselineskip}{4mm}
\end{figure}

\begin{figure}[t]
\caption{Parameter estimation - Experiment 1, $\alpha_X = 1$, $\alpha_{\varepsilon}=1$}
\label{fig:sim-exp1-param}
\begin{center}
    \vspace{-1em}
\includegraphics[width=.8\textwidth]{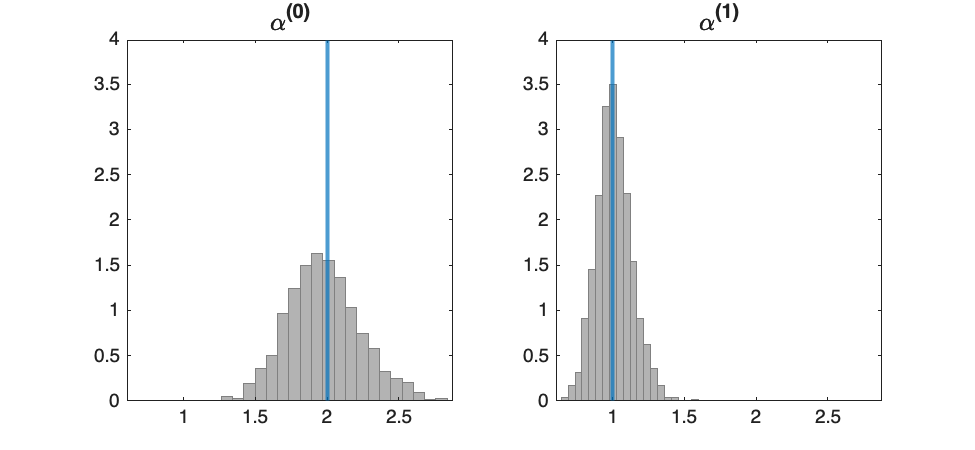} 
\vspace{-1em}
\end{center}
{\footnotesize {\em Notes:} %Specification with $\alpha_X = 1$ and $\alpha_{\varepsilon}=1$. 
In the tail, $\alpha^\po=2$ and $\alpha^\pl=1$ are indicated by the blue vertical lines.}\setlength{\baselineskip}{4mm}
\end{figure}

Figure \ref{fig:sim-exp1-loglog} depicts a log-log plot from one of the 1000 Monte Carlo simulations, where the tail observations align around a downward sloping line, contrasting with the flatter non-tail region. This distinct pattern between the tail and non-tail region is common in many empirical datasets as well, and suggests that the proposed tail estimator would be able to effectively capture the heavy tail behavior. The estimated tail indices (red solid lines) closely match their true asymptotic values (black dashed lines). Figure \ref{fig:sim-exp1-param} further shows the distributions of the parameter estimates from all 1000 Monte Carlo simulations. The distributions are both bell-shaped and centered around the true asymptotic values indicated by the blue vertical lines. Notably, $\alpha^\pl$ is more precisely estimated than $\alpha^\po$, as the former corresponds to a thicker tail.

\begin{figure}[t]
\caption{$\hat \P(Y=1|X=x)$ - Experiment 1, $\alpha_X = 1$, $\alpha_{\varepsilon}=1$}
\label{fig:sim-exp1-PY}
\vspace{-2em}
\begin{center}
\hspace*{-0.13\textwidth}\includegraphics[width=1.25\textwidth]{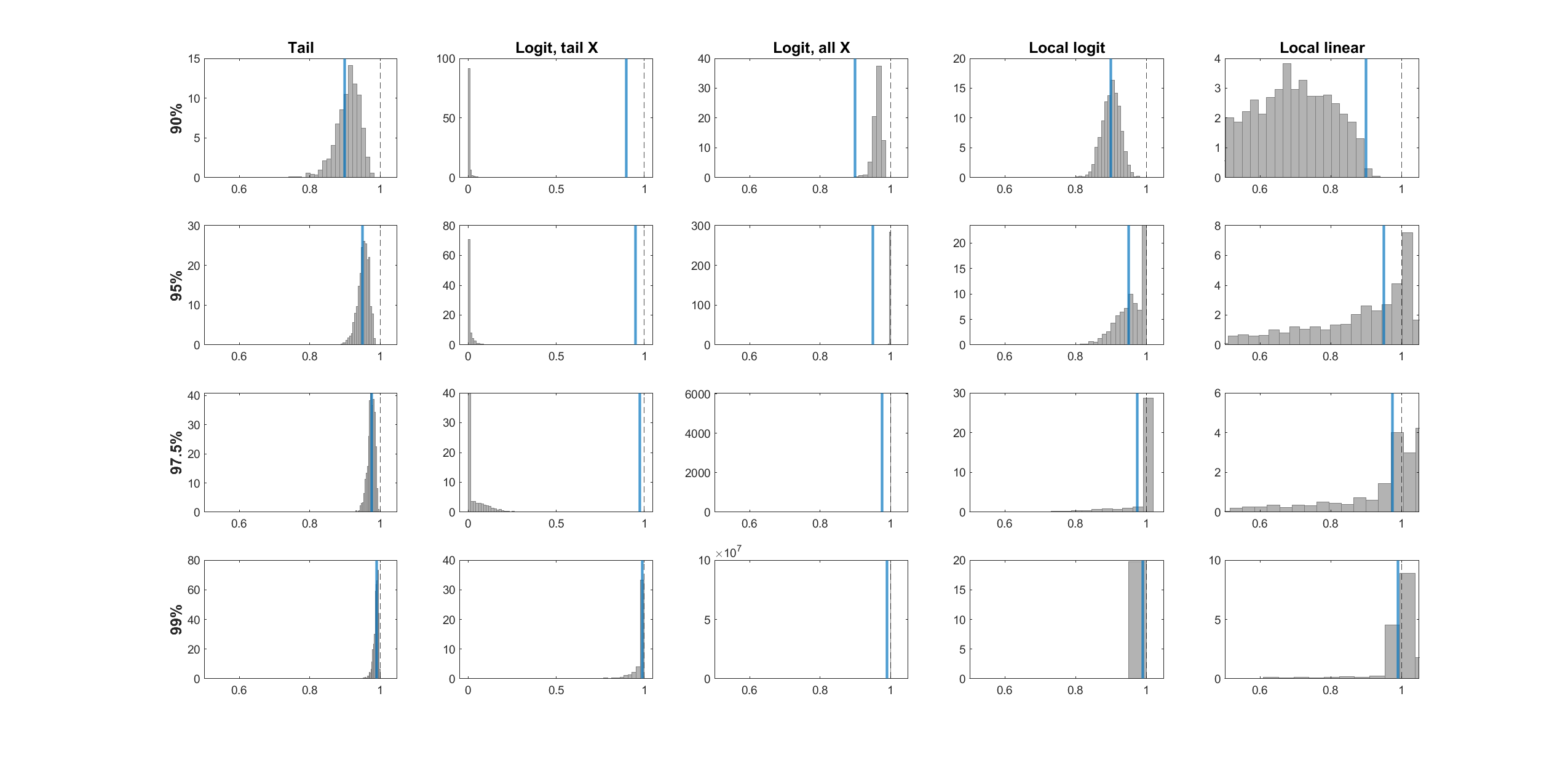}
\end{center}\vspace{-2em}
\footnotesize{\textit{Notes:} %Specification with $\alpha_X = 1$ and $\alpha_{\varepsilon}=1$. 
Each row evaluates the probability at a specific percentile of the distribution of $X$ (90th, 95th, 97.5th, and 99th). The true $\P(Y=1|X=x)$ is indicated by the blue vertical lines. }
\end{figure}

Figure \ref{fig:sim-exp1-PY} presents the estimated conditional probabilities for $X$ at the 90th, 95th, 97.5th, and 99th percentiles of its distribution. The blue vertical lines indicate the true conditional probability \(\P(Y=1\mid X=x)\), which increases with the evaluation point \(x\). The proposed tail estimator dominates all other alternatives, centered around the true probabilities with lower variance. The variance of the tail estimator decreases as the evaluation point $x$ increases, reflecting reduced estimation uncertainty as the conditional probability approaches one.

In contrast, the Logit estimators are largely biased due to the model misspecification that fails to account for the heavy tail. The direction of the bias depends on the relative positions of the estimation sample and the evaluation points. When the estimation sample includes all observations, the ``Logit, all $X$'' estimator shows an upward bias. This occurs because the estimation sample is overweighted by non-tail observations, and thus the tail evaluation points are too extreme given the misspecified thin-tailed logistic distributions, which leads to an overestimation of the probability of $Y=1$ in the tail evaluation points. Conversely, when the estimation sample includes tail observations only, the ``Logit, tail $X$'' estimator tends to exhibit a downward bias, except at the 99th percentile. The reason is that given the logistic model's misspecified assumption of a thinner tail, lower evaluation points appear relatively moderate compared to the tail observations for estimation, resulting in an underestimation of the probability of $Y=1$. 

The local Logit estimator performs reasonably well for less extreme evaluation points, such as when $x$ is at the 90th percentile, but it exhibits large bias and variance at more extreme points. The local linear estimator is even more flexible, and thus yields even larger variance across all evaluation points. Due to their poor estimation performance and relatively long computation times, we omit these nonparametric estimators in Experiment 2 and the empirical example below, except for the baseline case in Table \ref{tab:app-lps}.

Across all comparisons, we see that our tail estimator enjoys the best of both worlds as a semiparametric estimator that puts parametric assumptions only on the tail, the region of primary interest, while remaining agnostic in the non-tail region. Furthermore, as mentioned before, our approach offers significant improvement in finite samples for moderately large $x$ where $\pi(x)$ may not be very close to 1.

\begin{figure}[t]
\caption{Extreme elasticity estimation - Experiment 1, $\alpha_X = 1$, $\alpha_{\varepsilon}=1$}
\label{fig:sim-exp1-elas}
\begin{center}
    \vspace{-1em}
\hspace*{-0.13\textwidth}\includegraphics[width=1.25\textwidth]{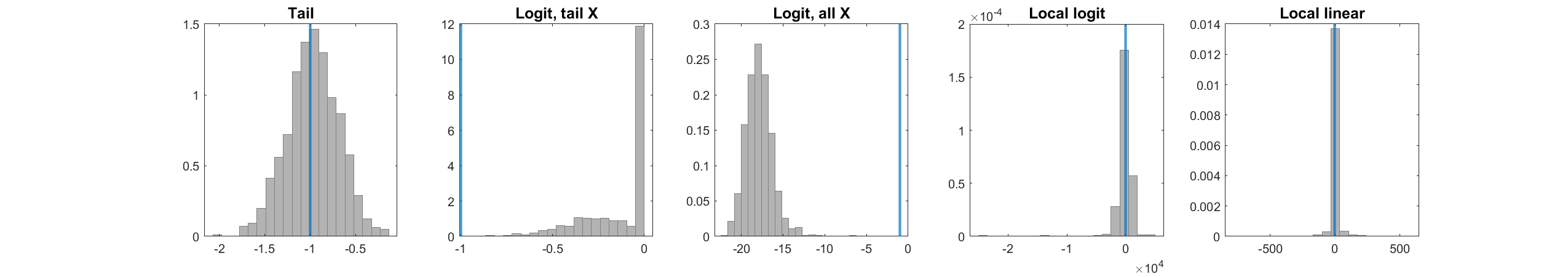}

\vspace{-1em}
\end{center}
\footnotesize{\textit{Notes:} %Specification with $\alpha_X = 1$ and $\alpha_{\varepsilon}=1$. 
The extreme elasticity, $-\left|\alpha^\pl-\alpha^\po\right|=-\alpha_{\varepsilon}=-1$, is indicated by the blue vertical lines. Evaluated at the 97.5th percentile of the distribution of $X$.}\setlength{\baselineskip}{4mm}
\end{figure}

Figure \ref{fig:sim-exp1-elas} plots the estimated extreme elasticities evaluated at the 97.5th percentile of the distribution of $X$, and the key messages are similar---the Logit estimators induce large bias, the nonparametric estimators exhibit high variance, whereas the proposed tail estimator flexibly and efficiently capture the tail behavior and yields the most accurate estimates.

\subsection{Experiment 2: panel data}\label{sec:sim-exp2}

\begin{table}[t]
\caption{Monte Carlo design - Experiment 2}
\label{tab:sim-exp2-dgp}
\begin{center}
    \vspace{-1em}
\begin{tabular}{ll} \hline \hline
Model:& $Y_{it} = \1\left(X_{it}-\varepsilon_{it} \ge  \text{med}_X-\text{med}_\varepsilon \right)$ \\
Covariate:& $X_{it}\sim |t_{\lambda_i}|$, $\lambda_i=\max\{\alpha_X+\tilde \lambda_i,0\}$,\\ & $\alpha_X=1,1.5,2$, $\tilde \lambda_i\sim 0.2N(-0.25,0.1^2) + 0.8N(0.25,0.1^2)$ \\
Error term:& $\varepsilon_{it}\sim |t_{\alpha_{\varepsilon}}| $, $\alpha_{\varepsilon}=1,1.5,2$ \\

Sample Size:& $N=10000$, $T=100$ \\
\# Repetitions:& $N_{sim}=1000$ \\ \hline
\end{tabular}
\vspace{-1em}
\end{center}
\end{table}

Experiment 2 examines panel data with large $N$ and large $T$, focusing on pseudo out-of-sample forecasting performance, which is closely related to the empirical analysis of bank loan charge-off rates. 

The Monte Carlo design is adapted from the cross-sectional case in Experiment 1, with specifics described in Table \ref{tab:sim-exp2-dgp}. Now we incorporate unit-specific tail thickness $\lambda_i$ for $X_{it}$, where the underlying distribution of $\lambda_i$ is bimodal. For the tail estimator, we employ bias corrections as in \citet{fernandez2018fixed} and \citet{stammann2016estimating}. See Sections \ref{sec:panel-largeT} and \ref{sec:alt-est} for additional details on the tail and alternative estimators, respectively. The cutoffs for both the tail estimator and ``Logit, tail $X$'' are set at the 90th percentile of the $X$ distributions. Recall that the tail estimator is robust over a range of cutoffs, as discussed in footnote \ref{com:cutoff}. 

The accuracy of density forecasts is evaluated using the log predictive score (LPS), as recommended by \citet{AmisanoGiacomini2007}. The LPS is calculated as $ LPS = \frac{1}{ N_{f}^{\dagger}}\sum_{i\in \mathcal N_{f}^{\dagger}}\log\hat{p}(y_{i,T+1}|D) $, where $ y_{i,T+1} $ is the outcome at time $ T+1 $, and $ \hat{p}(y_{i,T+1}|D) $ is the predictive likelihood based on the estimated model and observed data $ D $. $\mathcal N_{f}^{\dagger}$ is the set of units forecasted, characterized by following criteria: (1) $X_{it}$ exceeds the 90th percentile in the estimation sample, (2) $Y_{it}$ switches values over time among these tail observations, and (3) $X_{i,T+1}$ also falls within the tail during the forecasting period. %The expression $ \exp(LPS_{A} - LPS_{B}) $ quantifies the odds of future outcomes using predictor A over B. 
To assess the significance of the differences in the LPS, we combine the tests from \citet{AmisanoGiacomini2007} for density forecasts and \citet{qu2023comparing} for panel data.

\begin{table}[t]
\caption{Parameter estimation and forecast evaluation - Experiment 2}
\label{tab:sim-exp2-lps}
\begin{center}
    \vspace{-1em}
\scalebox{1}{
\begin{tabular}{lllrrr} \hline \hline 
& & & $\alpha_X=1$ & $\alpha_X=1.5$& $\alpha_X=2$ \\ \hline
\multirow{4}{*}{$\alpha_{\varepsilon}=1$}& $\hat \alpha^*$ & & -1.05 (0.03) & -1.06 (0.03) & -1.07 (0.03) \\ \cline{2-6}
& \multirow{3}{*}{$\text{LPS}\cdot N_{f}^{\dagger}$} & \it Tail & \it -113 \phantom{***} & \it -235 \phantom{***} & \it -318 \phantom{***} \\ \cline{3-6}
& & Logit, tail $X$ & -148 *** & -47 *** & -8 *** \\ 
& & Logit, all $X$ & -480 *** & -89 *** & -38 *** \\ \hline\hline
\multirow{4}{*}{$\alpha_{\varepsilon}=1.5$}& $\hat \alpha^*$ & & -1.53 (0.06) & -1.53 (0.04) & -1.52 (0.04) \\ \cline{2-6}
& \multirow{3}{*}{$\text{LPS}\cdot N_{f}^{\dagger}$} & \it Tail & \it -37 \phantom{***} & \it -106 \phantom{***} & \it -172 \phantom{***} \\ \cline{3-6}
& & Logit, tail $X$ & -50 *** & -43 *** & -7 *** \\ 
& & Logit, all $X$ & -267 *** & -64 *** & -33 *** \\ \hline\hline
\multirow{4}{*}{$\alpha_{\varepsilon}=2$}& $\hat \alpha^*$ & & -2.02 (0.09) & -1.98 (0.06) & -1.95 (0.05) \\ \cline{2-6}
& \multirow{3}{*}{$\text{LPS}\cdot N_{f}^{\dagger}$} & \it Tail & \it -15 \phantom{***} & \it -55 \phantom{***} & \it -102 \phantom{***} \\ \cline{3-6}
& & Logit, tail $X$ & -13 *** & -35 *** & -7 *** \\ 
& & Logit, all $X$ & -135 *** & -55 *** & -20 *** \\ \hline

\end{tabular} 
} 
\end{center}
{\footnotesize {\em Notes:} $\hat \alpha^*$ is estimated by the tail estimator, averaged across $N_{sim}=1000$ repetitions. Standard errors across repetitions are shown in parentheses. For the tail estimator, the table reports the exact values of $\text{LPS}\cdot N_{f}^{\dagger}$ (averaged across $N_{sim}=1000$ repetitions). For other estimators, the table reports their differences from the tail estimator. The tests compare other estimators with the tail estimator, with significance levels indicated by *: 10\%, **: 5\%, and ***: 1\%.}\setlength{\baselineskip}{4mm}
\end{table}

In Table \ref{tab:sim-exp2-lps}, the first row in each subpanel presents extreme elasticity estimates from the tail estimator, which are close to the true values, $-\alpha_{\varepsilon}$. Subsequent rows compare forecast accuracy among estimators, and we see that the tail estimator consistently outperforms both ``Logit, tail $X$'' and ``Logit, all $X$'' across all specifications. Although ``Logit, tail $X$'' is better than ``Logit, all $X$,'' it still performs worse than the tail estimator, especially with smaller $\alpha_X$ and $\alpha_{\varepsilon}$, where heavy tails are more pronounced. Therefore, it is important to distinguish the tail from the middle of the sample as well as account for heavy tail patterns. To further demonstrate this point, Figure \ref{fig:sim-exp2-lps} in the Appendix shows the scatter plots of the LPS from all Monte Carlo repetitions in the setup with $\alpha_X = 1$ and $\alpha_{\varepsilon}=1$.

\section{Empirical example: housing prices and bank riskiness}\label{sec:app}

Charge-off rates serve as an indicator of bank losses. A bank could be riskier for a particular type of loan if its corresponding charge-off rates exceed a certain level. In our analysis, we focus on a panel of small banks with assets less than \$1 billion, similar to \citet{liu2023forecasting}. Since the banks are small, it is reasonable to assume that they operate primarily in local markets. In this empirical example, we examine the impact of substantial local housing price declines on the riskiness of small banks, considering that this channel played a pivotal role during the 2007--2008 financial crisis.

\subsection{Data and sample}

In this empirical example, we focus on the setup in \eqref{eq:longT} in Section \ref{sec:panel-largeT} for panel data models with large $N$ and large $T$. The binary outcome $Y_{it}$ is a risk dummy based on the loan charge-off rate for a specific loan type of bank $i$ in quarter $t$. $Y_{it} = 1$ if the charge-off rate is greater than a level $c$. We present results for $c = 0$ in the main text and relegate robustness checks for alternative $c$ values to the Appendix. The qualitative findings are consistent across different levels of $c$. This exercise aligns with policy analysis practices, where policymakers often compare current charge-off rates to historical averages to monitor bank risk: see, for example, the Federal Reserve Board's Financial Stability Report.\footnote{\href{https://www.federalreserve.gov/publications/financial-stability-report.htm}{https://www.federalreserve.gov/publications/financial-stability-report.htm}.}  For instance, the historical average charge-off rate for Residential Real Estate (RRE) loans is around 0.1\%, with the average plus two standard deviations around 0.2\%. Robustness checks include these levels.

The extreme regressor $X_{it}$ captures decreases in local housing prices. To convert the extreme values to the right tail, we define $X_{it}$ as the deflation rate of the local housing price in the previous quarter.\footnote{We define the tail region as observations above the 90th percentile of $X_{it}$, as in Experiment 2. Given that the 90th percentiles are positive in all our samples, $\log X_{it}$ is well-defined in the tail.}
The additional covariate $Z_i$ is given by the average quarterly change in the local unemployment rate, accounting for local economic conditions.

Our data are obtained from the following sources. Bank balance sheet data at a quarterly frequency, such as loan charge-off rates, are constructed based on the Call Reports from the Federal Reserve Bank of Chicago.\footnote{Please refer to Appendix D in \citet{liu2023forecasting} for details on constructing loan charge-off rates from the raw data.} The local market is defined at the county level, and the local market for each bank is determined based on the annual Summary of Deposits from the Federal Deposit Insurance Corporation.\footnote{We calculate the deposits received by each bank from every county and link the bank to the county from which it received the largest amount of deposits. We also assess the robustness of our analysis by constructing weighted averages of the covariates, using deposit proportions from each county as weights. The results are very similar, which can be attributed to the concentration of deposits across counties and the similarity in covariate values among neighboring counties. } Housing price indices at a quarterly frequency (all transactions, not seasonally adjusted) are sourced from the Federal Housing Finance Agency, and the 3-digit zip code data are converted to the county level using the HUD USPS ZIP Code Crosswalk from the Department of Housing and Urban Development.\footnote{We use the 2010Q1 Crosswalk, the earliest available version that is close to our empirical time frame. We also experimented with the county-level Zillow Home Value Index, but it unfortunately has more missing data within our study period, though the results are qualitatively similar.} The county-level not seasonally adjusted unemployment rates are obtained from the Bureau of Labor Statistics website, and we aggregate the monthly data to a quarterly frequency by time averaging. Finally, we standardize the county-level housing price deflations and unemployment rate changes using the means and standard deviations calculated from their corresponding aggregate time series.

Our baseline sample focuses on the RRE charge-off rates. The estimation sample spans from 1999Q4 to 2009Q3, comprising $N=8538$ small banks across 40 quarters.\footnote{We exclude banks with: 1) average non-missing domestic total assets exceeding \$1 billion, or 2) missing target charge-off rates for all periods in the sample. These criteria result in the removal of 583 and 162 banks, respectively, for the baseline RRE sample.} There are $N_{e}^{\dagger}=2642$ banks in the tail for more than one period and contributing to the likelihood, that is, $X_{it}$ is above the 90th percentile of the estimation sample and $Y_{it}$ switches values across time for these tail observations. We perform a pseudo out-of-sample forecast for the period of 2009Q4. There are $N_{f}^{\dagger}=2098$ banks that satisfy the conditions for $N_{e}^{\dagger}$ and additionally have $X_{i,T+1}$ fall in the tail during the forecasting period.\footnote{To account for banks' endogenous exit choices, one could extend to a panel Tobit model as in \citet{liu2023forecasting}, which is left for future research.}

\begin{figure}[t]
\caption{Log-log plot - banking application, baseline sample}
\label{fig:app-loglog}
\begin{center}
    \vspace{-1em}
\includegraphics[width=.8\textwidth]{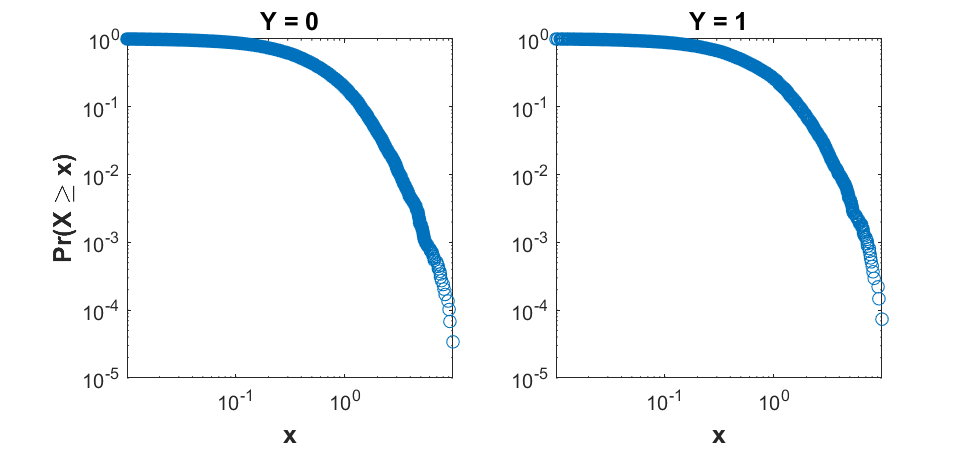}
\vspace{-2em}
\end{center}
%{\footnotesize {\em Notes:} Baseline sample: RRE, forecasting period = 2009Q4, $T=40$.}\setlength{\baselineskip}{4mm}
\end{figure}

For robustness check, we also consider (non-farm) non-residential commercial real estate (CRE) charge-off rates, various time dimensions in the estimation samples ranging from 32 to 48 quarters, and different forecasting periods $T+1 =$ 2009Q3 and 2009Q4. The sample statistics of these samples are reported in Table \ref{tab:app-stat} in the Appendix. The main results remain consistent across different samples. 

Based on the sample skewness and kurtosis in Table \ref{tab:app-stat} for all samples, as well as the log-log plot in Figure \ref{fig:app-loglog} for the baseline sample,\footnote{Fitted lines, such as those in Figure \ref{fig:sim-exp1-loglog}, are not plotted here, because these lines could be unit-specific due to observed heterogeneity $Z_i$ and unobserved heterogeneity $\{\lambda_i,\C_i\}$.} we see that $X_{it}|Y_{it}=y$ indeed exhibits heavy right tails, so our tail estimator would be more appealing. Furthermore, it is worth noting that in finite samples, the probability of $Y_{it}=1$ for tail $x$, $\pi(x)$, is not necessarily very close to 1. This is evident from both the empirical data (see the log-log plot in Figure \ref{fig:app-loglog} and the histogram in Figure \ref{fig:app-hist} in the Appendix) and the estimated model (see the predictive probability in Figure \ref{fig:app-prob}).

\subsection{Results}

\begin{table}[t]
\caption{Forecast evaluation - banking application, baseline samples}
\label{tab:app-lps}
\begin{center}
    \vspace{-1em}
\scalebox{1}{
\begin{tabular}{lrr} \hline \hline
& RRE\phantom{***} & CRE\phantom{***} \\ \hline
\it Tail & \it -1433.72\phantom{***} & \it -1134.48\phantom{***} \\ \hline
Logit, tail $X$ & -18.66 *** & -23.49 *** \\
Logit, all $X$ & -160.25 *** & -54.26 *** \\
Local Logit & -429.20 *** & -516.43 *** \\
\hline
\end{tabular} 
} 
\end{center}
{\footnotesize {\em Notes:} Forecasting period = 2009Q4, $T=40$. For the tail estimator, the table reports the exact values of $\text{LPS}\cdot N_{f}^{\dagger}$. For other estimators, the table reports their differences from the tail estimator. The tests compare other estimators with the tail estimator, with significance levels indicated by *: 10\%, **: 5\%, and ***: 1\%.}\setlength{\baselineskip}{4mm}
\vspace{-.5em}
\end{table}

Table \ref{tab:app-lps} compares forecasting performance across estimators. The estimators are similar to those in the Monte Carlo Experiment 2: see Sections \ref{sec:alt-est} and \ref{sec:sim-exp2} for more details. The proposed tail estimator is the overall best. ``Logit, tail $X$'' ranks second, yet still significantly worse than the tail estimator, indicating the importance of carefully addressing the heavy tail behavior. ``Logit, all $X$'' ranks third, indicating the presence of large misspecification and possible distinct pattern in the tail compared to the middle range. Local Logit yields the least accurate forecasts, suggesting that while being flexible, the nonparametric approach could be too noisy given the limited data available in the tail. %Intuitively, for the baseline RRE sample, the odds, given by the exponential of the difference in the LPS, indicate that the tail estimator provides probability forecasts that are on average 1\%, 8\%, 23\% more likely than those from ``Logit, tail $X$,'' ``Logit, all $X$,'' and local Logit, respectively.

Table \ref{tab:app-lps-ape-all} in the Appendix provides further details on the parameter and APE estimates based on the tail estimator. First, the coefficient on $\log X_{it}$ is always significant, with values around 0.95--1.15 for the RRE samples and around 1.2--1.4 for the CRE ones. The negative of this coefficient can be roughly viewed as the homogeneous part of the extreme elasticity, and the estimated values suggest the potential presence of heavy tails. Second, the coefficient on $Z_i\log X_{it}$ is mostly positive, being larger and significant for the RRE samples, while smaller and insignificant for the CRE ones. This difference aligns with intuition: larger increases in unemployment could directly amplify the impact of housing price drops on residential loan risk, whereas changes in unemployment may not directly affect non-residential commercial loan risk. Third, the estimated APEs are around 0.15 for the RRE samples and 0.13 for the CRE ones, which could be interpreted as that in the tail, a 1\% decrease in housing prices in the previous quarter corresponds to approximately a 0.15 (0.13) increase in the probability of high risk for RRE (CRE) loans. Finally, the unobserved heterogeneity exhibits a larger dispersion than the observed heterogeneity, as the sample variances of the estimated $\tilde A_i$ are around 1 while the sample variances of $\hat\theta^*_{Z\log X}Z_i\log X_{it}$ range from 0.1--0.3.

Tables \ref{tab:app-lps-ape-all} and \ref{tab:app-lps-ape-c} in the Appendix also show that our results are robust with respect to the level $c=0.01, 0.02, 0.05, 0.1, 0.2, 0.5, 1$,\footnote{As the level $c$ increases, the coefficient on $Z_i\log X_{it}$ becomes less significant, particularly becoming insignificant when $c=1$. This is because a larger $c$ reduces the occurrence of $Y_{it}=1$, leading to fewer banks in the tail with $Y_{it}$ switching values across time. Then, the effective sample size for estimation $N_{e}^{\dagger}$ substantially decreases, resulting in noisier estimates.} to both CRE and RRE loans, to the estimation sample's time dimension $T=32,36,40,44,48$, and to the forecasting period being 2009Q3 and 2009Q4. 

\begin{figure}[t]
\caption{Predictive probability of high risk - baseline sample}
\label{fig:app-prob}
\begin{center}
    \vspace{-1.5em}
\includegraphics[width=.8\textwidth]{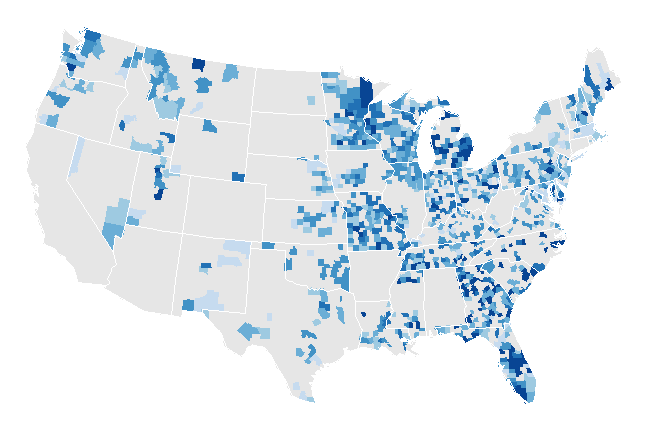}

\vspace{-1em}
\includegraphics[width=1\textwidth, trim=0 2em 0 2em, clip]{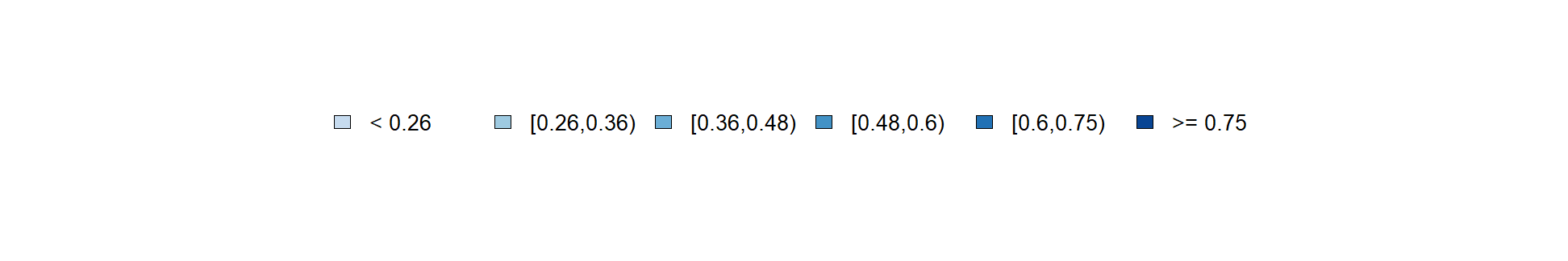}
\vspace{-5em}
\end{center}
%{\footnotesize {\em Notes:} Baseline sample: RRE, forecasting period = 2009Q4, $T=40$.}\setlength{\baselineskip}{4mm}
\end{figure}

Figure \ref{fig:app-prob} plots the average predictive probability of high risk across counties, with darker shades indicating higher risk levels in those counties. Notably, Florida and the Great Lakes regions appear as areas with elevated risk. This heterogeneity in risk may arise from observed heterogeneity, such as variations in local housing prices $X_{i,T+1}$ and local unemployment rates $Z_i$, as well as unobserved heterogeneity $\{\lambda_i,\C_i\}$, which is captured by $\tilde A_i$, where larger values of $\tilde A_i$ are associated with higher risk: see \eqref{eq:longT}. 

\begin{table}[t]
\caption{Heterogeneity and bank characteristics - baseline sample}
\label{tab:app-het}
\begin{center}
    \vspace{-1em}
\scalebox{1}{
\begin{tabular}{lrrrrr} 
\hline\hline 
& \multicolumn{2}{c}{Initial} & & \multicolumn{2}{c}{Average} \\ 
\cline{2-3} \cline{5-6}
 & Coef. & s.e. & & Coef. & s.e. \\ \hline
Log assets      & 0.20{***}         & (0.03) & & 0.19{***}         & (0.03) \\ 
Loan frac.      & 0.28{***}         & (0.10) & & 0.68{***}         & (0.11) \\ 
Capital-assets  & -0.65{\phantom{***}} & (0.56) & & -2.05{***}        & (0.60) \\ 
Loan-assets     & 0.83{***}         & (0.18) & & 1.77{***}         & (0.19) \\ 
ALLL-loan       & 2.61{\phantom{***}}  & (3.18) & & 15.01{***}        & (3.45) \\ 
Diversification & 1.00{**}{\phantom{*}}         & (0.44) & & 2.79{***}         & (0.50) \\ 
Ret.\ on assets & -30.53{***}       & (9.62) & & -36.51{***}       & (11.66) \\ 
OCA             & -19.34{**}{\phantom{*}}        & (9.05) & & -60.19{***}       & (12.30) \\ 
Intercept       & -3.68{***}        & (0.37) & & -4.25{***}        & (0.38) \\  \hline
\end{tabular} 
} 
\end{center}
{\footnotesize {\em Notes:} %Baseline sample: RRE, forecasting period = 2009Q4, $T=40$. 
Significance levels are indicated by *: 10\%, **: 5\%, and ***: 1\%.}\setlength{\baselineskip}{4mm}
\vspace{-.5em}
\end{table}

Accordingly, in Table \ref{tab:app-het}, we regress the estimated $\tilde A_i$ on bank characteristics. The ``Initial'' column uses bank characteristics in the initial period 1999Q4, which are exogenous to subsequent dynamics. The ``Average'' column uses time-averaged bank characteristics over the estimation period, incorporating more recent information. The findings across both columns are in general consistent. We see that larger banks (measured by log assets), those that specialize in RRE loans (measured by RRE loans to total loans ratio), in lending activities (measured by loan to assets ratio), with more diversified earnings (measured by the share of non-interest income to total income), and with higher operational efficiency (measured by the negative of overhead costs to assets ratio), tend to have a higher $\tilde A_i$ and thus higher risk, as these characteristics could help increase banks' capacity to assume riskier RRE loans. On the other hand, higher profitability (measured by return on assets) corresponds to a lower $\tilde A_i$, potentially due to reverse causality where reduced risk boosts profitability. The credit quality (measured by ALLL to total loans ratio) and the capital-asset ratio do not have significant effects in the regression using initial bank characteristics.

\section{Conclusion}\label{sec:conclusion}
This paper proposes a novel semiparametric method based on Bayes' theorem and RV functions. It models heavy tail behavior through a Pareto approximation, while making no parametric assumptions on the relationship between the outcome and covariates outside of the tail region. 

The proposed method is particularly useful in panel data models, accounting for unobserved unit-specific heterogeneity. We show that under regularity conditions, our objective function asymptotically aligns with a panel Logit regression on tail observations using $\log X_{it}$ as a regressor. Then, various established econometric techniques could be applicable, which could be convenient for empirical research. Specifically, in panels with large $N$ and small $T$, the unobserved unit-specific tail thickness and RV functions can be canceled out via the conditional MLE; in panels with large $N$ and large $T$, bias correction methods \citep{fernandez2018fixed,stammann2016estimating} can be employed to estimate unit-specific parameters and predict unit-specific future outcomes. Furthermore, we also extend our method to dynamic panel data models that incorporate lagged outcomes.

The potential applications of our method could encompass both microeconomics and macroeconomics studies, particularly valuable in light of recent extreme events. For example, one may be interested in analyzing the effect of large sovereign debts on country default risks, or assessing the impact of extreme weather on productivity at the regional, firm, and individual levels.

\ifsubmission
\section*{Acknowledgments}
We thank Frank Diebold, Christian Haefke, Bo Honor\'{e}, Roger Klein, Ulrich M\"{u}ller, Frank Schorfheide, and seminar participants at the FRB Chicago, Monash University, University of Melbourne, Reserve Bank of Australia, University of Sydney, FRB Philadelphia, UC Irvine, UCSD, and Penn State, as well as conference participants at the Applied Time Series Econometrics Workshop at the FRB St.\ Louis, Dolomiti Macro Meetings, AiE Conference and Festschrift in Honor of Joon Y.\ Park, Midwest Econometrics Group Annual Meeting, Greater New York Econometrics Colloquium, Women in Macroeconomics Workshop II, NBER Summer Institute, and NBER-NSF Time Series Conference for helpful comments and discussions. %The authors are solely responsible for any remaining errors.
\else\fi

%\clearpage
\begin{singlespace}
    \bibliographystyle{econometrica}
\bibliography{mybib}
\end{singlespace}

\end{document}

\title{\vspace{-2.5em}\Large Online Appendix:\\Binary Outcome Models with Extreme Covariates:\\
Estimation and Prediction}
\author{Laura Liu  \and Yulong Wang }

\maketitle

%=====================================================================================================
%\clearpage
\appendix
\renewcommand{\theequation}{A.\arabic{equation}}
\setcounter{equation}{0} % Reset equation counter for the appendix
\renewcommand{\thefigure}{A.\arabic{figure}}
\setcounter{figure}{0} 
\renewcommand{\thetable}{A.\arabic{table}}
\setcounter{table}{0} 
\setcounter{footnote}{0} 
\setcounter{page}{1}
\renewcommand{\thepage}{A-\arabic{page}}

This online appendix is organized as follows.  In Appendix \ref{sec:proofs}, we provide proofs for all propositions and theorems in the main paper, and establish the asymptotic properties of the proposed tail estimator in both cross-sectional and panel data setups. In Appendix \ref{sec:app-tab-fig}, we present additional figures and tables that supplement the simulation and empirical results in the main text.

\section{Proofs}\label{sec:proofs}
\subsection{Proofs for Section \ref{sec:c-s}}

We first state some useful properties of RV functions in the following lemma.
\begin{lemma}[Properties of RV functions]\label{lem:rv}
Let $f\in RV_{-\alpha }$ and $g\in RV_{-\beta }$ with $\alpha,\beta>0$. We have the following:
\begin{enumerate}[label=(\alph*)]
\item $f(x)g(x)\in RV_{-\alpha -\beta }.$
%\item $f(x)+g(x)\in RV_{-\min \{\alpha,\beta \}}.$
\item If $f'(x)$ is non-increasing, $\frac{xf'(x)}{f(x)}\rightarrow-\alpha$ as $x\rightarrow\infty$, and hence $f'(x)\in RV_{-\alpha -1}.$
\item If $\alpha>1$, $\int_{x}^{\infty }f(s)ds\in RV_{-\alpha +1}.$
\end{enumerate}
\end{lemma}
\begin{proof}
See Section B in \citet{de2006extreme}.
\end{proof}

%~~~~~~~~~~~~~~~~~~~~~~~~~~~~~~~~~~~~~~~~~~~~~~~~~~~~~~~~~~~~~~~~~~~~~
\bigskip
\noindent\begin{proof}[Proof of Proposition \ref{prop:elas}]

\noindent Recall that
\begin{align*}
\pi(x) &=\frac{f_{X|Y}\left( x|1\right) \P\left(
Y=1\right) }{f_{X|Y}\left( x|1\right) \P\left( Y=1\right)
+f_{X|Y}\left( x|0\right) \P\left( Y=0\right) } =\frac{1}{1+\frac{f_{X|Y}\left( x|0\right) \P\left( Y=0\right) }{
f_{X|Y}\left( x|1\right) \P\left( Y=1\right) }}.
\end{align*}
Let $R=\frac{\P\left( Y=0\right)}{\P\left( Y=1\right)}$. Then, $0<R<\infty$ according to the non-degeneracy condition (c). Define the term in the denominator
\begin{align*}
\pi ^*(x) &=\frac{1}{\pi(x) }-1 =\frac{\P\left( Y=0\right) f_{X|Y}\left( x|0\right) }{\P\left( Y=1\right) f_{X|Y}\left( x|1\right) }=R\cdot\frac{f_{X|Y}\left( x|0\right) }{f_{X|Y}\left( x|1\right) }.
\end{align*}

Given the RV condition (a) and the non-increasing condition (b), we have $f_{X|Y}\left( x|y\right) \in RV_{-\alpha^\py-1}$ and \[\frac{xf'_{X|Y}\left( x|y\right)}{f_{X|Y}\left( x|y\right)}\rightarrow-\alpha^\py-1,\] as $x\rightarrow\infty$, by Lemma \ref{lem:rv}(b). Then the elasticity of $\pi^*(x)$ is given by
\begin{align}
\delta^*(x)&=\frac{\partial \pi ^*(x) }{\partial x}\frac{x}{\pi
^*(x) } 
=\frac{xf_{X|Y}'\left( x|0\right) }{f_{X|Y}\left( x|0\right) }-
\frac{xf_{X|Y}'\left( x|1\right) }{f_{X|Y}\left( x|1\right) } \notag
\\
&\rightarrow\left(-\alpha ^\po-1\right)-\left(-\alpha ^\pl-1\right)=\alpha^\pl-\alpha^\po. \label{eq:limit ratio}
\end{align}

As in \eqref{eq:elas-decomp}, the extreme elasticity can be decomposed as
\begin{align*}
\delta(x) &=\delta_{\pi}(x)+\delta_{1-\pi}(x).
\end{align*}
And each component can be calculated via the chain rule
\begin{align*}
\delta_{\pi}(x) &=-\delta^*(x)\frac{ \pi^*(x) }{1+\pi^*(x)}=-\delta^*(x)\left(1-\pi(x)\right), \\
\delta_{1-\pi}(x)&=\delta^*(x)\frac{ 1 }{1+\pi^*(x)}=\delta^*(x)\pi(x),
\end{align*}
where $\delta^*(x)$ is given by \eqref{eq:limit ratio}. 
\begin{itemize}
\item If $\alpha^\po>\alpha^\pl$, $\pi(x)\rightarrow 1$ as $x\rightarrow\infty$, and hence $\delta(x) \sim\delta^*(x)\rightarrow\alpha^\pl-\alpha^\po$. 
\item If $\alpha^\po<\alpha^\pl$, $\pi(x)\rightarrow0$ as $x\rightarrow\infty$, and hence $\delta(x) \sim-\delta^*(x)\rightarrow\alpha^\po-\alpha^\pl$. 
\item If $\alpha^\po=\alpha^\pl$, $\delta^*(x)\rightarrow0$ and $\pi(x)\rightarrow \frac 1 {1+R}$, as $x\rightarrow\infty$, so we have \[\delta(x) =\delta_{\pi}(x)+\delta_{1-\pi}(x)\sim\delta^*(x)\frac {1-R} {1+R}.\] As $0<R<\infty$, it follow that $\left|\frac{1-R}{1+R}\right|<1$, and thus $\delta(x) \rightarrow 0$. 
\end{itemize}
Combining all cases, we have that, as $x\rightarrow\infty$,
\begin{align*}
\delta(x) \rightarrow-|\alpha^\pl-\alpha^\po|=-|\alpha^*|.
\end{align*}
\hfill\end{proof}

%~~~~~~~~~~~~~~~~~~~~~~~~~~~~~~~~~~~~~~~~~~~~~~~~~~~~~~~~~~~~~~~~~~~~~
\bigskip
\noindent\begin{proof}[Proof of Proposition \ref{prop:threshold-xing}] 

\noindent For $Y=1$, 
\begin{align*}
1-F_{X|Y}\left( x|1\right) &=\P\left( X> x|Y=1\right) =\frac{\P\left( X> x,Y=1\right) }{\P\left( Y=1\right) } \\
&=\frac{\int_{x}^{\infty }\P\left( Y=1|X=s\right) f_X\left( s\right) ds}{\P\left( Y=1\right) } \\
&=\frac{\int_{x}^{\infty }F_{\varepsilon }\left( s\right) f_X\left( s\right) ds}{\P\left( Y=1\right) },
\end{align*}
where the third equality is by Bayes' rule, and the fourth equality is by the independence condition (c). 
For the numerator, we have that $F_{\varepsilon }(x)\rightarrow1$ as $x\rightarrow\infty$, and thus $\int_{x}^{\infty }F_{\varepsilon }\left( s\right) f_X\left( s\right) ds\rightarrow\int_{x}^{\infty } f_X\left( s\right) ds=1-F_X(x)\in RV_{-\alpha_X}$ by the RV condition (a). For the denominator, the RV condition (a) implies that $X$ and $\varepsilon$ have overlapping support, at least in the right tail, so $0<\P\left( Y=1\right)<1$. Then, 
\begin{align}
1-F_{X|Y}\left( x|1\right)\in RV_{-\alpha_X}\text{ and }\alpha^\pl=\alpha_X.\label{eq:xing-alpha1}
\end{align}

For $Y=0$, a similar argument yields that
\begin{align*}
1-F_{X|Y}\left( x|0\right) &=\P\left( X> x|Y=0\right) =\frac{\P\left( X> x,Y=0\right) }{\P\left( Y=0\right) } \\
&=\frac{\int_{x }^{\infty }\P\left( Y=0|X=s\right) f_X\left( s\right) ds}{\P\left( Y=0\right) } \\
&=\frac{\int_{x }^{\infty }\left( 1-F_{\varepsilon }\left( s\right) \right) f_X\left( s\right) ds}{\P\left( Y=0\right) }.
\end{align*}
For the numerator, by the RV condition (a), we have that $1-F_{\varepsilon }(x)\in RV_{-\alpha_{\varepsilon}}$, and $f_X(x)\in RV_{-\alpha_X-1}$ based on Lemma \ref{lem:rv}(b), and thus $\int_{x}^{\infty }\left( 1-F_{\varepsilon }\left( s\right) \right) f_X\left( s\right) ds\in RV_{-\alpha_X-\alpha_{\varepsilon}}$ by Lemma \ref{lem:rv}(a,c). For the denominator, again, the overlapping support implies that $0<\P\left( Y=0\right)<1$. Then, 
\begin{align}
1-F_{X|Y}\left( x|0\right)\in RV_{-\alpha_X-\alpha_{\varepsilon}}\text{ and }\alpha^\po=\alpha_X+\alpha_{\varepsilon}.\label{eq:xing-alpha0}
\end{align}

Combining \eqref{eq:xing-alpha1} and \eqref{eq:xing-alpha0}, we further obtain that $\alpha_{\varepsilon}=\alpha^\po-\alpha^\pl$.
\end{proof}

%~~~~~~~~~~~~~~~~~~~~~~~~~~~~~~~~~~~~~~~~~~~~~~~~~~~~~~~~~~~~~~~~~~~~~
\bigskip
\begin{proposition}[Cross-sec.\ data: existence of tail average of partial effects]\label{prop:avg-partial-effects}
Under conditions (a) and (b) of Proposition \ref{prop:elas}, assume further that $\pi(x)$ and $f_{X|Y}\left( x|y\right) $ are continuously differentiable for $x\ge \x$, for some $\x>0$. Then, the tail average of the partial effects $\E\left[\pi'(X)|X\ge\x \right]$ exists, and converges to 0 as $\x\rightarrow\infty$.
\end{proposition}
\begin{proof}
Let the unconditional pdf be $f(x)=f_{X|Y}\left(x|1\right) \P\left(Y=1\right)
+f_{X|Y}\left( x|0\right) \P\left( Y=0\right)$, and $F(x)$ is the corresponding unconditional cdf. Then, the tail average of the partial effects is given by
\[\E\left[\pi'(X)|X\ge\x \right]=\frac{\int_{\x}^\infty\pi'(x) f(x)dx}{1-F(\x)}.\]  

For the numerator,
\begin{align*}
\left|\int_{\x}^\infty\pi'(x) f(x)dx\right|&=
\left|\left.\pi(x)f(x)\right|_{\x}^\infty-\int_{\x}^\infty\pi(x)f'(x)dx\right|\\
&\le\pi(\x)f(\x)+\int_{\x}^\infty\pi(x)\left|f'(x)\right|dx\\
&\le f(\x)+\int_{\x}^\infty\left|f'(x)\right|dx\\
&= 2f(\x).
\end{align*}
The first line is by integration by parts, given that $\pi(x)$ and $f_{X|Y}\left( x|y\right) $ are continuously differentiable for $x\ge \x$. The second and third lines are by the fact that $0\le\pi(x)\le 1$ and that $f(x)\rightarrow0$ as $x\rightarrow\infty$. The last line follows from that $f'(x)$ is non-increasing in $x\ge\x$.

Combining with the denominator, we have that
\begin{align*}
\left|\E\left[\pi'(X)|X\ge\x \right]\right|&\le\frac{2f(\x)}{1-F(\x)}\\
& = 2\frac{f_{X|Y}\left(\x|1\right) \P\left(Y=1\right)+f_{X|Y}\left(\x|0\right) \P\left(Y=0\right)}{\left(1-F_{X|Y}\left(\x|1\right) \right)\P\left(Y=1\right)+\left(1-F_{X|Y}\left(\x|0\right) \right) \P\left(Y=0\right)}\\
&\le 2\max\left\{\frac{f_{X|Y}\left(\x|1\right) }{1-F_{X|Y}\left(\x|1\right)},\frac{f_{X|Y}\left(\x|0\right) }{1-F_{X|Y}\left(\x|0\right)}\right\}\\
& = \frac{2}{\x}\max\{\alpha^\pl,\alpha^\po\}(1+o(1)),
\end{align*}
as $\x\rightarrow\infty$. The second line plugs in the expressions of $f(x)$ and $F(x)$. The third line is by the properties of weighted averages. The last line is by the RV approximation of the pdf in \eqref{eq:pdf}. Therefore, the tail average of the partial effects exists, and converges to 0 as $\x\rightarrow\infty$.
\end{proof}

%~~~~~~~~~~~~~~~~~~~~~~~~~~~~~~~~~~~~~~~~~~~~~~~~~~~~~~~~~~~~~~~~~~~~~
\bigskip
\noindent\begin{proof}[Proof of Theorem \ref{thm:c-s}]

\noindent This proof builds on \citet{WangTsai2009}. Recall that the log likelihood function is given by
\begin{align*}
\ell_N^\py\left( \theta^\py \right) &=\sumi \left( \log Z_i'\theta^\py  - Z_i'\theta^\py \log  \frac{X_i}{\x_N^\py} \right) \1\left\{ \Xi_{N,i}^\py\right\},
\end{align*}
and its first and second derivatives, i.e., the score and Hessian, are 
\begin{align}
S_N^\py\left( \theta^\py \right)&=\frac{\partial \ell_N^\py\left( \theta^\py \right) }{\partial \theta^\py } =\sumi \left( \frac 1 {Z_i'\theta^\py}-\log\frac{X_i}{\x_N^\py} \right) Z_i\1\left\{ \Xi_{N,i}^\py\right\}, \label{eq:c-s-score}\\
H_N^\py\left( \theta^\py \right)&=\frac{\partial^2 \ell_N^\py\left( \theta^\py \right) }{\partial \theta^\py \partial \theta^{\py\prime}} =-\sumi \frac{Z_iZ_i'}{\left(Z_i'\theta^\py\right)^2} \1\left\{ \Xi_{N,i}^\py\right\}.\label{eq:c-s-hessian}
\end{align}

Also note that as $0<\xi_N<1$, Assumption \ref{assn:c-s-est}(b) implies that $\x^\py_N\rightarrow\infty$ as $N\rightarrow\infty$. Combined with the tail conditions in Assumption \ref{assn:c-s-tail}(a,b), this further leads to that for $x\ge\x_N^\py$, 
\begin{align}
&C^\py\left( Z_i\right)(x)^{-\alpha^\py\left( Z_i\right) }\left( 1+D^\py\left( Z_i\right)(x)^{-\beta^\py\left( Z_i\right) }+r^\py_i\left( x,Z_i\right) \right)\label{eq:c-s-approx-C}\\
&=C^\py\left( Z_i\right)(x)^{-\alpha^\py\left( Z_i\right) }\left(1+o(1)\right),\notag\\
&D^\py\left( Z_i\right)(x)^{-\beta^\py\left( Z_i\right) }+r^\py_i\left( x,Z_i\right)=D^\py\left( Z_i\right)(x)^{-\beta^\py\left( Z_i\right) }\left(1+o(1)\right), \label{eq:c-s-approx-D}
\end{align}
as $N\rightarrow\infty$, almost surely.

Note that $N^\py$ is a random variable, so we introduce \(\xi^\py_N\) in \eqref{eq:c-s-xi-def}, a non-random sequence representing the asymptotic proportion of tail observations for $Y_i=y$. Following from \eqref{eq:c-s-approx-C}, 
\begin{align}
    \P\left(\Xi^\py_{N,i}\right) &=\P(X_i>\x^\py_N|Y_i=y) \P(Y_i=y)\label{eq:c-s-xi-p}\\
        &= \E\left[C^\py(Z_i)\left(\x^\py_N\right)^{-\alpha^\py(Z_i)}\right]\left(1 + o(1)\right)\cdot\P(Y_i=y)\notag\\
    &= \xi_N^\py\left(1 + o(1)\right).\notag
    \end{align}

%~~~~~~~~~~~~~~~~~~~~~~~~~~~~~~~~~~~~~~~~~~~~~~~~~~~~~~~~~~~~~~~~~~~~~
\bigskip
\noindent\textbf{Part 1.}
Let 
\begin{align}
S_{N,i}^\py=\frac{1}{\sqrt{\xi^\py_N}} \left(H_{N0}^\py\right)^{-1/2}\left( \frac 1 {Z_i'\theta^\py_0}-\log\frac{X_i}{\x_N^\py} \right) Z_i\1\left\{ \Xi_{N,i}^\py\right\},\label{eq:c-s-si-def}
\end{align} 
which is i.i.d.\ across $i$ by Assumption \ref{assn:c-s-model}(a). Then, according to the score defined in \eqref{eq:c-s-score}, \[\frac{1}{\sqrt{\xi^\py_N}}\left(H_{N0}^\py\right)^{-1/2}S_N^\py\left(\theta^\py_0\right)=\sumi S_{N,i}^\py,\] and we can apply the central limit theorem (CLT) to obtain its asymptotic normality.

First, for the mean, 
\[
\E\left[S_{N,i}^\py\right] =\frac{1}{\sqrt{\xi^\py_N}}\left(H_{N0}^\py\right)^{-1/2}\E\left[ \left( \frac 1 {Z_i'\theta^\py_0}-\log\frac{X_i}{\x_N^\py} \right) Z_i\1\left\{ \Xi_{N,i}^\py\right\}\right].
\]
For the first term in the expectation,
\begin{align*}
\E\left[\frac 1 {Z_i'\theta^\py_0}Z_i\1\left\{ \Xi_{N,i}^\py\right\}\right]
&=\E\left[\left.\frac 1 {Z_i'\theta^\py_0}Z_i\P\left( \Xi_{N,i}^\py\right|Z_i\right)\right]  \\
&=\E\left[\frac 1 {Z_i'\theta^\py_0}Z_i C^\py\left( Z_i\right) \left(\x_N^\py\right)^{-\alpha^\py\left( Z_i\right) }\right]\\
& \quad+ \E\left[\frac 1 {Z_i'\theta^\py_0}Z_i C^\py\left( Z_i\right) D^\py\left( Z_i\right) \left(\x_N^\py\right)^{-\alpha^\py\left( Z_i\right)-\beta^\py\left( Z_i\right) }\right]\left(1+o(1)\right),
\end{align*}
where we substitute the tail approximation in Assumption \ref{assn:c-s-tail}(a) and apply the supremum condition of the remainder term in Assumption \ref{assn:c-s-tail}(b) as in equation \eqref{eq:c-s-approx-D}. 
For the second term in the expectation,
\begin{align*}
&\E\left[\log\frac{X_i}{\x_N^\py}Z_i\1\left\{ \Xi_{N,i}^\py\right\}\right]\\
&=\E\left[\left.Z_i\int_0^{\infty}\P\left(\log\frac{X_i}{\x_N^\py}>s\right|Z_i\right)ds\right]\\
&=\E\left[\left.Z_i\int_0^{\infty}\P\left(X_i>\x_N^\py e^s\right|Z_i\right)ds\right] \\
&=\E\left[Z_i C^\py\left( Z_i\right) \left(\x_N^\py\right)^{-\alpha^\py\left( Z_i\right) }\int_0^\infty e^{-s\alpha^\py\left( Z_i\right)}ds\right]\\
&\quad+ \E\left[Z_i C^\py\left( Z_i\right) D^\py\left( Z_i\right) \left(\x_N^\py\right)^{-\alpha^\py\left( Z_i\right)-\beta^\py\left( Z_i\right) }\int_0^\infty e^{-s\left(\alpha^\py\left( Z_i\right)+\beta^\py\left( Z_i\right)\right)}ds\right]\left(1+o(1)\right)\\
&=\E\left[Z_i C^\py\left( Z_i\right) \left(\x_N^\py\right)^{-\alpha^\py\left( Z_i\right) }\cdot\frac 1 {\alpha^\py\left( Z_i\right)}\right]\\
&\quad+ \E\left[Z_i C^\py\left( Z_i\right) D^\py\left( Z_i\right) \left(\x_N^\py\right)^{-\alpha^\py\left( Z_i\right)-\beta^\py\left( Z_i\right) }\cdot\frac 1 {\alpha^\py\left( Z_i\right)+\beta^\py\left( Z_i\right)}\right]\left(1+o(1)\right),
\end{align*}
The first equality follows from that for a generic random variable $X>0$, $\E[X]=\int_0^\infty\P(X>s)ds$, which is given by integration by parts. The third equality is again by the tail approximation in Assumption \ref{assn:c-s-tail}(a,b) and equation \eqref{eq:c-s-approx-D}. For both terms, we re-derive and fix minor typos regarding the $o(1)$ terms in equations (B.1) and (B.2) in \citet{WangTsai2009}.
Plugging both terms back into the mean expression, and noting that $\alpha^\py\left( Z_i\right)=Z_i'\theta^\py_0$, we have that as $N\rightarrow\infty$,
\begin{align}
&\E\left[S_{N,i}^\py\right]\notag\\ 
&=\frac{1}{\sqrt{\xi^\py_N}}\left(H_{N0}^\py\right)^{-1/2}\E\left[Z_i \frac {\beta^\py\left( Z_i\right)C^\py\left( Z_i\right) D^\py\left( Z_i\right)} {\alpha^\py\left( Z_i\right)\left(\alpha^\py\left( Z_i\right)+\beta^\py\left( Z_i\right)\right)}\left(\x_N^\py\right)^{-\alpha^\py\left( Z_i\right)-\beta^\py\left( Z_i\right) }\right]\left(1+o(1)\right)\notag\\ 
&=o\left(\frac 1{\sqrt N}\right),\label{eq:c-s-si-mean}
\end{align}
where the last line is by Assumption \ref{assn:c-s-est}(b).

Second, for the variance,
\begin{align*}
\V\left[S_{N,i}^\py\right] &=\E\left[S_{N,i}^\py S_{N,i}^{\py\prime}\right]-\E\left[S_{N,i}^\py \right] \E\left[S_{N,i}^\py \right]' \\
&= \frac{1}{\xi^\py_N}\left(H_{N0}^\py\right)^{-1/2}\E\left[ \left( \frac 1 {Z_i'\theta^\py_0}-\log\frac{X_i}{\x_N^\py} \right)^2 Z_iZ_i'\1\left\{ \Xi_{N,i}^\py\right\}\right]\left(H_{N0}^\py\right)^{-1/2} +o\left(\frac 1 N\right).
\end{align*}
In the second equality, the first term is by the definition of $S_{N,i}^\py$ in \eqref{eq:c-s-si-def}, and the second term is by $\E\left[S_{N,i}^\py\right]=o\left(\frac 1{\sqrt N}\right)$ in \eqref{eq:c-s-si-mean}. For the expectation term, by the law of total expectation,
\begin{align*}
& \E\left[ \left( \frac 1 {Z_i'\theta^\py_0}-\log\frac{X_i}{\x_N^\py} \right)^2 Z_iZ_i'\1\left\{ \Xi_{N,i}^\py\right\}\right]\\
&= \E\left[ \E\left[\left. \left( 1-Z_i'\theta^\py_0\log\frac{X_i}{\x_N^\py} \right)^2\right|Z_i,\Xi_{N,i}^\py\right] \frac 1 {\left(Z_i'\theta^\py_0\right)^2}Z_iZ_i'\1\left\{ \Xi_{N,i}^\py\right\}\right]. 
\end{align*}
Let us first consider the inner expectation $\E\left[\left. \left( 1-Z_i'\theta^\py_0\log\frac{X_i}{\x_N^\py} \right)^2\right|Z_i,\Xi_{N,i}^\py\right]$. Denote $\tilde X_i=Z_i'\theta^\py_0\log\frac{X_i}{\x_N^\py}$. Then, by the tail approximation in Assumption \ref{assn:c-s-tail}(a,b) and equation \eqref{eq:c-s-approx-C}, we have that $1-F_{\left.\tilde X_i\right|Z_i,\Xi_{N,i}^\py}\left(\tilde x\left|Z_i,\Xi_{N,i}^\py\right.\right)=\exp(-\tilde x)\left(1 + o(1)\right)$ for $\tilde x>0$, almost surely. Intuitively, $\tilde X_i$ approximately follows a standard exponential distribution given $\left\{Z_i,\Xi_{N,i}^\py\right\}$. Then, following from integration by parts, 
\begin{align*}
\E\left[\tilde X_i\left|Z_i,\Xi_{N,i}^\py\right.\right] &=\int_0^\infty 1-F_{\left.\tilde X_i\right|Z_i,\Xi_{N,i}^\py}\left(s\left|Z_i,\Xi_{N,i}^\py\right.\right)ds\\ 
&=\int_0^\infty \exp(-s)\left(1 + o(1)\right)ds=1 + o(1),\\
\E\left[\tilde X_i^2\left|Z_i,\Xi_{N,i}^\py\right.\right] &=\int_0^\infty 2s \left[1-F_{\left.\tilde X_i\right|Z_i,\Xi_{N,i}^\py}\left(s\left|Z_i,\Xi_{N,i}^\py\right.\right)\right]ds\\ 
&=\int_0^\infty 2s \exp(-s)\left(1 + o(1)\right)ds=2\left(1 + o(1)\right).
\end{align*}
Therefore, we have that
\begin{align*}
\E\left[\left. \left( 1-Z_i'\theta^\py_0\log\frac{X_i}{\x_N^\py} \right)^2\right|Z_i,\Xi_{N,i}^\py\right]=\E\left[\left(1-\tilde X_i\right)^2\left|Z_i,\Xi_{N,i}^\py\right.\right]=1+o(1),
\end{align*}
and thus
\begin{align*}
\E\left[ \left( \frac 1 {Z_i'\theta^\py_0}-\log\frac{X_i}{\x_N^\py} \right)^2 Z_iZ_i'\1\left\{ \Xi_{N,i}^\py\right\}\right]
&= \E\left[  \frac 1 {\left(Z_i'\theta^\py_0\right)^2}Z_iZ_i'\1\left\{ \Xi_{N,i}^\py\right\}\right]\left(1 + o(1)\right)\\
&= \E\left[ \left. \frac 1 {\left(Z_i'\theta^\py_0\right)^2}Z_iZ_i'\right| \Xi_{N,i}^\py\right]\P\left(\Xi_{N,i}^\py\right)\left(1 + o(1)\right)\\
&= H_{N0}^\py \xi^\py_N\left(1 + o_p(1)\right), 
\end{align*}
where the second equality is by the definition of $H_{N0}^\py$ in Assumption \ref{assn:c-s-est}(c) and the approximation to $\P\left(\Xi_{N,i}^\py\right)$ in \eqref{eq:c-s-xi-p}. Substituting this back to the expression of $\V\left[S_{N,i}^\py\right]$, we have that as $N\rightarrow\infty$,
\begin{align}
&\V\left[S_{N,i}^\py\right] =\I_{d_Z} \left(1 + o(1)\right)+o\left(\frac 1 N\right)\rightarrow\I_{d_Z}. \label{eq:c-s-si-var}
\end{align}

Therefore, as $\theta^\py_0\in \intr\left(\Theta^\py\right)$, by the CLT, we have that as $N\rightarrow\infty$,
\begin{align}
\frac{1}{\sqrt{N\xi^\py_N}}\left(H_{N0}^\py\right)^{-1/2}S_N^\py\left(\theta^\py_0\right)=\frac{1}{\sqrt{N}}\sumi S_{N,i}^\py \overset{d}{\rightarrow }\mathcal{N}\left( 0,\I_{d_Z} \right), \label{eq:c-s-si-clt}
\end{align}
where the mean and variance of $S_{N,i}^\py$ are given in \eqref{eq:c-s-si-mean} and \eqref{eq:c-s-si-var}, respectively.

%~~~~~~~~~~~~~~~~~~~~~~~~~~~~~~~~~~~~~~~~~~~~~~~~~~~~~~~~~~~~~~~~~~~~~
\bigskip
\noindent\textbf{Part 2.}
Similarly, for the Hessian matrix,
\begin{align}
&\frac{1}{N\xi^\py_N}\left(H_{N0}^\py\right)^{-1/2}\cdot H_N^\py\left( \theta^\py_0 \right)\cdot\left(H_{N0}^\py\right)^{-1/2}\label{eq:c-s-hi-lln}\\ 
&=\frac{1}{N\xi^\py_N}\left(H_{N0}^\py\right)^{-1/2}\cdot\left[ -\sumi \frac{Z_iZ_i'}{\left(Z_i'\theta^\py_0\right)^2} \1\left\{ \Xi_{N,i}^\py\right\}\right]\cdot\left(H_{N0}^\py\right)^{-1/2}\notag\\ 
&=-\frac{1}{\xi^\py_N}\left(H_{N0}^\py\right)^{-1/2}\cdot\E\left[ \frac{Z_iZ_i'}{\left(Z_i'\theta^\py_0\right)^2} \1\left\{ \Xi_{N,i}^\py\right\}\right]\cdot\left(H_{N0}^\py\right)^{-1/2}\cdot \left(1 + o_p(1)\right)\notag\\  
&=-\frac{1}{\xi^\py_N}\left(H_{N0}^\py\right)^{-1/2}\cdot\E\left[ 
    \left.\frac{Z_iZ_i'}{\left(Z_i'\theta^\py_0\right)^2} \right|\Xi_{N,i}^\py\right]\P\left(\Xi_{N,i}^\py\right)\cdot\left(H_{N0}^\py\right)^{-1/2}\cdot \left(1 + o_p(1)\right)\notag\\  
&=-\frac{1}{\xi^\py_N}\left(H_{N0}^\py\right)^{-1/2}\cdot H_{N0}^\py\xi^\py_N\cdot\left(H_{N0}^\py\right)^{-1/2}\cdot \left(1 + o_p(1)\right)\notag\\ 
&\overset{p}{\rightarrow }-\I_{d_Z}.\notag
\end{align}
The first equality is by the definition of the Hessian matrix \eqref{eq:c-s-hessian}. The second equality follows from the law of large numbers (LLN), where the finite mean is given by Assumption \ref{assn:c-s-est}(c). The fourth equality is by the definition of $H_{N0}^\py$ in Assumption \ref{assn:c-s-est}(c) and the approximation to $\P\left(\Xi_{N,i}^\py\right)$ in \eqref{eq:c-s-xi-p}.

%~~~~~~~~~~~~~~~~~~~~~~~~~~~~~~~~~~~~~~~~~~~~~~~~~~~~~~~~~~~~~~~~~~~~~
\bigskip
\noindent\textbf{Part 3.}
%This part of the proof is a modified version of the proof of Theorem 4 in \citet{WangTsai2009}.
For $y\in\{0,1\}$, let $\zeta_N=\left(H_{N0}^\py\right)^{1/2}\left(\theta^\py-\theta^\py_0\right)$, $\zeta_{N0}=\left(H_{N0}^\py\right)^{1/2}\theta^\py_0$, and $W_i=\left(H_{N0}^\py\right)^{-1/2}Z_i$. Then, the log likelihood function can be rewritten as
\begin{align*}
\tilde\ell_N^\py\left( \zeta_N \right)
=\sumi \left( \log \left(W_i'\left(\zeta_N + \zeta_{N0}\right)\right) - W_i'\left(\zeta_N + \zeta_{N0}\right) \log  \frac{X_i}{\x_N^\py} \right) \1\left\{ \Xi_{N,i}^\py\right\}, 
\end{align*}
the corresponding score and Hessian are denoted by $\tilde S_N^\py\left( \zeta_N \right)$ and $\tilde H_N^\py\left( \zeta_N \right)$, respectively,
and the MLE estimate is denoted by $\hat\zeta_N$.

First, the new Hessian matrix is given by \[\tilde H_N^\py\left( \zeta_N \right)=-\sumi \frac{W_iW_i'}{\left(W_i'\left(\zeta_N + \zeta_{N0}\right)\right)^2} \1\left\{ \Xi_{N,i}^\py\right\}.\] It is positive definite for all $\zeta_N=\left(H_{N0}^\py\right)^{1/2}\left(\theta^\py-\theta^\py_0\right)$ with $\theta^\py\in\Theta^\py$, as $\Theta^\py$ is a convex cone such that for all $\theta^\py\in\Theta^\py$,  $Z_i'\theta^\py=W_i'\left(\zeta_N + \zeta_{N0}\right)>0$ almost surely: see the discussion after Assumption \ref{assn:c-s-tail}. Then, the log likelihood function is $\tilde\ell_N^\py\left( \zeta_N \right)$ is strictly concave over its domain, and the MLE estimate $\hat\zeta_N$ is unique.

Second, let \[\mathcal U^\py_{N,C} = \left\{u\in\mathbb R^{d_Z}:\,\frac{1}{\sqrt{N\xi^\py_N}}\left(H_{N0}^\py\right)^{-1/2}u+\theta^\py_0\in\Theta^\py\text{ and }\|u\|=C\right\}.\] Note that $\mathcal U^\py_{N,C}\neq\emptyset$ for any $C>0$. This follows from three facts: first, $\Theta^\py$ is a convex cone; second, $\intr\left(\Theta^\py\right)\neq\emptyset$ as $\theta^\py_0\in \intr\left(\Theta^\py\right)$; and third, $H_{N0}^\py$ is finite and full rank by Assumption \ref{assn:c-s-est}(c). Then, let $u$ be an arbitrary non-random vector in $\mathcal U^\py_{N,C}$. Applying the second-order Taylor expansion of $\tilde\ell^\py_N\left( \frac{u}{\sqrt{N\xi^\py_N}}\right)$ around $\zeta_N=0$, we have that
\begin{equation*}
\tilde\ell^\py_N\left( \frac{u}{\sqrt{N\xi^\py_N}}\right) -\tilde\ell^\py_N\left( 0\right) = \frac{1}{\sqrt{N\xi^\py_N}}u'\tilde S^\py_N\left( 0\right) + \frac{1}{2N\xi^\py_N}u'\tilde H^\py_N\left( 0\right)u + o_{p}(1)
\end{equation*}
From \eqref{eq:c-s-si-clt} and \eqref{eq:c-s-hi-lln}, we have that as $N\rightarrow\infty$,
\begin{align}
\frac{1}{\sqrt{N\xi^\py_N}}\tilde S_N^\py\left( 0 \right)&=\frac{1}{\sqrt{N\xi^\py_N}}S_N^\py\left(\theta^\py_0\right) \overset{d}{\rightarrow }\mathcal{N}\left( 0,\I_{d_Z} \right),\label{eq:c-s-zeta-clt}\\
\frac{1}{N\xi^\py_N}\tilde H_N^\py\left( 0 \right)
&=\frac{1}{N\xi^\py_N}\left(H_{N0}^\py\right)^{-1/2}\cdot H_N^\py\left( \theta^\py_0 \right)\cdot\left(H_{N0}^\py\right)^{-1/2}\overset{p}{\rightarrow }-\I_{d_Z}. \label{eq:c-s-zeta-lln}
\end{align}
This implies that when $C$ is large enough, the quadratic term dominates the linear one with an arbitrarily large probability. That is, for any $\varepsilon>0$, there exists a $C>0$ such that 
\begin{equation*}
\underset{N}{\lim \sup }\,\P\left( \sup_{u\in \mathcal U^\py_{N,C}}\tilde\ell^\py_N\left( \frac{u}{\sqrt{N\xi^\py_N}}\right) <\tilde\ell^\py_N\left(0\right) \right) >1-\varepsilon.
\end{equation*}
Therefore, $\tilde\ell^\py_N\left( \cdot \right)$ must have at least one local maximizer, which is of order $O_{p}\left( \sqrt{N\xi^\py_N}\right)$. And by the uniqueness of the MLE, this local maximizer is the global maximizer, and $\hat\zeta_N=O_{p}\left( \sqrt{N\xi^\py_N}\right)$.

Finally, by the first order condition, $\tilde S_N^\py\left( \hat\zeta_N\right)=0$. Applying the first-order Taylor expansion around $\zeta_N=0$, we have that 
\[0=\tilde S_N^\py\left( \hat\zeta_N\right)=\tilde S_N^\py\left( 0\right)+\tilde H_N^\py\left( 0\right)\hat\zeta_N+o_{p}(1),\] which implies that 
\[\sqrt{N\xi^\py_N}\hat\zeta_N=-\left(\frac{1}{N\xi^\py_N}\tilde H_N^\py\left( 0\right)\right)^{-1}\cdot\frac{1}{\sqrt{N\xi^\py_N}}\left(\tilde S_N^\py\left( 0\right)+o_{p}(1)\right)\overset{d}{\rightarrow }\mathcal{N}\left( 0,\I_{d_Z}\right),\]
following from \eqref{eq:c-s-zeta-clt} and \eqref{eq:c-s-zeta-lln}. 
This is equivalent to
\begin{align*}
\sqrt{N\xi^\py_N}\left(H_{N0}^\py\right)^{1/2}\left( \hat\theta^\py-\theta^\py_0\right) \overset{d}{\rightarrow }\mathcal{N}\left( 0,\I_{d_Z}\right). 
\end{align*}

In addition, the asymptotic independence between $\hat\theta^\pl$ and $\hat\theta^\po$ follows from i.i.d.\ observations across $i$ in Assumption \ref{assn:c-s-model}(a). This completes the proof of the theorem.
\end{proof}

%~~~~~~~~~~~~~~~~~~~~~~~~~~~~~~~~~~~~~~~~~~~~~~~~~~~~~~~~~~~~~~~~~~~~~
\bigskip
\begin{proposition}[Cross-sectional data: proportion of tail observations]\label{prop:c-s-ny}
Suppose Assumptions \ref{assn:c-s-model}--\ref{assn:c-s-est} hold. For $y\in\{0,1\}$, as $N\rightarrow\infty$, 
\begin{align*}
\frac{N^\py}{N}=\xi_N^\py\left(1+o_p(1)\right).
\end{align*}
\end{proposition}

\noindent\begin{proof}
First, following from Assumption \ref{assn:c-s-tail}(a,b) and equation \eqref{eq:c-s-xi-p}, we have that
\begin{align}
\E\left[\frac{N^\py}{N}\right] &= \P\left(\Xi^\py_{N,i}\right) = \xi_N^\py\left(1 + o(1)\right).\label{eq:c-s-xi-mean}
\end{align}
Second, the variance of $N^\py/N$ is bounded by
\begin{align}
\V\left[\frac{N^\py}{N}\right] &= \frac{1}{N}\P\left(\Xi^\py_{N,i}\right)\left(1 - \P\left(\Xi^\py_{N,i}\right)\right)\label{eq:c-s-xi-var} \\
&\leq \frac{1}{N}\P\left(\Xi^\py_{N,i}\right) = \frac{1}{N}\xi_N^\py\left(1 + o(1)\right).\notag
\end{align}
Combining both, we obtain that as $N\rightarrow\infty$,
\begin{align*}
\E\left[\left(\frac{N^\py}{N\xi_N^\py}-1\right)^2\right] &=\V\left[\frac{N^\py}{N\xi_N^\py}\right]+o(1) = \frac{1}{\left(\xi_N^\py\right)^2}\V\left[\frac{N^\py}{N}\right]+o(1)\\
    &\le \frac{1}{N\xi_N^\py}\left(1 + o(1)\right)+o(1)\rightarrow0,
\end{align*}
where the first line is by the expression for the mean \eqref{eq:c-s-xi-mean}, the second line is by the bound on the variance \eqref{eq:c-s-xi-var}, and the last convergence is by $N\xi_N^\py\rightarrow\infty$ in Assumption \ref{assn:c-s-est}(a). This part of the proof is slightly different from \citet{WangTsai2009}: first, we explicitly account for the additive $o(1)$ term; second, as $N^\py$ is a random variable, we let $N\xi_N^\py\rightarrow\infty$ instead.

By Chebyshev's inequality, convergence in the second moment implies convergence in probability. Then, we have that as $N\rightarrow\infty$, $\frac{N^\py}{N\xi_N^\py}\overset{p}{\rightarrow}1$, and equivalently, $\frac{N^\py}{N}=\xi_N^\py\{1 + o_p(1)\}$.

\hfill\end{proof}

%~~~~~~~~~~~~~~~~~~~~~~~~~~~~~~~~~~~~~~~~~~~~~~~~~~~~~~~~~~~~~~~~~~~~~
\begin{remark} 
\normalfont{
Combining Theorem \ref{thm:c-s} with Proposition \ref{prop:c-s-ny}, we have that as $N\rightarrow\infty$,
\begin{align*}
\sqrt{N^\py}\left(H_{N0}^\py\right)^{1/2}\left( \hat\theta^\py-\theta^\py_0\right) \overset{d}{\rightarrow }\mathcal{N}\left( 0,\I_{d_Z}\right). 
\end{align*}
}
\end{remark}

%~~~~~~~~~~~~~~~~~~~~~~~~~~~~~~~~~~~~~~~~~~~~~~~~~~~~~~~~~~~~~~~~~~~~~
\subsection{Proofs for Section \ref{sec:panel}}
We first present two lemmas for the LLN and the CLT with i.n.i.d.\ observations under rare events, and then provide the proof of our main result in Section \ref{sec:panel}: Theorem \ref{thm:panel}.

Let us consider a generic setup with i.n.i.d.\ observations under rare events.  Let $\{X_{N,i} : 1 \leq i \leq N, N \geq 1\}$ be a triangular array of i.n.i.d.\ $d_X$-dimensional random vectors, and $\{\Xi_{N,i} : 1 \leq i \leq N, N \geq 1\}$ be a triangular array of events with $\P_i(\Xi_{N,i}) = p_{N,i}\rightarrow0$ as $N\rightarrow\infty$. Subsript $i$ denotes quantities given potential individual-specific parameters. Also define $p_N = \avgi p_{N,i}$.
\begin{lemma}[Rare event LLN for i.n.i.d.\ observations] \label{lem:inid-lln}
    Suppose that:
\begin{enumerate}[label=(\alph*)]
    \item $Np_N \to \infty$, as $N \to \infty$.
\item $\E_i\left[\|X_{N,i}\|^2 | \Xi_{N,i}\right] \leq M < \infty$, for sufficiently large $N$ and all $i$. % E|X|^2 = tr(Var(X)) + |E(X)|^2
\end{enumerate}
Then, $\E_i\left[X_{N,i} | \Xi_{N,i}\right] = \mu_{N,i}$ exists, and as $N \to \infty$,
\[
\frac{1}{Np_N} \sumi X_{N,i} \1\{\Xi_{N,i}\} - \frac{1}{Np_N} \sumi p_{N,i}\mu_{N,i} \xrightarrow{p} 0.
\]
\end{lemma}    

\noindent\begin{proof}
Let $S_N = \sumi \left(X_{N,i} \1\{\Xi_{N,i}\} - p_{N,i}\mu_{N,i}\right)$. We want to show that $S_N / (Np_N) \xrightarrow{p} 0$.

By the multivariate Chebyshev's inequality, for any $\varepsilon > 0$,
\[
\P\left(\frac{\left\|S_N\right\|}{Np_N} > \varepsilon\right) \leq \frac{\tr\left(\V\left[S_N\right]\right)}{\varepsilon^2 \left(Np_N\right)^2}.
\]
Now, let's compute the trace of $\V\left[S_N\right]$:
\begin{align*}
\tr\left(\V\left[S_N\right]\right) &= \sumi \tr\left(\V_i\left[X_{N,i}\1\{\Xi_{N,i}\}\right]\right) \\
&\le\sumi\E_i\left[\|X_{N,i}\|^2 \1\{\Xi_{N,i}\}\right]  = \sumi\E_i\left[\left.\|X_{N,i}\|^2 \right|\Xi_{N,i}\right]\P\left(\Xi_{N,i}\right) \\
&\le M\sumi p_{N,i} = MNp_N,
\end{align*}
where the first line is by independence, the second line is by the definition of variance, and the third line is by condition (b).
Substituting this back into the multivariate Chebyshev's inequality:
\[
\P\left(\frac{\left\|S_N\right\|}{Np_N} > \varepsilon\right) \leq \frac{M}{\varepsilon^2 Np_N}\rightarrow0,
\]
as $N \to \infty$, following from condition (a).

Therefore, $S_N / (Np_N) \xrightarrow{p} 0$, which is equivalent to
\[
\frac{1}{Np_N} \sumi X_{N,i} \1\{\Xi_{N,i}\} - \frac{1}{Np_N} \sumi p_{N,i}\mu_{N,i} \xrightarrow{p} 0.
\]
\hfill\end{proof}

%~~~~~~~~~~~~~~~~~~~~~~~~~~~~~~~~~~~~~~~~~~~~~~~~~~~~~~~~~~~~~~~~~~~~~
\begin{lemma}[Rare event CLT for i.n.i.d.\ observations]\label{lem:inid-clt}
Suppose that:
\begin{enumerate}[label=(\alph*)] 
\item $Np_N \to \infty$, as $N \to \infty$.
\item $\E_i\left[\left.\|X_{N,i}\|^{2+\kappa} \right| \Xi_{N,i}\right] \leq M < \infty$ for some $\kappa > 0$, for sufficiently large $N$ and all $i$.
\item Let $\V_i\left[\left.X_{N,i} \right| \Xi_{N,i}\right] = \Sigma_{N,i}$. Its smallest eigenvalue $\lambda_{\min}\left(\Sigma_{N,i}\right) \geq \underline{\sigma}^2_N$, where $\underline{\sigma}^2_N\left(Np_N\right)^{\kappa/\left(2+\kappa\right)}\to\infty$ as $N\to\infty$, for all $i$.
\end{enumerate}
Let $\E_i\left[X_{N,i} | \Xi_{N,i}\right] = \mu_{N,i}$, $S_N = \sumi \left(X_{N,i} \1\{\Xi_{N,i}\}-p_{N,i}\mu_{N,i}\right)$, and $\Sigma_N = \sumi p_{N,i}\Sigma_{N,i}$. Then, as $N \to \infty$,
\[
\Sigma_N^{-1/2} S_N \xrightarrow{d} N\left(0, \I_{d_X}\right).
\]
\end{lemma}
\begin{proof}
Let $X^*_{N,i} = X_{N,i} \1\{\Xi_{N,i}\}-p_{N,i}\mu_{N,i}$. So we have $\E_i\left[X^*_{N,i}\right] = 0$ and $\V_i\left[X^*_{N,i}\right] = p_{N,i}\Sigma_{N,i}$.
We'll use the Lindeberg-Feller CLT for triangular arrays. We need to verify two conditions: 
\begin{enumerate}
\item Lindeberg condition: for any $\varepsilon > 0$, as $N \to \infty$,
\[
\frac{1}{\tr\left(\Sigma_N\right)} \sumi \E_i\left[\left\|X^*_{N,i}\right\|^2 \1\left\{\left\|X^*_{N,i}\right\| > \varepsilon \sqrt{\tr\left(\Sigma_N\right)}\right\}\right] \to 0.
\]
\item Variance condition: as $N \to \infty$,
\[
\frac{\max_{1 \leq i \leq N} \tr\left(\V_i\left[X^*_{N,i}\right]\right)}{\tr\left(\Sigma_N\right)} \to 0.
\]
\end{enumerate}

First, for the Lindeberg condition,
\begin{align*}
\E_i\left[\left\|X^*_{N,i}\right\|^2 \1\left\{\left\|X^*_{N,i}\right\| > \varepsilon \sqrt{\tr\left(\Sigma_N\right)}\right\}\right]
&\le\E_i\left[\left\|X_{N,i}\right\|^2 \1\{\Xi_{N,i}\}\1\left\{\left\|X^*_{N,i}\right\| > \varepsilon \sqrt{\tr\left(\Sigma_N\right)}\right\}\right] \\
&\leq \frac{\E_i\left[\left\|X_{N,i}\right\|^{2+\kappa} \1\{\Xi_{N,i}\}\right]}{\left(\varepsilon\sqrt{\tr\left(\Sigma_N\right)}\right)^{\kappa}} \\
&\leq \frac{Mp_{N,i}}{\left(\varepsilon\sqrt{\tr\left(\Sigma_N\right)}\right)^{\kappa}}, 
\end{align*}
where the first line is by the definition of $X^*_{N,i}$, the second line is by the indicator function on $\left\|X^*_{N,i}\right\| > \varepsilon \sqrt{\tr\left(\Sigma_N\right)}$, and the third line is by condition (b).
Also, we can obtain a lower bound of $\tr\left(\Sigma_N\right)$ in the denominator based on condition (c):
\begin{align}
\tr\left(\Sigma_N\right) = \sumi p_{N,i}\tr\left(\Sigma_{N,i}\right) \geq \sumi p_{N,i}\cdot\min_{1 \leq i \leq N} \tr\left(\Sigma_{N,i}\right) =  Np_N\cdot d_X\underline{\sigma}^2_N.\label{eq:inid-clt-denom}
\end{align}
Therefore, we have that
\begin{align*}
&\frac{1}{\tr\left(\Sigma_N\right)} \sumi \E_i\left[\left\|X^*_{N,i}\right\|^2 \1\left\{\left\|X^*_{N,i}\right\| > \varepsilon \sqrt{\tr\left(\Sigma_N\right)}\right\}\right] \\
&\leq \frac{MNp_N}{\varepsilon^{\kappa} \left(\tr\left(\Sigma_N\right)\right)^{1+\kappa/2}} \le \frac{MNp_N}{\varepsilon^{\kappa}\left(Np_N\cdot d_X\underline\sigma_N^2\right)^{1+\kappa/2}}\\ 
&=\frac{MNp_N}{\varepsilon^{\kappa}d_X^{1+\kappa/2}\left(Np_N\right)^{\kappa/2}\left(\underline\sigma_N^2\right)^{1+\kappa/2}}\to 0,
\end{align*}
where the second inequality is by \eqref{eq:inid-clt-denom} and the convergence to 0 is by condition (c).

Second, for the variance condition,
\begin{align*}
\frac{\max_{1 \leq i \leq N} \tr\left(\V_i\left[X^*_{N,i}\right]\right)}{\tr\left(\Sigma_N\right)}
&= \frac{\max_{1 \leq i \leq N} p_{N,i}\tr\left(\Sigma_{N,i}\right)}{\sumi p_{N,i}\tr\left(\Sigma_{N,i}\right)}\\
&\leq \frac{\max_{1 \leq i \leq N} p_{N,i}\cdot\max_{1 \leq i \leq N} \tr\left(\Sigma_{N,i}\right)}{Np_N\min_{1 \leq i \leq N} \tr\left(\Sigma_{N,i}\right)}\\
&\leq \frac{M^{2/(2+\kappa)}}{Np_N\cdot d_X\underline{\sigma}^2_N}\rightarrow 0.
\end{align*}
By conditions (b) and (c), $d_X\underline{\sigma}^2_N \leq \tr(\Sigma_{N,i}) \leq M^{2/(2+\kappa)}$ for sufficiently large $N$ and all $i$, so we obtain the inequality in the last line.  Further by conditions (a) and (b), we have that in the denominator, $\underline{\sigma}^2_NNp_N=\underline{\sigma}^2_N\left(Np_N\right)^{\kappa/\left(2+\kappa\right)}\cdot\left(Np_N\right)^{2/\left(2+\kappa\right)}\to\infty$, so the whole term converges to 0.
% Holder's inequality:(E|Y|^p)^(1/p) ≤ (E|Y|^q)^(1/q) for 0 < p < q

With both conditions satisfied, we can apply the Lindeberg-Feller CLT and obtain that as $N \to \infty$,
\[
\Sigma_N^{-1/2}S_N \xrightarrow{d} N(0, \I_{d_X}).
\]
\hfill\end{proof}

%~~~~~~~~~~~~~~~~~~~~~~~~~~~~~~~~~~~~~~~~~~~~~~~~~~~~~~~~~~~~~~~~~~~~~
\bigskip
\noindent\begin{proof}[Proof of Theorem \ref{thm:panel}]

\noindent The structure of the proof is similar to the one for Theorem \ref{thm:c-s} and further incorporates both panel data and i.n.i.d.\ observations. Recall that the log likelihood function is given by
\begin{align*}
\ell_N\left( \theta^* \right) &= \sumi\left[ Z_i'\theta^* (1-Y_{i1})\log\frac{X_{i1}}{X_{i2}}-\log\left(1+\left(\frac{X_{i1}}{X_{i2}}\right)^{Z_i'\theta^* }\right) \right]
\1\{\Xi_{N,i}\},
\end{align*}
and its first and second derivatives, i.e., the score and Hessian, are 
\begin{align}
S_N\left( \theta^* \right)&=\frac{\partial \ell_N\left( \theta^* \right) }{\partial \theta^* } =\sumi \left( \frac{1}{1+\left(\frac{X_{i1}}{X_{i2}}\right)^{Z_i'\theta^* }}-Y_{i1} \right) Z_i\log\frac{X_{i1}}{X_{i2}}\1\left\{ \Xi_{N,i}\right\}, \label{eq:panel-score}\\
H_N\left( \theta^* \right)&=\frac{\partial^2 \ell_N\left( \theta^* \right) }{\partial \theta^* \partial \theta^{*\prime}} =-\sumi \frac{\left(\frac{X_{i1}}{X_{i2}}\right)^{Z_i'\theta^* }}{\left(1+\left(\frac{X_{i1}}{X_{i2}}\right)^{Z_i'\theta^* }\right)^2} Z_iZ_i'\left(\log\frac{X_{i1}}{X_{i2}}\right)^2 \1\left\{ \Xi_{N,i}\right\}.\label{eq:panel-hessian}
\end{align}

Also note that Assumption \ref{assn:panel-est}(a) implies that $\xn\rightarrow\infty$ as $N\rightarrow\infty$. Combined with the tail conditions in Assumption \ref{assn:panel-tail}(a,b), this further leads to that for $x\ge\xn$,
\begin{align}
    &C^\py_i\left( Z_i\right)(x)^{-\tilde\alpha^\py_i\left( Z_i\right) }\left( 1+D^\py_i\left( Z_i\right)(x)^{-\beta^\py_i\left( Z_i\right) }+r^\py_i\left( x,Z_i\right) \right)\label{eq:panel-approx-C}\\
    &=C^\py_i\left( Z_i\right)(x)^{-\tilde\alpha^\py_i\left( Z_i\right) }\left(1+o(1)\right),\notag%\\
    %&D^\py_i\left( Z_i\right)(x)^{-\beta^\py_i\left( Z_i\right) }+r^\py_i\left( x,Z_i\right)=D^\py_i\left( Z_i\right)(x)^{-\beta^\py_i\left( Z_i\right) }\left(1+o(1)\right), \label{eq:panel-approx-D}
\end{align}
as $N\rightarrow\infty$, almost surely in $Z_i$ and for all $i$.

%~~~~~~~~~~~~~~~~~~~~~~~~~~~~~~~~~~~~~~~~~~~~~~~~~~~~~~~~~~~~~~~~~~~~~
\bigskip
\noindent\textbf{Part 1.} 
Note that $N_{\Xi}$ is a random variable, so we introduce \(\xi_{N,i}\) and \(\xi_N\) in \eqref{eq:panel-xi-def}, non-random sequences representing the asymptotic proportions of tail observations with switching outcome values. Let us first link \(\xi_{N,i}\) and \(\xi_N\) to $p_{N,i}=\P_i\left( \Xi_{N,i} \right)$ and $p_N=\avgi p_{N,i}$ in the framework of rare event with i.n.i.d.\ observations.

For each $i=1,\cdots,N$, by Assumption \ref{assn:panel-model} on model specification,
\begin{align}
\P_i\left( \left. \Xi_{N,i}\right| Z_i\right) 
&=\P_i\left( \left. Y_{i1}+Y_{i2}=1,X_{i1}\ge  \x_N,X_{i2}\ge 
\x_N\right| Z_i\right) \notag \\
&=\P_i\left( Y_{i1}=1,Y_{i2}=0,X_{i1}\ge  \x_N,X_{i2}\ge 
\x_N|Z_i\right) \notag \\
&\quad+\P_i\left( Y_{i1}=0,Y_{i2}=1,X_{i1}\ge  \x_N,X_{i2}\ge 
\x_N|Z_i\right) \notag \\
&=\P_i\left( Y_{i1}=1,X_{i1}\ge  \x_N|Z_i\right) \P_i
\left( Y_{i2}=0,X_{i2}\ge  \x_N|Z_i\right) \notag \\
&\quad+\P_i\left( Y_{i1}=0,X_{i1}\ge  \x_N|Z_i\right) \P_i
\left( Y_{i2}=1,X_{i2}\ge  \x_N|Z_i\right). \label{eq:nxi}
\end{align}
By Assumption \ref{assn:panel-tail}(a,b) and equation \eqref{eq:panel-approx-C}, the first term in \eqref{eq:nxi} is given by
\begin{align*}
\P_i\left( Y_{i1}=1,X_{i1}\ge  \x_N|Z_i\right) 
&=\P_i\left( X_{i1}\ge  \x_N|Y_{i1}=1,Z_i\right) \P_i
\left( Y_{i1}=1|Z_i\right) \\
&=\x_N^{-\tilde\alpha^\pl_i\left( Z_i\right) }C^\pl_i\left(
Z_i\right)\left(1+o(1)\right)\cdot\P_i
\left( Y_{i1}=1|Z_i\right),
\end{align*}
as $\x_N\rightarrow \infty,$ for all $i$. Similar expressions hold for the other three terms in \eqref{eq:nxi}, and it follows that as $N\rightarrow \infty, $
\begin{align}
    p_{N,i}&=\P_i\left( \Xi_{N,i} \right) =\E_i\left[\P_i\left( \left. \Xi_{N,i}\right| Z_i\right) \right]\label{eq:panel-pi}\\
&=2\E_i\left[\prod_{y\in\{0,1\}}\x_N^{-\tilde\alpha^\py_i\left( Z_i\right) }C^\py_i\left( Z_i\right)\tilde\L^\py_i \left( \x_N,Z_i\right) \P_i(Y_{it}=y|Z_i)\right]\notag \\
&=2\E_i\left[\prod_{y\in\{0,1\}}\x_N^{-\tilde\alpha^\py_i\left( Z_i\right) }C^\py_i\left( Z_i\right) \P_i(Y_{it}=y|Z_i)\right]\left(1+o(1)\right)\notag\\
&=\xi_{N,i}\left(1+o(1)\right).\notag
\end{align}
Accordingly, by the definition of $\xi_N$ in \eqref{eq:panel-xi-def}, 
\begin{align}
    p_N=\avgi  p_{N,i} &=\avgi \xi_{N,i}\left(1+o(1)\right)=\xi_N\left(1+o(1)\right).\label{eq:panel-p}
\end{align}
Therefore, $Np_N\rightarrow\infty$ as $N\rightarrow\infty$ by Assumption \ref{assn:panel-est}(a), so condition (a) in Lemmas \ref{lem:inid-lln} and \ref{lem:inid-clt} (rare event LLN and CLT for i.n.i.d.\ observations) is satisfied.

%~~~~~~~~~~~~~~~~~~~~~~~~~~~~~~~~~~~~~~~~~~~~~~~~~~~~~~~~~~~~~~~~~~~~~
\bigskip
\noindent\textbf{Part 2.} 
In this part, we derive the asymptotic expression for \[\P_i\left(\left.Y_{i1}=1\right|X_{i1}=x_1,X_{i2}=x_2,Z_i=z,\Xi_{N,i}\right).\]

First, by the tail approximation in Assumption \ref{assn:panel-tail}(a,b), for $i=1,\cdots,N$, for $y\in\{0,1\}$, and for $\P_{Z_i}$-almost all $z\in supp(Z_i)$, the pdf is given by
\begin{align}
f_{i,X_{it}|Y_{it},Z_i}\left( x|y,z\right)&=\frac{\partial\left(1-F_{i,X_{it}|Y_{it},Z_i}\left( x|y,z\right)\right)}{\partial x}\label{eq:panel-f}\\
&= C^\py_i\left( z\right)\tilde\alpha_i^\py\left( z\right)
x^{-\tilde\alpha_i^\py\left( z\right)-1 }\left( 1+B_i^\py\left(x, z\right)\right),\notag
\end{align}
where 
\begin{align}
    &B_i^\py\left(x, z\right) \label{eq:panel-B}\\
    &=x^{-\beta^\py_i\left( z\right)}\left[D_i^\py\left( z\right)\left(1+\frac{\beta^\py_i\left( z\right)}{\tilde\alpha_i^\py\left( z\right)}\right)+x^{\beta^\py_i\left( z\right)}r^\py_i\left(x,z\right)-\frac 1 {\tilde\alpha_i^\py\left( z\right)}x^{\beta^\py_i\left( z\right)+1}\frac{\partial r^\py_i\left(x,z\right)}{\partial x} \right]\notag\\
    &= O\left(x^{-\beta^\py_i\left( z\right)}\right).\notag
\end{align}

Second, by the Bayes' theorem representation, for $x\ge\xn$,
\begin{align}
&\P_i\left(\left.Y_{it}=1\right|X_{it}=x,Z_i=z,X_{it}\ge\xn\right)=\P_i\left(\left.Y_{it}=1\right|X_{it}=x,Z_i=z\right)\label{eq:panel-PY}\\
&=\frac{f_{i,X_{it}|Y_{it},Z_i}\left( x|1,z\right)\P_i\left(\left.Y_{it}=1\right|Z_i=z\right)}{\sum_{y\in\{0,1\}}f_{i,X_{it}|Y_{it},Z_i}\left( x|y,z\right)\P_i\left(\left.Y_{it}=y\right|Z_i=z\right)}\notag\\
&=\frac 1 {1+\frac{f_{i,X_{it}|Y_{it},Z_i}\left( x|0,z\right)\P_i\left(\left.Y_{it}=0\right|Z_i=z\right)}{f_{i,X_{it}|Y_{it},Z_i}\left( x|1,z\right)\P_i\left(\left.Y_{it}=1\right|Z_i=z\right)}}\notag\\ 
&=\frac 1 {1+\frac{C^\po_i\left( z\right)\tilde\alpha_i^\po\left( z\right)x^{-\tilde\alpha_i^\po\left( z\right)-1 }\left( 1+B_i^\po\left(x, z\right)\right)\P_i\left(\left.Y_{it}=0\right|Z_i=z\right)}{C^\pl_i\left( z\right)\tilde\alpha_i^\pl\left( z\right)x^{-\tilde\alpha_i^\pl\left( z\right)-1 }\left( 1+B_i^\pl\left(x, z\right)\right)\P_i\left(\left.Y_{it}=1\right|Z_i=z\right)}}\notag\\ 
&=\frac 1 {1+A_i(z)x^{\tilde\alpha_i^\pl\left( z\right)-\tilde\alpha_i^\po\left( z\right)}}\left( 1+O\left(x^{-\beta^\pl_i\left( z\right)}\right)+O\left(x^{-\beta^\po_i\left( z\right)}\right)\right)\notag\\ 
&=\frac 1 {1+A_i(z)x^{z'\theta^*_0}}\left( 1+O\left(x^{-\beta^{\min}_i\left( z\right)}\right)\right),\notag
\end{align}
where \(A_i(z)=\frac{C^\po_i\left( z\right)\tilde\alpha_i^\po\left( z\right)\P_i\left(\left.Y_{it}=0\right|Z_i=z\right)}{C^\pl_i\left( z\right)\tilde\alpha_i^\pl\left( z\right)\P_i\left(\left.Y_{it}=1\right|Z_i=z\right)}\) and $\beta^{\min}_i\left( z\right)=\min\left\{\beta^\pl_i\left( z\right),\beta^\po_i\left( z\right)\right\}$.
The third and fourth lines in \eqref{eq:panel-PY} follow from the tail approximation in Assumption \ref{assn:panel-tail}(a,b) and equations \eqref{eq:panel-f} and \eqref{eq:panel-B}, and the last line is by additive form of the tail index in \eqref{eq:heter_alpha_panel}.

Finally, we cancel out $A_i(z)$ by conditioning on event $Y_{i1}+Y_{i2}=1$. By the model specification in Assumption \ref{assn:panel-model}, 
\begin{align}
&\P_i\left(\left.Y_{i1}=1\right|X_{i1}=x_1,X_{i2}=x_2,Z_i=z,\Xi_{N,i}\right)\label{eq:panel-PY-cond}\\ 
&=\frac{\P_i\left(\left.Y_{i1}=1\right|X_{i1}=x,Z_i=z,X_{i1}\ge\xn\right)\P_i\left(\left.Y_{i2}=0\right|X_{i2}=x,Z_i=z,X_{i2}\ge\xn\right)}{\sum_{y\in\{0,1\}}\P_i\left(\left.Y_{i1}=y\right|X_{i1}=x,Z_i=z,X_{i1}\ge\xn\right)\P_i\left(\left.Y_{i2}=1-y\right|X_{i2}=x,Z_i=z,X_{i2}\ge\xn\right)}\notag\\
&=\frac 1 {1+\frac{\P_i\left(\left.Y_{i1}=0\right|X_{i1}=x,Z_i=z,X_{i1}\ge\xn\right)\P_i\left(\left.Y_{i2}=1\right|X_{i2}=x,Z_i=z,X_{i2}\ge\xn\right)}{\P_i\left(\left.Y_{i1}=1\right|X_{i1}=x,Z_i=z,X_{i1}\ge\xn\right)\P_i\left(\left.Y_{i2}=0\right|X_{i2}=x,Z_i=z,X_{i2}\ge\xn\right)}}\notag\\
&=\frac 1 {1+\left(\frac{x_1}{x_2}\right)^{z'\theta^*_0}}\left( 1+O\left(\max\{x_1,x_2\}^{-\beta^{\min}_i\left( z\right)}\right)\right),\notag
\end{align}
where we plug in \eqref{eq:panel-PY} to obtain the last equality. Similarly,
\begin{align}
    &\P_i\left(\left.Y_{i1}=0\right|X_{i1}=x_1,X_{i2}=x_2,Z_i=z,\Xi_{N,i}\right)\label{eq:panel-PY-cond0}\\
    &=\frac {\left(\frac{x_1}{x_2}\right)^{z'\theta^*_0}} {1+\left(\frac{x_1}{x_2}\right)^{z'\theta^*_0}}\left( 1+O\left(\max\{x_1,x_2\}^{-\beta^{\min}_i\left( z\right)}\right)\right).\notag
    \end{align}

%~~~~~~~~~~~~~~~~~~~~~~~~~~~~~~~~~~~~~~~~~~~~~~~~~~~~~~~~~~~~~~~~~~~~~
\bigskip
\noindent\textbf{Part 3.}
Let 
\begin{align}
S_{N,i}= \frac 1 {\sqrt{N\xi_N}}\left(H_{N0}\right)^{-1/2}\left( \frac{1}{1+\left(\frac{X_{i1}}{X_{i2}}\right)^{Z_i'\theta^*_0 }}-Y_{i1} \right) Z_i\log\frac{X_{i1}}{X_{i2}},\label{eq:panel-si-def}
\end{align} 
which is i.n.i.d.\ across $i$ by Assumption \ref{assn:panel-model}(a). Then, according to the score defined in \eqref{eq:panel-score}, 
\[\frac 1 {\sqrt{N\xi_N}}\left(H_{N0}\right)^{-1/2}S_N\left(\theta^*_0\right)=\sumi S_{N,i}\1\left\{ \Xi_{N,i}\right\},\] 
and we can apply Lemma \ref{lem:inid-clt} (rare event CLT for i.n.i.d.\ observations) to obtain its asymptotic normality.

First, for the mean, 
\begin{align*}
&\E_i\left[\left.S_{N,i}\right|\Xi_{N,i}\right] 
=\frac 1 {\sqrt{N\xi_N}}\left(H_{N0}\right)^{-1/2}\E_i\left[\left. \left( \frac{1}{1+\left(\frac{X_{i1}}{X_{i2}}\right)^{Z_i'\theta^*_0 }}-Y_{i1} \right) Z_i\log\frac{X_{i1}}{X_{i2}}\right|\Xi_{N,i}\right]\\ 
&=\frac 1 {\sqrt{N\xi_N}}\left(H_{N0}\right)^{-1/2}\E_i\left[ \left.\left( \frac{1}{1+\left(\frac{X_{i1}}{X_{i2}}\right)^{Z_i'\theta^*_0 }}-\E_i\left[\left.Y_{i1}\right|X_i,Z_i,\Xi_{N,i}\right] \right) Z_i\log\frac{X_{i1}}{X_{i2}}\right|\Xi_{N,i}\right]\\
    &=\frac 1 {\sqrt{N\xi_N}}\left(H_{N0}\right)^{-1/2}\E_i\left[ \left.O\left(\max\{X_{i1},X_{i2}\}^{-\beta^{\min}_i\left( Z_i\right)}\right) Z_i\log\frac{X_{i1}}{X_{i2}}\right|\Xi_{N,i}\right]\\
    &=\frac 1 {\sqrt{N\xi_N}}\left(H_{N0}\right)^{-1/2}o\left(\xn^{-\underline\beta}\right)\E_i\left[ \left. Z_i\log\frac{X_{i1}}{X_{i2}}\right|\Xi_{N,i}\right],
\end{align*}
where the third equality is by \eqref{eq:panel-PY-cond} and $0<\frac{1}{1+\left(\frac{X_{i1}}{X_{i2}}\right)^{Z_i'\theta^*_0 }}<1$, and the last equality is by the lower bound of $\beta^\py_i(Z_i)$ in Assumption \ref{assn:panel-tail}(a). Note that from Assumption \ref{assn:panel-est}(a), $\xn^{-\underline\beta}=o\left(\left(N\xi_N\right)^{-\frac{1+\kappa}{2+\kappa}}\right)$;  from Assumption \ref{assn:panel-est}(b,c), $H_{N0}$ is finite and its smallest eigenvalue is $O\left(\left(N\xi_N\right)^{-\frac{\kappa}{2+\kappa}}\right)$, so the largest eigenvalue of $\left(H_{N0}\right)^{-1/2}$ is $o\left(\left(N\xi_N\right)^{\frac{\kappa}{2(2+\kappa)}}\right)$; from Assumption \ref{assn:panel-est}(c), $\E_i\left[ \left. Z_i\log\frac{X_{i1}}{X_{i2}}\right|\Xi_{N,i}\right]$ is finite. Combining the bounds on all three terms, we have that
\begin{align}
\E_i\left[\left.S_{N,i}\right|\Xi_{N,i}\right] =o\left(\frac 1 {N\xi_N}\right).\label{eq:panel-si-mean}
\end{align}
Therefore, as $p_{N,i}=\xi_{N,i}\left(1+o(1)\right)$ in \eqref{eq:panel-pi}, we obtain that
\begin{align}
    \sumi p_{N,i}\E_i\left[\left.S_{N,i}\right|\Xi_{N,i}\right] =\sumi \xi_{N,i}\cdot \left(1+o(1)\right)\cdot o\left(\frac 1 {N\xi_N}\right)= o\left(1\right).\label{eq:panel-s-sum-mean}
    \end{align}

Second, for the variance,
\begin{align*}
&\V_i\left[\left.S_{N,i}\right|\Xi_{N,i}\right]\\
&=\E_i\left[\left.S_{N,i}S_{N,i}'\right|\Xi_{N,i}\right]-\E_i\left[\left.S_{N,i}\right|\Xi_{N,i}\right] \E_i\left[\left.S_{N,i}\right|\Xi_{N,i}\right]' \\
&= \frac 1 {N\xi_N}\left(H_{N0}\right)^{-1/2}\E_i\left[\left. \left( \frac{1}{1+\left(\frac{X_{i1}}{X_{i2}}\right)^{Z_i'\theta^*_0 }}-Y_{i1} \right)^2 Z_iZ_i'\left(\log\frac{X_{i1}}{X_{i2}}\right)^2\right|\Xi_{N,i}\right]\left(H_{N0}\right)^{-1/2} \\ 
&\quad+o\left(\frac 1 {\left(N\xi_N\right)^2}\right).
\end{align*}
In the second equality, the first term is by the definition of $S_{N,i}$ in \eqref{eq:panel-si-def}, and the second term is by \eqref{eq:panel-si-mean}. For the expectation term, by the law of total expectation,
\begin{align*}
& \E_i\left[\left. \left( \frac{1}{1+\left(\frac{X_{i1}}{X_{i2}}\right)^{Z_i'\theta^*_0 }}-Y_{i1} \right)^2 Z_iZ_i'\left(\log\frac{X_{i1}}{X_{i2}}\right)^2\right|\Xi_{N,i}\right]\\
&= \E_i\left[\left. \E_i\left[\left.\left( \frac{1}{1+\left(\frac{X_{i1}}{X_{i2}}\right)^{Z_i'\theta^*_0 }}-Y_{i1} \right)^2\right|X_i,Z_i,\Xi_{N,i}\right] Z_iZ_i'\left(\log\frac{X_{i1}}{X_{i2}}\right)^2\right|\Xi_{N,i}\right].
\end{align*}
Note that the conditional expectation in the inner term is given by 
\begin{align*}
    &\E_i\left[\left.\left( \frac{1}{1+\left(\frac{X_{i1}}{X_{i2}}\right)^{Z_i'\theta^*_0 }}-Y_{i1} \right)^2\right|X_i,Z_i,\Xi_{N,i}\right]\\ 
&=\E_i\left[\left.\left( \left(Y_{i1}-\E_i\left[\left.Y_{i1}\right|X_i,Z_i,\Xi_{N,i}\right]\right)+\left(\E_i\left[\left.Y_{i1}\right|X_i,Z_i,\Xi_{N,i}\right]-\frac{1}{1+\left(\frac{X_{i1}}{X_{i2}}\right)^{Z_i'\theta^*_0 }}\right) \right)^2\right|X_i,Z_i,\Xi_{N,i}\right]\\ 
&=\V_i\left[\left.Y_{i1}\right|X_i,Z_i,\Xi_{N,i}\right]+\E_i\left[\left.\left( \E_i\left[\left.Y_{i1}\right|X_i,Z_i,\Xi_{N,i}\right]-\frac{1}{1+\left(\frac{X_{i1}}{X_{i2}}\right)^{Z_i'\theta^*_0 }}\right)^2\right|X_i,Z_i,\Xi_{N,i}\right]\\
&=\P_i\left(\left.Y_{i1}=1\right|X_i,Z_i,\Xi_{N,i}\right)\P_i\left(\left.Y_{i1}=0\right|X_i,Z_i,\Xi_{N,i}\right)+o\left(\xn^{-2\underline\beta}\right)\\
&=\frac{\left(\frac{X_{i1}}{X_{i2}}\right)^{Z_i'\theta^*_0 }}{\left(1+\left(\frac{X_{i1}}{X_{i2}}\right)^{Z_i'\theta^*_0 }\right)^2}\left(1+o\left(\xn^{-\underline\beta}\right)\right)+o\left(\xn^{-2\underline\beta}\right),
\end{align*}
where the last three lines are by \eqref{eq:panel-PY-cond}, \eqref{eq:panel-PY-cond0}, and $0<\frac{1}{1+\left(\frac{X_{i1}}{X_{i2}}\right)^{Z_i'\theta^*_0 }}<1$.
Therefore, we have that
\begin{align*}
& \E_i\left[\left. \left( \frac{1}{1+\left(\frac{X_{i1}}{X_{i2}}\right)^{Z_i'\theta^*_0 }}-Y_{i1} \right)^2 Z_iZ_i'\left(\log\frac{X_{i1}}{X_{i2}}\right)^2\right|\Xi_{N,i}\right]\\
&= \E_i\left[\left. \frac{\left(\frac{X_{i1}}{X_{i2}}\right)^{Z_i'\theta^*_0 }}{\left(1+\left(\frac{X_{i1}}{X_{i2}}\right)^{Z_i'\theta^*_0 }\right)^2} Z_iZ_i'\left(\log\frac{X_{i1}}{X_{i2}}\right)^2\right|\Xi_{N,i}\right]\left(1+o\left(\xn^{-\underline\beta}\right)\right)\\ 
&\quad+\E_i\left[\left. Z_iZ_i'\left(\log\frac{X_{i1}}{X_{i2}}\right)^2\right|\Xi_{N,i}\right]o\left(\xn^{-\underline\beta}\right)\\ 
&= H_{N0,i}\left(1+o\left(1\right)\right)+o\left(\xn^{-2\underline\beta}\right),
\end{align*}
where the second equality is by the definition of $H_{N0,i}$ and the fact that $\E_i\left[\left. Z_iZ_i'\left(\log\frac{X_{i1}}{X_{i2}}\right)^2\right|\Xi_{N,i}\right]$ is finite by Assumption \ref{assn:panel-est}(c).
Substituting this back to the expression of $\V\left[S_{N,i}\right]$, we have that as $N\rightarrow\infty$,
\begin{align*}
&\V_i\left[\left.S_{N,i}\right|\Xi_{N,i}\right] \\ 
 &=\frac 1 {N\xi_N}\left[\left(H_{N0}\right)^{-1/2}H_{N0,i}\left(H_{N0}\right)^{-1/2}\left(1+o\left(1\right)\right) +\left(H_{N0}\right)^{-1}o\left(\xn^{-\underline\beta}\right)\right]+o\left(\frac 1 {\left(N\xi_N\right)^2}\right)\notag\\ 
&=\frac 1 {N\xi_N}\left(H_{N0}\right)^{-1/2}H_{N0,i}\left(H_{N0}\right)^{-1/2}\left(1+o\left(1\right)\right)+o\left(\frac 1 {\left(N\xi_N\right)^2}\right)
\end{align*}
The argument for the second line is similar to that for \eqref{eq:panel-si-mean}: $\xn^{-2\underline\beta}=o\left(\left(N\xi_N\right)^{-\frac{2\left(1+\kappa\right)}{2+\kappa}}\right)$ by Assumption \ref{assn:panel-est}(a), and the largest eigenvalue of $\left(H_{N0}\right)^{-1}$ is $o\left(\left(N\xi_N\right)^{\frac{\kappa}{2+\kappa}}\right)$ by Assumption \ref{assn:panel-est}(b,c), so $\left(H_{N0}\right)^{-1}o\left(\xn^{-\underline\beta}\right)=o\left(\frac 1 {N\xi_N}\right)$. Then, as $p_{N,i}=\xi_{N,i}\left(1+o(1)\right)$ in \eqref{eq:panel-pi}, we obtain that
\begin{align}
    &\sumi p_{N,i}\V_i\left[\left.S_{N,i}\right|\Xi_{N,i}\right] \label{eq:panel-s-sum-var}\\ 
    &=\sumi \xi_{N,i}\left(1+o\left(1\right)\right)\cdot\left[\frac 1 {N\xi_N}\left(H_{N0}\right)^{-1/2}H_{N0,i}\left(H_{N0}\right)^{-1/2}\left(1+o\left(1\right)\right)+o\left(\frac 1 {\left(N\xi_N\right)^2}\right)\right]\notag\\
    &=\left(H_{N0}\right)^{-1/2}\left(H_{N0}\right)\left(H_{N0}\right)^{-1/2}\left(1+o\left(1\right)\right)+o\left(1\right)\rightarrow \I_{d_Z}.\notag
    \end{align}

For Lemma \ref{lem:inid-clt}, condition (a) is satisfied by \eqref{eq:panel-p} in Part 1, condition (b) is by Assumption \ref{assn:panel-est}(c), and condition (c) is by Assumption \ref{assn:panel-est}(b).
Therefore, as $\theta^*_0\in \intr\left(\theta^*\right)$, we have that as $N\rightarrow\infty$,
\begin{align}
    \frac 1 {\sqrt{N\xi_N}}\left(H_{N0}\right)^{-1/2}S_N\left(\theta^*_0\right)=\sumi S_{N,i}\I\left\{\Xi_{N,i}\right\} \overset{d}{\rightarrow }\mathcal{N}\left( 0,\I_{d_Z} \right), \label{eq:panel-si-clt}
\end{align}
where $\sumi p_{N,i}\E_i\left[\left.S_{N,i}\right|\Xi_{N,i}\right]$ and $\sumi p_{N,i}\V_i\left[\left.S_{N,i}\right|\Xi_{N,i}\right]$ are given in \eqref{eq:panel-s-sum-mean} and \eqref{eq:panel-s-sum-var}, respectively.

%~~~~~~~~~~~~~~~~~~~~~~~~~~~~~~~~~~~~~~~~~~~~~~~~~~~~~~~~~~~~~~~~~~~~~
\bigskip
\noindent\textbf{Part 4.}
Similarly, for the Hessian matrix,
\begin{align}
    &\frac{1}{N\xi_N}\left(H_{N0}\right)^{-1/2}\cdot H_N\left( \theta^*_0 \right)\cdot\left(H_{N0}\right)^{-1/2}\label{eq:panel-hi-lln}\\ 
    &=\frac{1}{N\xi_N}\left(H_{N0}\right)^{-1/2}\cdot\left[ -\sumi \frac{\left(\frac{X_{i1}}{X_{i2}}\right)^{Z_i'\theta^*_0 }}{\left(1+\left(\frac{X_{i1}}{X_{i2}}\right)^{Z_i'\theta^*_0 }\right)^2} Z_iZ_i'\left(\log\frac{X_{i1}}{X_{i2}}\right)^2 \1\left\{ \Xi_{N,i}^\py\right\}\right]\cdot\left(H_{N0}\right)^{-1/2}\notag\\ 
    &=-\left(H_{N0}\right)^{-1/2}\cdot\frac{1}{N\xi_N}\sumi p_{N,i}\E_i\left[ \left.\frac{\left(\frac{X_{i1}}{X_{i2}}\right)^{Z_i'\theta^*_0 }}{\left(1+\left(\frac{X_{i1}}{X_{i2}}\right)^{Z_i'\theta^*_0 }\right)^2} Z_iZ_i'\left(\log\frac{X_{i1}}{X_{i2}}\right)^2 \right|\Xi_{N,i}\right]\cdot\left(H_{N0}\right)^{-1/2}\cdot \left(1 + o_p(1)\right)\notag\\ 
    &=-\left(H_{N0}\right)^{-1/2}\cdot\frac{1}{N\xi_N}\sumi \xi_{N,i}\E_i\left[ \left.\frac{\left(\frac{X_{i1}}{X_{i2}}\right)^{Z_i'\theta^*_0 }}{\left(1+\left(\frac{X_{i1}}{X_{i2}}\right)^{Z_i'\theta^*_0 }\right)^2} Z_iZ_i'\left(\log\frac{X_{i1}}{X_{i2}}\right)^2 \right|\Xi_{N,i}\right]\cdot\left(H_{N0}\right)^{-1/2}\cdot \left(1 + o_p(1)\right)\notag\\ 
    &=-\left(H_{N0}\right)^{-1/2}\cdot H_{N0}\cdot\left(H_{N0}\right)^{-1/2}\cdot \left(1 + o_p(1)\right)\notag\\ 
    &\overset{p}{\rightarrow }-\I_{d_Z}.\notag
\end{align}
The first equality is by the definition of the Hessian matrix \eqref{eq:panel-hessian}. The second equality follows from Lemma \ref{lem:inid-lln} (rare event LLN for i.n.i.d.\ observations), where condition (a) is satisfied by \eqref{eq:panel-p} in Part 1, and condition (b) is by Assumption \ref{assn:panel-est}(c). The third equality is by $p_{N,i}=\xi_{N,i}\left(1+o(1)\right)$ in \eqref{eq:panel-pi}. The fourth equality is by the definition of $H_{N0}$.

%~~~~~~~~~~~~~~~~~~~~~~~~~~~~~~~~~~~~~~~~~~~~~~~~~~~~~~~~~~~~~~~~~~~~~
\bigskip
\noindent\textbf{Part 5.}
Let $\zeta_N=\left(H_{N0}\right)^{1/2}\left(\theta^*-\theta^*_0\right)$, $\zeta_{N0}=\left(H_{N0}\right)^{1/2}\theta^*_0$, and $W_i=\left(H_{N0}\right)^{-1/2}Z_i$. Then, the log likelihood function can be rewritten as
\begin{align*}
\tilde\ell_N\left( \zeta_N \right)
= \sumi\left[ W_i'\left(\zeta_N + \zeta_{N0}\right) (1-Y_{i1})\log\frac{X_{i1}}{X_{i2}}-\log\left(1+\left(\frac{X_{i1}}{X_{i2}}\right)^{W_i'\left(\zeta_N + \zeta_{N0}\right) }\right) \right]
\1\{\Xi_{N,i}\},
\end{align*}
the corresponding score and Hessian are denoted by $\tilde S_N\left( \zeta_N \right)$ and $\tilde H_N\left( \zeta_N \right)$, respectively,
and the MLE estimate is denoted by $\hat\zeta_N$.

First, the Hessian matrix is given by 
\[\tilde H_N\left( \zeta_N \right)=-\sumi \frac{\left(\frac{X_{i1}}{X_{i2}}\right)^{W_i'\left(\zeta_N + \zeta_{N0}\right)}}{\left(1+\left(\frac{X_{i1}}{X_{i2}}\right)^{W_i'\left(\zeta_N + \zeta_{N0}\right)}\right)^2} W_iW_i'\left(\log\frac{X_{i1}}{X_{i2}}\right)^2 \1\left\{ \Xi_{N,i}^\py\right\}.\] 
It is positive definite for all $\zeta_N=\left(H_{N0}\right)^{1/2}\left(\theta^*-\theta^*_0\right)$ with $\theta^*\in\Theta^*$. The reason is that Assumption \ref{assn:panel-est}(b,c) implies that both $H_{N0}$ and $\E_i\left[ \left.Z_iZ_i' \right|\Xi_{N,i}\right]$ are finite and positive definite, and thus $\E_i\left[ \left.W_iW_i' \right|\Xi_{N,i}\right]$ is finite and positive definite. 
Then, the log likelihood function is $\tilde\ell_N\left( \zeta_N \right)$ is strictly concave over its domain, and the MLE estimate $\hat\zeta_N$ is unique.

Second, let \[\mathcal U_{N,C} = \left\{u\in\mathbb R^{d_Z}:\,\frac{1}{\sqrt{N\xi_N}}\left(H_{N0}\right)^{-1/2}u+\theta^*_0\in\theta^*\text{ and }\|u\|=C\right\}.\] Note that $\mathcal U_{N,C}\neq\emptyset$ for any $C>0$. This follows from three facts: first, $\theta^*$ is a convex cone; second, $\intr\left(\theta^*\right)\neq\emptyset$ as $\theta^*_0\in \intr\left(\theta^*\right)$; and third, $H_{N0}$ is finite and full rank by Assumption \ref{assn:panel-est}(b,c). Then, let $u$ be an arbitrary non-random vector in $\mathcal U_{N,C}$. Applying the second-order Taylor expansion of $\tilde\ell_N\left( \frac{u}{\sqrt{N\xi_N}}\right)$ around $\zeta_N=0$, we have that
\begin{equation*}
\tilde\ell_N\left( \frac{u}{\sqrt{N\xi_N}}\right) -\tilde\ell_N\left( 0\right) = \frac{1}{\sqrt{N\xi_N}}u'\tilde S_N\left( 0\right) + \frac{1}{2N\xi_N}u'\tilde H_N\left( 0\right)u + o_{p}(1)
\end{equation*}
From \eqref{eq:panel-si-clt} and \eqref{eq:panel-hi-lln}, we have that as $N\rightarrow\infty$,
\begin{align}
    \frac{1}{\sqrt{N\xi_N}}\tilde S_N\left( 0 \right)&=\frac{1}{\sqrt{N\xi_N}}S_N\left(\theta^*_0\right) \overset{d}{\rightarrow }\mathcal{N}\left( 0,\I_{d_Z} \right),\label{eq:panel-zeta-clt}\\
    \frac{1}{N\xi_N}\tilde H_N\left( 0 \right)
    &=\frac{1}{N\xi_N}\left(H_{N0}\right)^{-1/2}\cdot H_N\left( \theta^*_0 \right)\cdot\left(H_{N0}\right)^{-1/2}\overset{p}{\rightarrow }-\I_{d_Z}. \label{eq:panel-zeta-lln}
\end{align}
This implies that when $C$ is large enough, the quadratic term dominates the linear one with an arbitrarily large probability. That is, for any $\varepsilon>0$, there exists a $C>0$ such that 
\begin{equation*}
\underset{N}{\lim \sup }\,\P\left( \sup_{u\in \mathcal U_{N,C}}\tilde\ell_N\left( \frac{u}{\sqrt{N\xi_N}}\right) <\tilde\ell_N\left(0\right) \right) >1-\varepsilon.
\end{equation*}
Therefore, $\tilde\ell_N\left( \cdot \right)$ must have at least one local maximizer, which is of order $O_{p}\left( \sqrt{N\xi_N}\right)$. And by the uniqueness of the MLE, this local maximizer is the global maximizer, and $\hat\zeta_N=O_{p}\left( \sqrt{N\xi_N}\right)$.

Finally, by the first order condition, $\tilde S_N\left( \hat\zeta_N\right)=0$. Applying the first-order Taylor expansion around $\zeta_N=0$, we have that 
\[0=\tilde S_N\left( \hat\zeta_N\right)=\tilde S_N\left( 0\right)+\tilde H_N\left( 0\right)\hat\zeta_N+o_{p}(1),\] which implies that 
\[\sqrt{N\xi_N}\hat\zeta_N=-\left(\frac{1}{N\xi_N}\tilde H_N\left( 0\right)\right)^{-1}\cdot\frac{1}{\sqrt{N\xi_N}}\left(\tilde S_N\left( 0\right)+o_{p}(1)\right)\overset{d}{\rightarrow }\mathcal{N}\left( 0,\I_{d_Z}\right),\]
following from \eqref{eq:panel-zeta-clt} and \eqref{eq:panel-zeta-lln}. 
This is equivalent to
\begin{align*}
\sqrt{N\xi_N}\left(H_{N0}\right)^{1/2}\left( \hat\theta^*-\theta^*_0\right) \overset{d}{\rightarrow }\mathcal{N}\left( 0,\I_{d_Z}\right). 
\end{align*}
\hfill\end{proof}

%~~~~~~~~~~~~~~~~~~~~~~~~~~~~~~~~~~~~~~~~~~~~~~~~~~~~~~~~~~~~~~~~~~~~~
\begin{proposition}[Panel data: proportion of tail switchers]\label{prop:panel-nxi}
Suppose Assumptions \ref{assn:panel-model}--\ref{assn:panel-est} hold.  as $N\rightarrow\infty$, 
\begin{align*}
\frac{N_{\Xi}}{N}=\xi_N\left(1+o_p(1)\right).
\end{align*}
\end{proposition}

\noindent\begin{proof} First, for the mean of $N_{\Xi}/N$, 
\begin{align*}
\E\left[\frac{N_{\Xi}}{N}\right]
=\E\left[\avgi\1\{\Xi_{N,i}\}\right]
=\avgi\P_i\left(\Xi_{N,i}\right)
=\xi_N\left(1+o(1)\right),
\end{align*}
where the second equality follows from the independence across $i$ in Assumption \ref{assn:panel-model}(a), and the last equality follows from \eqref{eq:panel-p}. Second, let $\V_i$ be the variance given $\left\{\lambda_i,\C_i\right\}$. Similarly, based on the independence across $i$ and \eqref{eq:panel-p},
\begin{align*}
\V\left[\frac{N_{\Xi}}{N}\right]
&=\V\left[\avgi\1\{\Xi_{N,i}\}\right]
=\frac 1 {N^2}\sumi\V_i\left[\1\{\Xi_{N,i}\}\right]
=\frac 1 {N^2}\sumi\P_i\left(\Xi_{N,i}\right)\left(1-\P_i\left(\Xi_{N,i}\right)\right),\\
&\le\frac 1 {N^2}\sumi\P_i\left(\Xi_{N,i}\right)=\frac{\xi_N} N \left(1+o(1)\right).
\end{align*}

Together, we have that as $N\rightarrow \infty$, 
\begin{align*}
\E\left[\left(\frac{N_{\Xi}}{N\xi_N}-1\right)^2\right]&=\V\left[\frac{N_{\Xi}}{N\xi_N}\right]+o(1)=\frac 1 {\xi_n^2}\V\left[\frac{N_{\Xi}}{N}\right]+o(1)\\
&\le\frac 1 {N\xi_N} \left(1+o(1)\right)+o(1)\rightarrow 0,
\end{align*}
where the convergence to 0 follows from Assumption \ref{assn:panel-est}(a). As convergence in the second moment implies convergence in probability, we obtain that as $N\rightarrow \infty$,
\begin{align*}
\frac{N_{\Xi}}{N}\rightarrow\xi_N\left(1+o_p(1)\right).
\end{align*}
\hfill\end{proof}

%~~~~~~~~~~~~~~~~~~~~~~~~~~~~~~~~~~~~~~~~~~~~~~~~~~~~~~~~~~~~~~~~~~~~~
\begin{remark} 
\normalfont{
Combining Theorem \ref{thm:panel} with Proposition \ref{prop:panel-nxi}, we have that as $N\rightarrow\infty$,
\begin{align*}
\sqrt{N_{\Xi}}\left(H_{N0}\right)^{1/2}\left( \hat\theta^*-\theta^*_0\right) \overset{d}{\rightarrow }\mathcal{N}\left( 0,\I_{d_Z}\right). 
\end{align*}
}
\end{remark}

%=====================================================================
%\newpage
\section{Additional Tables and Figures}\label{sec:app-tab-fig}
%\subsection{Monte Carlo}
Tables \ref{tab:sim-exp1-param}--\ref{tab:sim-exp1-ape} provide detailed results, including bias, standard deviation, and RMSE, for various estimands across all model specifications in Monte Carlo Experiment 1. See Table \ref{tab:sim-exp1-param} for tail parameters, Table \ref{tab:sim-exp1-elas} for extreme elasticity, Table \ref{tab:sim-exp1-py} for conditional probability, and Table \ref{tab:sim-exp1-ape} for partial effects. Figure \ref{fig:sim-exp2-lps} shows the scatter plots of the LPS from all Monte Carlo repetitions in Experiment 2 with $\alpha_X = 1$ and $\alpha_{\varepsilon}=1$.

Figure \ref{fig:app-hist} and Tables \ref{tab:app-stat}--\ref{tab:app-lps-ape-c} present more details for the empirical example on housing prices and bank riskiness. Figure \ref{fig:app-hist} depicts the histograms for distributions of $X_{it}|Y_{it}=y$ for the baseline sample. Table \ref{tab:app-stat} summarizes the descriptive statistics for all samples. Table \ref{tab:app-lps-ape-all} reports a forecasting comparison across estimators for all samples, along with the parameter and APE estimates based on the tail estimator. Lastly, Table \ref{tab:app-lps-ape-c} shows a robustness check across different values of level $c$ for the baseline sample.

\begin{table}[h]
\caption{Parameter estimation - Experiment 1}
\label{tab:sim-exp1-param}
\begin{center}
\scalebox{.8}{
\begin{tabular}{ll|rrr|rrr|rrr|rrr} \hline \hline 
& & \multicolumn{3}{c|}{$\alpha_X=0.5$} & \multicolumn{3}{c|}{$\alpha_X=1$} &\multicolumn{3}{c|}{$\alpha_X=1.5$} &\multicolumn{3}{c}{$\alpha_X=2$} \\
& & Bias & SD & \footnotesize RMSE & Bias & SD &\footnotesize RMSE & Bias & SD &\footnotesize RMSE & Bias & SD &\footnotesize RMSE \\ \hline
\multirow2{*}{$\alpha_{\varepsilon}=0.5$}&$\hat\alpha^\po$ &0.002 &0.131 &0.131 &-0.020 &0.180 &0.181 &-0.025 &0.237 &0.238 &-0.107 &0.289 &0.309 \\ 
&$\hat\alpha^\pl$ &0.003 &0.062 &0.062 &-0.005 &0.118 &0.118 &-0.036 &0.174 &0.178 &-0.075 &0.231 &0.243 \\ \hline
\multirow2{*}{$\alpha_{\varepsilon}=1$}&$\hat\alpha^\po$ &0.026 &0.186 &0.188 &-0.023 &0.253 &0.254 &-0.081 &0.279 &0.290 &-0.155 &0.325 &0.360 \\ 
&$\hat\alpha^\pl$ &0.004 &0.063 &0.063 &0.010 &0.124 &0.124 &-0.005 &0.182 &0.182 &-0.033 &0.238 &0.241 \\ \hline
\multirow2{*}{$\alpha_{\varepsilon}=1.5$}&$\hat\alpha^\po$ &0.067 &0.269 &0.277 &-0.043 &0.301 &0.304 &-0.144 &0.327 &0.357 &-0.277 &0.367 &0.460 \\ 
&$\hat\alpha^\pl$ &0.004 &0.063 &0.063 &0.011 &0.124 &0.125 &0.003 &0.185 &0.185 &-0.017 &0.243 &0.244 \\ \hline
\multirow2{*}{$\alpha_{\varepsilon}=2$}&$\hat\alpha^\po$ &0.105 &0.326 &0.342 &-0.095 &0.353 &0.365 &-0.267 &0.373 &0.458 &-0.427 &0.387 &0.576 \\ 
&$\hat\alpha^\pl$ &0.004 &0.063 &0.063 &0.011 &0.124 &0.125 &0.004 &0.185 &0.185 &-0.013 &0.244 &0.244 \\ \hline

\end{tabular} 
} 
\end{center}
{\footnotesize {\em Notes:} Bias, SD, and RMSE are calculated based on $N_{sim}=1000$ repetitions and with respect to the theoretical values as the tails, $\alpha^\po=\alpha_X+\alpha_{\varepsilon}$ and $\alpha^\pl=\alpha_X$.}\setlength{\baselineskip}{4mm}
\end{table}

%\clearpage
\floatpagestyle{plain}
\begin{table}[h]
\caption{Extreme elasticity estimation - Experiment 1}
\label{tab:sim-exp1-elas}
\begin{center}
\scalebox{.675}{
\begin{tabular}{ll|rrr|rrr|rrr|rrr} \hline \hline 
& & \multicolumn{3}{c|}{$\alpha_X=0.5$} & \multicolumn{3}{c|}{$\alpha_X=1$} &\multicolumn{3}{c|}{$\alpha_X=1.5$} &\multicolumn{3}{c}{$\alpha_X=2$} \\
& & Bias & SD & \footnotesize RMSE & Bias & SD &\footnotesize RMSE & Bias & SD &\footnotesize RMSE & Bias & SD &\footnotesize RMSE \\ \hline
\multirow{5}{*}{\rotatebox{90}{$\alpha_{\varepsilon}=0.5$}} & Tail &0.00 &0.15 &0.15 &0.01 &0.21 &0.21 &-0.01 &0.30 &0.30 &0.03 &0.36 &0.37 \\ 
& Logit, tail $X$ &0.47 &0.04 &0.47 &0.23 &0.08 &0.25 &0.07 &0.10 &0.12 &-0.02 &0.11 &0.11 \\ 
& Logit, all $X$ &-17.49 &10.26 &20.27 &-2.21 &0.81 &2.35 &-1.21 &0.29 &1.25 &-0.84 &0.17 &0.86 \\ 
& Local Logit &2091.55 &29194.28 &29215.53 &-39.20 &1033.30 &1033.48 &-2.70 &42.08 &42.14 &-0.39 &9.44 &9.44 \\ 
& Local linear &0.47 &0.07 &0.47 &-3.88 &10.72 &11.40 &-0.30 &1.17 &1.21 &-0.11 &1.77 &1.77 \\ \hline
\multirow{5}{*}{\rotatebox{90}{$\alpha_{\varepsilon}=1$}} & Tail &-0.02 &0.20 &0.20 &0.03 &0.29 &0.29 &0.08 &0.34 &0.35 &0.12 &0.40 &0.42 \\ 
& Logit, tail $X$ &1.00 &0.01 &1.00 &0.87 &0.19 &0.89 &0.50 &0.29 &0.58 &0.31 &0.26 &0.40 \\ 
& Logit, all $X$ &-14.24 &11.21 &17.85 &-16.82 &1.59 &16.90 &-6.49 &0.46 &6.51 &-4.11 &0.26 &4.11 \\ 
& Local Logit &202.17 &1870.36 &1877.34 &-129.29 &1661.59 &1664.55 &-8.35 &363.05 &362.97 &-1.92 &30.73 &30.78 \\ 
& Local linear &0.96 &0.08 &0.97 &4.28 &46.29 &46.47 &-3.25 &56.96 &57.03 &-0.49 &3.34 &3.37 \\ \hline
\multirow{5}{*}{\rotatebox{90}{$\alpha_{\varepsilon}=1.5$}} & Tail &-0.06 &0.28 &0.28 &0.05 &0.33 &0.33 &0.15 &0.38 &0.40 &0.26 &0.45 &0.52 \\ 
& Logit, tail $X$ &1.50 &0.01 &1.50 &1.44 &0.51 &1.53 &1.19 &0.42 &1.26 &0.92 &0.48 &1.04 \\ 
& Logit, all $X$ &-12.09 &12.31 &16.81 &-30.68 &1.76 &30.73 &-10.91 &0.53 &10.92 &-6.47 &0.31 &6.48 \\ 
& Local Logit &1.50 &0.00 &1.50 &-860.37 &7495.52 &7512.01 &-100.10 &1580.91 &1583.14 &-11.44 &96.72 &97.35 \\ 
& Local linear &1.46 &0.08 &1.47 &3.06 &159.59 &159.54 &2.16 &81.37 &81.36 &-0.95 &7.16 &7.21 \\ \hline
\multirow{5}{*}{\rotatebox{90}{$\alpha_{\varepsilon}=2$}} & Tail &-0.10 &0.33 &0.35 &0.11 &0.38 &0.39 &0.27 &0.41 &0.49 &0.41 &0.47 &0.63 \\ 
& Logit, tail $X$ &2.00 &0.01 &2.00 &1.81 &1.10 &2.12 &1.84 &0.40 &1.89 &1.60 &0.59 &1.70 \\ 
& Logit, all $X$ &-11.67 &12.40 &16.57 &-29.01 &8.22 &30.12 &-13.84 &0.56 &13.85 &-7.93 &0.31 &7.94 \\ 
& Local Logit &2.00 &0.00 &2.00 &14.23 &1434.99 &1410.94 &-137.52 &610.54 &625.27 &-34.39 &216.59 &219.19 \\ 
& Local linear &1.96 &0.08 &1.97 &7.36 &53.36 &53.84 &-13.43 &151.57 &152.08 &-3.60 &18.53 &18.87 \\ \hline

\end{tabular} 
} 
\end{center}
{\footnotesize {\em Notes:} Bias, SD, and RMSE are calculated based on $N_{sim}=1000$ repetitions and with respect to the theoretical values as the tails, $-|\alpha^\pl-\alpha^\po|=-\alpha_{\varepsilon}$. Evaluated at the 97.5th percentile of the distribution of $X$.}\setlength{\baselineskip}{4mm}
\end{table}

\clearpage
\floatpagestyle{empty}
\begin{table}[h]
\caption{$\hat \P(Y=1|X=x)$ - Experiment 1}
\label{tab:sim-exp1-py}
\begin{center}
\vspace{-1em}
\scalebox{.8}{
\begin{tabular}{lll|rrr|rrr|rrr|rrr} \hline \hline 
& & & \multicolumn{3}{c|}{$\alpha_X=0.5$} & \multicolumn{3}{c|}{$\alpha_X=1$} &\multicolumn{3}{c|}{$\alpha_X=1.5$} &\multicolumn{3}{c}{$\alpha_X=2$} \\
& & & Bias & SD & \footnotesize RMSE & Bias & SD &\footnotesize RMSE & Bias & SD &\footnotesize RMSE & Bias & SD &\footnotesize RMSE \\ \hline
\multirow{20}{*}{\rotatebox{90}{$\alpha_{\varepsilon}=0.5$}} & \multirow{5}{*}{\rotatebox{90}{$x:90\%$}} & Tail &0.01 &0.03 &0.03 &0.03 &0.07 &0.07 &0.02 &0.09 &0.09 &0.03 &0.09 &0.09 \\ 
& & Logit, tail $X$ &-0.88 &0.05 &0.88 &-0.68 &0.08 &0.68 &-0.60 &0.07 &0.61 &-0.55 &0.07 &0.56 \\ 
& & Logit, all $X$ &-0.18 &0.13 &0.22 &-0.13 &0.03 &0.13 &-0.07 &0.02 &0.08 &-0.05 &0.01 &0.05 \\ 
& & Local Logit &0.01 &0.05 &0.06 &0.00 &0.04 &0.04 &0.00 &0.03 &0.03 &-0.00 &0.03 &0.03 \\ 
& & Local linear &-0.41 &0.02 &0.41 &-0.12 &0.06 &0.13 &-0.01 &0.02 &0.02 &-0.00 &0.02 &0.02 \\ \cline{2-15}
& \multirow{5}{*}{\rotatebox{90}{$x:95\%$}} & Tail &0.00 &0.02 &0.02 &0.01 &0.04 &0.04 &0.02 &0.06 &0.06 &0.02 &0.06 &0.06 \\ 
& & Logit, tail $X$ &-0.93 &0.06 &0.93 &-0.73 &0.09 &0.73 &-0.62 &0.08 &0.63 &-0.56 &0.08 &0.56 \\ 
& & Logit, all $X$ &-0.02 &0.13 &0.13 &-0.05 &0.06 &0.08 &-0.03 &0.03 &0.04 &-0.01 &0.02 &0.02 \\ 
& & Local Logit &0.01 &0.12 &0.12 &0.01 &0.07 &0.07 &0.00 &0.05 &0.05 &-0.00 &0.04 &0.04 \\ 
& & Local linear &-0.45 &0.02 &0.45 &-0.06 &0.09 &0.11 &-0.00 &0.02 &0.02 &-0.00 &0.02 &0.02 \\ \cline{2-15}
& \multirow{5}{*}{\rotatebox{90}{$x:97.5\%$}} & Tail &-0.00 &0.01 &0.01 &0.00 &0.03 &0.03 &0.01 &0.03 &0.04 &0.02 &0.04 &0.04 \\ 
& & Logit, tail $X$ &-0.94 &0.07 &0.94 &-0.70 &0.09 &0.70 &-0.57 &0.07 &0.58 &-0.50 &0.06 &0.51 \\ 
& & Logit, all $X$ &0.00 &0.08 &0.08 &0.05 &0.06 &0.08 &0.05 &0.03 &0.06 &0.04 &0.02 &0.05 \\ 
& & Local Logit &-0.11 &0.25 &0.27 &0.02 &0.12 &0.12 &0.01 &0.09 &0.09 &0.01 &0.07 &0.07 \\ 
& & Local linear &-0.47 &0.03 &0.47 &0.00 &0.09 &0.09 &0.00 &0.03 &0.03 &0.00 &0.04 &0.04 \\ \cline{2-15}
& \multirow{5}{*}{\rotatebox{90}{$x:99\%$}} & Tail &-0.00 &0.01 &0.01 &-0.00 &0.02 &0.02 &0.00 &0.03 &0.03 &0.00 &0.03 &0.03 \\ 
& & Logit, tail $X$ &-0.03 &0.10 &0.11 &-0.17 &0.13 &0.21 &-0.23 &0.09 &0.25 &-0.22 &0.07 &0.23 \\ 
& & Logit, all $X$ &0.00 &0.05 &0.05 &0.07 &0.04 &0.08 &0.11 &0.02 &0.12 &0.11 &0.02 &0.11 \\ 
& & Local Logit &-0.48 &0.08 &0.48 &-0.01 &0.27 &0.27 &0.02 &0.19 &0.19 &0.01 &0.14 &0.14 \\ 
& & Local linear &-0.41 &0.11 &0.43 &0.03 &0.08 &0.08 &0.01 &0.06 &0.06 &-0.00 &0.08 &0.08 \\ \hline
\multirow{20}{*}{\rotatebox{90}{$\alpha_{\varepsilon}=1$}} & \multirow{5}{*}{\rotatebox{90}{$x:90\%$}} & Tail &-0.00 &0.01 &0.01 &0.01 &0.03 &0.03 &0.02 &0.05 &0.06 &0.04 &0.06 &0.07 \\ 
& & Logit, tail $X$ &-0.98 &0.02 &0.98 &-0.90 &0.01 &0.90 &-0.83 &0.02 &0.83 &-0.79 &0.01 &0.79 \\ 
& & Logit, all $X$ &0.01 &0.04 &0.04 &0.06 &0.01 &0.06 &0.03 &0.02 &0.03 &0.02 &0.01 &0.02 \\ 
& & Local Logit &0.00 &0.02 &0.02 &0.00 &0.03 &0.03 &-0.00 &0.02 &0.02 &-0.00 &0.02 &0.02 \\ 
& & Local linear &-0.51 &0.02 &0.51 &-0.20 &0.10 &0.23 &-0.01 &0.03 &0.04 &-0.00 &0.01 &0.01 \\ \cline{2-15}
& \multirow{5}{*}{\rotatebox{90}{$x:95\%$}} & Tail &-0.00 &0.00 &0.00 &0.00 &0.02 &0.02 &0.01 &0.03 &0.03 &0.02 &0.03 &0.04 \\ 
& & Logit, tail $X$ &-1.00 &0.02 &1.00 &-0.94 &0.02 &0.94 &-0.88 &0.03 &0.88 &-0.83 &0.03 &0.83 \\ 
& & Logit, all $X$ &0.00 &0.03 &0.03 &0.05 &0.00 &0.05 &0.08 &0.01 &0.08 &0.07 &0.01 &0.07 \\ 
& & Local Logit &0.00 &0.04 &0.04 &0.00 &0.04 &0.04 &0.00 &0.04 &0.04 &-0.00 &0.03 &0.03 \\ 
& & Local linear &-0.52 &0.02 &0.52 &-0.07 &0.14 &0.16 &0.00 &0.03 &0.03 &-0.00 &0.02 &0.02 \\ \cline{2-15}
& \multirow{5}{*}{\rotatebox{90}{$x:97.5\%$}} & Tail &-0.00 &0.00 &0.00 &-0.00 &0.01 &0.01 &0.00 &0.02 &0.02 &0.01 &0.02 &0.02 \\ 
& & Logit, tail $X$ &-1.00 &0.02 &1.00 &-0.94 &0.06 &0.94 &-0.83 &0.07 &0.84 &-0.78 &0.06 &0.78 \\ 
& & Logit, all $X$ &-0.00 &0.02 &0.02 &0.02 &0.00 &0.02 &0.06 &0.00 &0.06 &0.08 &0.00 &0.08 \\ 
& & Local Logit &-0.12 &0.21 &0.24 &0.01 &0.05 &0.05 &0.01 &0.05 &0.05 &0.01 &0.05 &0.05 \\ 
& & Local linear &-0.51 &0.03 &0.51 &0.04 &0.13 &0.14 &0.01 &0.03 &0.03 &0.00 &0.03 &0.03 \\ \cline{2-15}
& \multirow{5}{*}{\rotatebox{90}{$x:99\%$}} & Tail &-0.00 &0.00 &0.00 &-0.00 &0.01 &0.01 &-0.00 &0.01 &0.01 &0.00 &0.02 &0.02 \\ 
& & Logit, tail $X$ &-0.00 &0.02 &0.02 &-0.01 &0.04 &0.04 &-0.05 &0.07 &0.09 &-0.07 &0.08 &0.11 \\ 
& & Logit, all $X$ &-0.00 &0.02 &0.02 &0.01 &0.00 &0.01 &0.04 &0.00 &0.04 &0.06 &0.00 &0.06 \\ 
& & Local Logit &-0.49 &0.08 &0.49 &0.00 &0.09 &0.09 &0.01 &0.08 &0.08 &0.02 &0.07 &0.08 \\ 
& & Local linear &-0.44 &0.12 &0.45 &0.06 &0.11 &0.12 &0.01 &0.05 &0.05 &0.00 &0.04 &0.04 \\ \hline

\end{tabular} 
} 
\end{center}
{\footnotesize {\em Notes:} Bias, SD, and RMSE are calculated based on $N_{sim}=1000$ repetitions w.r.t.\ the theoretical values.}\setlength{\baselineskip}{4mm}
\end{table}

\begin{table}[h]
\caption*{Table \ref{tab:sim-exp1-py}: $\hat \P(Y=1|X=x)$ - Experiment 1 (cont.)}
\label{tab:sim-exp1-py-2}
\vspace{-1em}
\begin{center}
\scalebox{.8}{
\begin{tabular}{lll|rrr|rrr|rrr|rrr} \hline \hline 
& & & \multicolumn{3}{c|}{$\alpha_X=0.5$} & \multicolumn{3}{c|}{$\alpha_X=1$} &\multicolumn{3}{c|}{$\alpha_X=1.5$} &\multicolumn{3}{c}{$\alpha_X=2$} \\
& & & Bias & SD & \footnotesize RMSE & Bias & SD &\footnotesize RMSE & Bias & SD &\footnotesize RMSE & Bias & SD &\footnotesize RMSE \\ \hline

\multirow{20}{*}{\rotatebox{90}{$\alpha_{\varepsilon}=1.5$}} & \multirow{5}{*}{\rotatebox{90}{$x:90\%$}} & Tail &-0.00 &0.00 &0.00 &0.00 &0.02 &0.02 &0.02 &0.03 &0.04 &0.03 &0.04 &0.05 \\ 
& & Logit, tail $X$ &-1.00 &0.02 &1.00 &-0.95 &0.00 &0.95 &-0.90 &0.01 &0.90 &-0.86 &0.01 &0.86 \\ 
& & Logit, all $X$ &-0.00 &0.04 &0.04 &0.05 &0.00 &0.05 &0.06 &0.01 &0.06 &0.05 &0.01 &0.05 \\ 
& & Local Logit &0.00 &0.01 &0.01 &0.00 &0.02 &0.02 &0.00 &0.02 &0.02 &-0.00 &0.02 &0.02 \\ 
& & Local linear &-0.52 &0.02 &0.52 &-0.23 &0.12 &0.26 &-0.01 &0.04 &0.04 &-0.00 &0.01 &0.01 \\ \cline{2-15}
& \multirow{5}{*}{\rotatebox{90}{$x:95\%$}} & Tail &-0.00 &0.00 &0.00 &-0.00 &0.01 &0.01 &0.00 &0.02 &0.02 &0.02 &0.02 &0.02 \\ 
& & Logit, tail $X$ &-1.00 &0.02 &1.00 &-0.98 &0.01 &0.98 &-0.94 &0.01 &0.94 &-0.91 &0.02 &0.91 \\ 
& & Logit, all $X$ &-0.00 &0.03 &0.03 &0.02 &0.00 &0.02 &0.05 &0.00 &0.05 &0.06 &0.00 &0.06 \\ 
& & Local Logit &0.00 &0.01 &0.01 &0.00 &0.02 &0.02 &0.00 &0.03 &0.03 &0.00 &0.02 &0.02 \\ 
& & Local linear &-0.52 &0.02 &0.52 &-0.06 &0.16 &0.17 &0.01 &0.03 &0.03 &-0.00 &0.02 &0.02 \\ \cline{2-15}
& \multirow{5}{*}{\rotatebox{90}{$x:97.5\%$}} & Tail &-0.00 &0.00 &0.00 &-0.00 &0.00 &0.00 &0.00 &0.01 &0.01 &0.01 &0.01 &0.01 \\ 
& & Logit, tail $X$ &-1.00 &0.02 &1.00 &-0.99 &0.03 &0.99 &-0.93 &0.07 &0.93 &-0.87 &0.07 &0.88 \\ 
& & Logit, all $X$ &-0.00 &0.02 &0.02 &0.01 &0.00 &0.01 &0.02 &0.00 &0.02 &0.05 &0.00 &0.05 \\ 
& & Local Logit &-0.12 &0.21 &0.24 &0.00 &0.04 &0.04 &0.01 &0.03 &0.03 &0.01 &0.03 &0.03 \\ 
& & Local linear &-0.50 &0.03 &0.50 &0.05 &0.14 &0.15 &0.01 &0.04 &0.04 &0.00 &0.02 &0.02 \\ \cline{2-15}
& \multirow{5}{*}{\rotatebox{90}{$x:99\%$}} & Tail &-0.00 &0.00 &0.00 &-0.00 &0.00 &0.00 &-0.00 &0.01 &0.01 &-0.00 &0.01 &0.01 \\ 
& & Logit, tail $X$ &-0.00 &0.02 &0.02 &-0.00 &0.01 &0.01 &-0.01 &0.03 &0.03 &-0.02 &0.05 &0.05 \\ 
& & Logit, all $X$ &-0.00 &0.02 &0.02 &0.00 &0.00 &0.00 &0.01 &0.00 &0.01 &0.02 &0.00 &0.02 \\ 
& & Local Logit &-0.49 &0.08 &0.49 &-0.00 &0.06 &0.06 &0.00 &0.03 &0.03 &0.01 &0.04 &0.04 \\ 
& & Local linear &-0.43 &0.12 &0.45 &0.06 &0.12 &0.14 &0.01 &0.05 &0.05 &0.00 &0.03 &0.03 \\ \hline
\multirow{20}{*}{\rotatebox{90}{$\alpha_{\varepsilon}=2$}} & \multirow{5}{*}{\rotatebox{90}{$x:90\%$}} & Tail &-0.00 &0.00 &0.00 &0.00 &0.01 &0.01 &0.01 &0.02 &0.03 &0.03 &0.03 &0.04 \\ 
& & Logit, tail $X$ &-1.00 &0.02 &1.00 &-0.97 &0.00 &0.97 &-0.93 &0.00 &0.93 &-0.90 &0.00 &0.90 \\ 
& & Logit, all $X$ &-0.00 &0.04 &0.04 &0.03 &0.00 &0.03 &0.05 &0.00 &0.05 &0.05 &0.00 &0.06 \\ 
& & Local Logit &0.00 &0.00 &0.00 &0.00 &0.01 &0.01 &0.00 &0.02 &0.02 &0.00 &0.02 &0.02 \\ 
& & Local linear &-0.51 &0.02 &0.51 &-0.24 &0.13 &0.27 &-0.01 &0.04 &0.05 &-0.00 &0.01 &0.01 \\ \cline{2-15}
& \multirow{5}{*}{\rotatebox{90}{$x:95\%$}} & Tail &-0.00 &0.00 &0.00 &-0.00 &0.00 &0.00 &0.00 &0.01 &0.01 &0.01 &0.01 &0.02 \\ 
& & Logit, tail $X$ &-1.00 &0.02 &1.00 &-0.99 &0.00 &0.99 &-0.97 &0.01 &0.97 &-0.94 &0.01 &0.94 \\ 
& & Logit, all $X$ &-0.00 &0.03 &0.03 &0.01 &0.00 &0.01 &0.03 &0.00 &0.03 &0.04 &0.00 &0.04 \\ 
& & Local Logit &0.00 &0.00 &0.00 &0.00 &0.01 &0.01 &0.00 &0.02 &0.02 &0.00 &0.02 &0.02 \\ 
& & Local linear &-0.51 &0.02 &0.51 &-0.06 &0.17 &0.18 &0.01 &0.04 &0.04 &-0.00 &0.01 &0.01 \\ \cline{2-15}
& \multirow{5}{*}{\rotatebox{90}{$x:97.5\%$}} & Tail &-0.00 &0.00 &0.00 &-0.00 &0.00 &0.00 &-0.00 &0.01 &0.01 &0.00 &0.01 &0.01 \\ 
& & Logit, tail $X$ &-1.00 &0.02 &1.00 &-0.99 &0.03 &0.99 &-0.97 &0.06 &0.97 &-0.93 &0.08 &0.93 \\ 
& & Logit, all $X$ &-0.00 &0.02 &0.02 &0.00 &0.00 &0.00 &0.01 &0.00 &0.01 &0.02 &0.00 &0.02 \\ 
& & Local Logit &-0.12 &0.21 &0.24 &0.00 &0.01 &0.01 &0.00 &0.02 &0.02 &0.00 &0.03 &0.03 \\ 
& & Local linear &-0.50 &0.03 &0.50 &0.06 &0.14 &0.16 &0.01 &0.04 &0.04 &-0.00 &0.02 &0.02 \\ \cline{2-15}
& \multirow{5}{*}{\rotatebox{90}{$x:99\%$}} & Tail &-0.00 &0.00 &0.00 &-0.00 &0.00 &0.00 &-0.00 &0.00 &0.00 &-0.00 &0.01 &0.01 \\ 
& & Logit, tail $X$ &-0.00 &0.02 &0.02 &0.00 &0.00 &0.00 &-0.00 &0.01 &0.01 &-0.00 &0.02 &0.02 \\ 
& & Logit, all $X$ &-0.00 &0.02 &0.02 &0.00 &0.00 &0.00 &0.00 &0.00 &0.00 &0.01 &0.00 &0.01 \\ 
& & Local Logit &-0.49 &0.08 &0.49 &-0.00 &0.02 &0.02 &0.00 &0.02 &0.02 &0.00 &0.03 &0.03 \\ 
& & Local linear &-0.43 &0.12 &0.44 &0.07 &0.12 &0.14 &0.01 &0.05 &0.05 &0.00 &0.02 &0.02 \\ \hline

\end{tabular} 
} 
\end{center}
{\footnotesize {\em Notes:} Bias, SD, and RMSE are calculated based on $N_{sim}=1000$ repetitions w.r.t.\ the theoretical values.}\setlength{\baselineskip}{4mm}
\end{table}

\begin{table}[h]
\caption{Partial effects estimation - Experiment 1}
\label{tab:sim-exp1-ape}
\vspace{-1em}
\begin{center}
\scalebox{.8}{
\begin{tabular}{lll|rrr|rrr|rrr|rrr} \hline \hline 
& & & \multicolumn{3}{c|}{$\alpha_X=0.5$} & \multicolumn{3}{c|}{$\alpha_X=1$} &\multicolumn{3}{c|}{$\alpha_X=1.5$} &\multicolumn{3}{c}{$\alpha_X=2$} \\
& & & Bias & SD & \footnotesize RMSE & Bias & SD &\footnotesize RMSE & Bias & SD &\footnotesize RMSE & Bias & SD &\footnotesize RMSE \\ \hline
\multirow{20}{*}{\rotatebox{90}{$\alpha_{\varepsilon}=0.5$}} & \multirow{5}{*}{\rotatebox{90}{$x:90\%$}} & Tail &0.01 &0.03 &0.03 &0.03 &0.07 &0.07 &0.02 &0.09 &0.09 &0.03 &0.09 &0.09 \\ 
& & Logit, tail $X$ &-0.88 &0.05 &0.88 &-0.68 &0.08 &0.68 &-0.60 &0.07 &0.61 &-0.55 &0.07 &0.56 \\ 
& & Logit, all $X$ &-0.18 &0.13 &0.22 &-0.13 &0.03 &0.13 &-0.07 &0.02 &0.08 &-0.05 &0.01 &0.05 \\ 
& & Local Logit &0.01 &0.05 &0.06 &0.00 &0.04 &0.04 &0.00 &0.03 &0.03 &-0.00 &0.03 &0.03 \\ 
& & Local linear &-0.41 &0.02 &0.41 &-0.12 &0.06 &0.13 &-0.01 &0.02 &0.02 &-0.00 &0.02 &0.02 \\ \cline{2-15}
& \multirow{5}{*}{\rotatebox{90}{$x:95\%$}} & Tail &0.00 &0.02 &0.02 &0.01 &0.04 &0.04 &0.02 &0.06 &0.06 &0.02 &0.06 &0.06 \\ 
& & Logit, tail $X$ &-0.93 &0.06 &0.93 &-0.73 &0.09 &0.73 &-0.62 &0.08 &0.63 &-0.56 &0.08 &0.56 \\ 
& & Logit, all $X$ &-0.02 &0.13 &0.13 &-0.05 &0.06 &0.08 &-0.03 &0.03 &0.04 &-0.01 &0.02 &0.02 \\ 
& & Local Logit &0.01 &0.12 &0.12 &0.01 &0.07 &0.07 &0.00 &0.05 &0.05 &-0.00 &0.04 &0.04 \\ 
& & Local linear &-0.45 &0.02 &0.45 &-0.06 &0.09 &0.11 &-0.00 &0.02 &0.02 &-0.00 &0.02 &0.02 \\ \cline{2-15}
& \multirow{5}{*}{\rotatebox{90}{$x:97.5\%$}} & Tail &-0.00 &0.01 &0.01 &0.00 &0.03 &0.03 &0.01 &0.03 &0.04 &0.02 &0.04 &0.04 \\ 
& & Logit, tail $X$ &-0.94 &0.07 &0.94 &-0.70 &0.09 &0.70 &-0.57 &0.07 &0.58 &-0.50 &0.06 &0.51 \\ 
& & Logit, all $X$ &0.00 &0.08 &0.08 &0.05 &0.06 &0.08 &0.05 &0.03 &0.06 &0.04 &0.02 &0.05 \\ 
& & Local Logit &-0.11 &0.25 &0.27 &0.02 &0.12 &0.12 &0.01 &0.09 &0.09 &0.01 &0.07 &0.07 \\ 
& & Local linear &-0.47 &0.03 &0.47 &0.00 &0.09 &0.09 &0.00 &0.03 &0.03 &0.00 &0.04 &0.04 \\ \cline{2-15}
& \multirow{5}{*}{\rotatebox{90}{$x:99\%$}} & Tail &-0.00 &0.01 &0.01 &-0.00 &0.02 &0.02 &0.00 &0.03 &0.03 &0.00 &0.03 &0.03 \\ 
& & Logit, tail $X$ &-0.03 &0.10 &0.11 &-0.17 &0.13 &0.21 &-0.23 &0.09 &0.25 &-0.22 &0.07 &0.23 \\ 
& & Logit, all $X$ &0.00 &0.05 &0.05 &0.07 &0.04 &0.08 &0.11 &0.02 &0.12 &0.11 &0.02 &0.11 \\ 
& & Local Logit &-0.48 &0.08 &0.48 &-0.01 &0.27 &0.27 &0.02 &0.19 &0.19 &0.01 &0.14 &0.14 \\ 
& & Local linear &-0.41 &0.11 &0.43 &0.03 &0.08 &0.08 &0.01 &0.06 &0.06 &-0.00 &0.08 &0.08 \\ \hline
\multirow{20}{*}{\rotatebox{90}{$\alpha_{\varepsilon}=1$}} & \multirow{5}{*}{\rotatebox{90}{$x:90\%$}} & Tail &-0.00 &0.01 &0.01 &0.01 &0.03 &0.03 &0.02 &0.05 &0.06 &0.04 &0.06 &0.07 \\ 
& & Logit, tail $X$ &-0.98 &0.02 &0.98 &-0.90 &0.01 &0.90 &-0.83 &0.02 &0.83 &-0.79 &0.01 &0.79 \\ 
& & Logit, all $X$ &0.01 &0.04 &0.04 &0.06 &0.01 &0.06 &0.03 &0.02 &0.03 &0.02 &0.01 &0.02 \\ 
& & Local Logit &0.00 &0.02 &0.02 &0.00 &0.03 &0.03 &-0.00 &0.02 &0.02 &-0.00 &0.02 &0.02 \\ 
& & Local linear &-0.51 &0.02 &0.51 &-0.20 &0.10 &0.23 &-0.01 &0.03 &0.04 &-0.00 &0.01 &0.01 \\ \cline{2-15}
& \multirow{5}{*}{\rotatebox{90}{$x:95\%$}} & Tail &-0.00 &0.00 &0.00 &0.00 &0.02 &0.02 &0.01 &0.03 &0.03 &0.02 &0.03 &0.04 \\ 
& & Logit, tail $X$ &-1.00 &0.02 &1.00 &-0.94 &0.02 &0.94 &-0.88 &0.03 &0.88 &-0.83 &0.03 &0.83 \\ 
& & Logit, all $X$ &0.00 &0.03 &0.03 &0.05 &0.00 &0.05 &0.08 &0.01 &0.08 &0.07 &0.01 &0.07 \\ 
& & Local Logit &0.00 &0.04 &0.04 &0.00 &0.04 &0.04 &0.00 &0.04 &0.04 &-0.00 &0.03 &0.03 \\ 
& & Local linear &-0.52 &0.02 &0.52 &-0.07 &0.14 &0.16 &0.00 &0.03 &0.03 &-0.00 &0.02 &0.02 \\ \cline{2-15}
& \multirow{5}{*}{\rotatebox{90}{$x:97.5\%$}} & Tail &-0.00 &0.00 &0.00 &-0.00 &0.01 &0.01 &0.00 &0.02 &0.02 &0.01 &0.02 &0.02 \\ 
& & Logit, tail $X$ &-1.00 &0.02 &1.00 &-0.94 &0.06 &0.94 &-0.83 &0.07 &0.84 &-0.78 &0.06 &0.78 \\ 
& & Logit, all $X$ &-0.00 &0.02 &0.02 &0.02 &0.00 &0.02 &0.06 &0.00 &0.06 &0.08 &0.00 &0.08 \\ 
& & Local Logit &-0.12 &0.21 &0.24 &0.01 &0.05 &0.05 &0.01 &0.05 &0.05 &0.01 &0.05 &0.05 \\ 
& & Local linear &-0.51 &0.03 &0.51 &0.04 &0.13 &0.14 &0.01 &0.03 &0.03 &0.00 &0.03 &0.03 \\ \cline{2-15}
& \multirow{5}{*}{\rotatebox{90}{$x:99\%$}} & Tail &-0.00 &0.00 &0.00 &-0.00 &0.01 &0.01 &-0.00 &0.01 &0.01 &0.00 &0.02 &0.02 \\ 
& & Logit, tail $X$ &-0.00 &0.02 &0.02 &-0.01 &0.04 &0.04 &-0.05 &0.07 &0.09 &-0.07 &0.08 &0.11 \\ 
& & Logit, all $X$ &-0.00 &0.02 &0.02 &0.01 &0.00 &0.01 &0.04 &0.00 &0.04 &0.06 &0.00 &0.06 \\ 
& & Local Logit &-0.49 &0.08 &0.49 &0.00 &0.09 &0.09 &0.01 &0.08 &0.08 &0.02 &0.07 &0.08 \\ 
& & Local linear &-0.44 &0.12 &0.45 &0.06 &0.11 &0.12 &0.01 &0.05 &0.05 &0.00 &0.04 &0.04 \\ \hline

\end{tabular} 
} 
\end{center}
{\footnotesize {\em Notes:} Bias, SD, and RMSE are calculated based on $N_{sim}=1000$ repetitions w.r.t.\ the theoretical values.}\setlength{\baselineskip}{4mm}
\end{table}

\begin{table}[h]
\caption*{Table \ref{tab:sim-exp1-ape}: Partial effects estimation - Experiment 1 (cont.)}
\label{tab:sim-exp1-ape-2}
\vspace{-1em}
\begin{center}
\scalebox{.8}{
\begin{tabular}{lll|rrr|rrr|rrr|rrr} \hline \hline 
& & & \multicolumn{3}{c|}{$\alpha_X=0.5$} & \multicolumn{3}{c|}{$\alpha_X=1$} &\multicolumn{3}{c|}{$\alpha_X=1.5$} &\multicolumn{3}{c}{$\alpha_X=2$} \\
& & & Bias & SD & \footnotesize RMSE & Bias & SD &\footnotesize RMSE & Bias & SD &\footnotesize RMSE & Bias & SD &\footnotesize RMSE \\ \hline

\multirow{20}{*}{\rotatebox{90}{$\alpha_{\varepsilon}=1.5$}} & \multirow{5}{*}{\rotatebox{90}{$x:90\%$}} & Tail &-0.00 &0.00 &0.00 &0.00 &0.02 &0.02 &0.02 &0.03 &0.04 &0.03 &0.04 &0.05 \\ 
& & Logit, tail $X$ &-1.00 &0.02 &1.00 &-0.95 &0.00 &0.95 &-0.90 &0.01 &0.90 &-0.86 &0.01 &0.86 \\ 
& & Logit, all $X$ &-0.00 &0.04 &0.04 &0.05 &0.00 &0.05 &0.06 &0.01 &0.06 &0.05 &0.01 &0.05 \\ 
& & Local Logit &0.00 &0.01 &0.01 &0.00 &0.02 &0.02 &0.00 &0.02 &0.02 &-0.00 &0.02 &0.02 \\ 
& & Local linear &-0.52 &0.02 &0.52 &-0.23 &0.12 &0.26 &-0.01 &0.04 &0.04 &-0.00 &0.01 &0.01 \\ \cline{2-15}
& \multirow{5}{*}{\rotatebox{90}{$x:95\%$}} & Tail &-0.00 &0.00 &0.00 &-0.00 &0.01 &0.01 &0.00 &0.02 &0.02 &0.02 &0.02 &0.02 \\ 
& & Logit, tail $X$ &-1.00 &0.02 &1.00 &-0.98 &0.01 &0.98 &-0.94 &0.01 &0.94 &-0.91 &0.02 &0.91 \\ 
& & Logit, all $X$ &-0.00 &0.03 &0.03 &0.02 &0.00 &0.02 &0.05 &0.00 &0.05 &0.06 &0.00 &0.06 \\ 
& & Local Logit &0.00 &0.01 &0.01 &0.00 &0.02 &0.02 &0.00 &0.03 &0.03 &0.00 &0.02 &0.02 \\ 
& & Local linear &-0.52 &0.02 &0.52 &-0.06 &0.16 &0.17 &0.01 &0.03 &0.03 &-0.00 &0.02 &0.02 \\ \cline{2-15}
& \multirow{5}{*}{\rotatebox{90}{$x:97.5\%$}} & Tail &-0.00 &0.00 &0.00 &-0.00 &0.00 &0.00 &0.00 &0.01 &0.01 &0.01 &0.01 &0.01 \\ 
& & Logit, tail $X$ &-1.00 &0.02 &1.00 &-0.99 &0.03 &0.99 &-0.93 &0.07 &0.93 &-0.87 &0.07 &0.88 \\ 
& & Logit, all $X$ &-0.00 &0.02 &0.02 &0.01 &0.00 &0.01 &0.02 &0.00 &0.02 &0.05 &0.00 &0.05 \\ 
& & Local Logit &-0.12 &0.21 &0.24 &0.00 &0.04 &0.04 &0.01 &0.03 &0.03 &0.01 &0.03 &0.03 \\ 
& & Local linear &-0.50 &0.03 &0.50 &0.05 &0.14 &0.15 &0.01 &0.04 &0.04 &0.00 &0.02 &0.02 \\ \cline{2-15}
& \multirow{5}{*}{\rotatebox{90}{$x:99\%$}} & Tail &-0.00 &0.00 &0.00 &-0.00 &0.00 &0.00 &-0.00 &0.01 &0.01 &-0.00 &0.01 &0.01 \\ 
& & Logit, tail $X$ &-0.00 &0.02 &0.02 &-0.00 &0.01 &0.01 &-0.01 &0.03 &0.03 &-0.02 &0.05 &0.05 \\ 
& & Logit, all $X$ &-0.00 &0.02 &0.02 &0.00 &0.00 &0.00 &0.01 &0.00 &0.01 &0.02 &0.00 &0.02 \\ 
& & Local Logit &-0.49 &0.08 &0.49 &-0.00 &0.06 &0.06 &0.00 &0.03 &0.03 &0.01 &0.04 &0.04 \\ 
& & Local linear &-0.43 &0.12 &0.45 &0.06 &0.12 &0.14 &0.01 &0.05 &0.05 &0.00 &0.03 &0.03 \\ \hline
\multirow{20}{*}{\rotatebox{90}{$\alpha_{\varepsilon}=2$}} & \multirow{5}{*}{\rotatebox{90}{$x:90\%$}} & Tail &-0.00 &0.00 &0.00 &0.00 &0.01 &0.01 &0.01 &0.02 &0.03 &0.03 &0.03 &0.04 \\ 
& & Logit, tail $X$ &-1.00 &0.02 &1.00 &-0.97 &0.00 &0.97 &-0.93 &0.00 &0.93 &-0.90 &0.00 &0.90 \\ 
& & Logit, all $X$ &-0.00 &0.04 &0.04 &0.03 &0.00 &0.03 &0.05 &0.00 &0.05 &0.05 &0.00 &0.06 \\ 
& & Local Logit &0.00 &0.00 &0.00 &0.00 &0.01 &0.01 &0.00 &0.02 &0.02 &0.00 &0.02 &0.02 \\ 
& & Local linear &-0.51 &0.02 &0.51 &-0.24 &0.13 &0.27 &-0.01 &0.04 &0.05 &-0.00 &0.01 &0.01 \\ \cline{2-15}
& \multirow{5}{*}{\rotatebox{90}{$x:95\%$}} & Tail &-0.00 &0.00 &0.00 &-0.00 &0.00 &0.00 &0.00 &0.01 &0.01 &0.01 &0.01 &0.02 \\ 
& & Logit, tail $X$ &-1.00 &0.02 &1.00 &-0.99 &0.00 &0.99 &-0.97 &0.01 &0.97 &-0.94 &0.01 &0.94 \\ 
& & Logit, all $X$ &-0.00 &0.03 &0.03 &0.01 &0.00 &0.01 &0.03 &0.00 &0.03 &0.04 &0.00 &0.04 \\ 
& & Local Logit &0.00 &0.00 &0.00 &0.00 &0.01 &0.01 &0.00 &0.02 &0.02 &0.00 &0.02 &0.02 \\ 
& & Local linear &-0.51 &0.02 &0.51 &-0.06 &0.17 &0.18 &0.01 &0.04 &0.04 &-0.00 &0.01 &0.01 \\ \cline{2-15}
& \multirow{5}{*}{\rotatebox{90}{$x:97.5\%$}} & Tail &-0.00 &0.00 &0.00 &-0.00 &0.00 &0.00 &-0.00 &0.01 &0.01 &0.00 &0.01 &0.01 \\ 
& & Logit, tail $X$ &-1.00 &0.02 &1.00 &-0.99 &0.03 &0.99 &-0.97 &0.06 &0.97 &-0.93 &0.08 &0.93 \\ 
& & Logit, all $X$ &-0.00 &0.02 &0.02 &0.00 &0.00 &0.00 &0.01 &0.00 &0.01 &0.02 &0.00 &0.02 \\ 
& & Local Logit &-0.12 &0.21 &0.24 &0.00 &0.01 &0.01 &0.00 &0.02 &0.02 &0.00 &0.03 &0.03 \\ 
& & Local linear &-0.50 &0.03 &0.50 &0.06 &0.14 &0.16 &0.01 &0.04 &0.04 &-0.00 &0.02 &0.02 \\ \cline{2-15}
& \multirow{5}{*}{\rotatebox{90}{$x:99\%$}} & Tail &-0.00 &0.00 &0.00 &-0.00 &0.00 &0.00 &-0.00 &0.00 &0.00 &-0.00 &0.01 &0.01 \\ 
& & Logit, tail $X$ &-0.00 &0.02 &0.02 &0.00 &0.00 &0.00 &-0.00 &0.01 &0.01 &-0.00 &0.02 &0.02 \\ 
& & Logit, all $X$ &-0.00 &0.02 &0.02 &0.00 &0.00 &0.00 &0.00 &0.00 &0.00 &0.01 &0.00 &0.01 \\ 
& & Local Logit &-0.49 &0.08 &0.49 &-0.00 &0.02 &0.02 &0.00 &0.02 &0.02 &0.00 &0.03 &0.03 \\ 
& & Local linear &-0.43 &0.12 &0.44 &0.07 &0.12 &0.14 &0.01 &0.05 &0.05 &0.00 &0.02 &0.02 \\ \hline

\end{tabular} 
} 
\end{center}
{\footnotesize {\em Notes:} Bias, SD, and RMSE are calculated based on $N_{sim}=1000$ repetitions w.r.t.\ the theoretical values.}\setlength{\baselineskip}{4mm}
\end{table}

\clearpage
\floatpagestyle{plain}
\begin{figure}
\caption{Log predictive score - Experiment 2, $\alpha_X = 1$, $\alpha_{\varepsilon}=1$}
\label{fig:sim-exp2-lps}
\begin{center}
    \includegraphics[width=.8\textwidth]{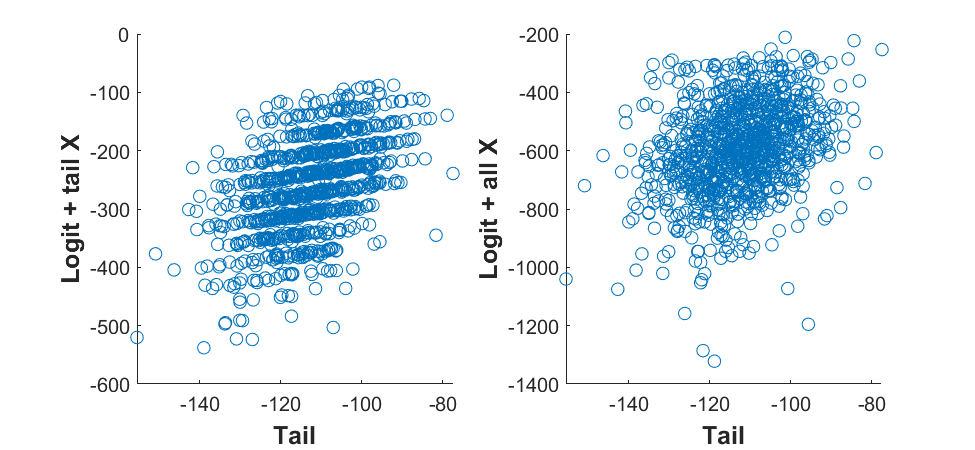}
\end{center}
{\footnotesize {\em Notes:} %Specification with $\alpha_X = 1$ and $\alpha_{\varepsilon}=1$. 
Each circle represents one Monte Carlo repetition. Note that the x- and y-scales are substantially different.}\setlength{\baselineskip}{4mm}
\end{figure}

%\newpage
%\subsection{Empirical Example: Empirical Example: Housing Price and Bank Loan Charge-Off}
\begin{figure}
\caption{Histogram - banking application, baseline sample}
\label{fig:app-hist}
\begin{center}
\includegraphics[width=.8\textwidth]{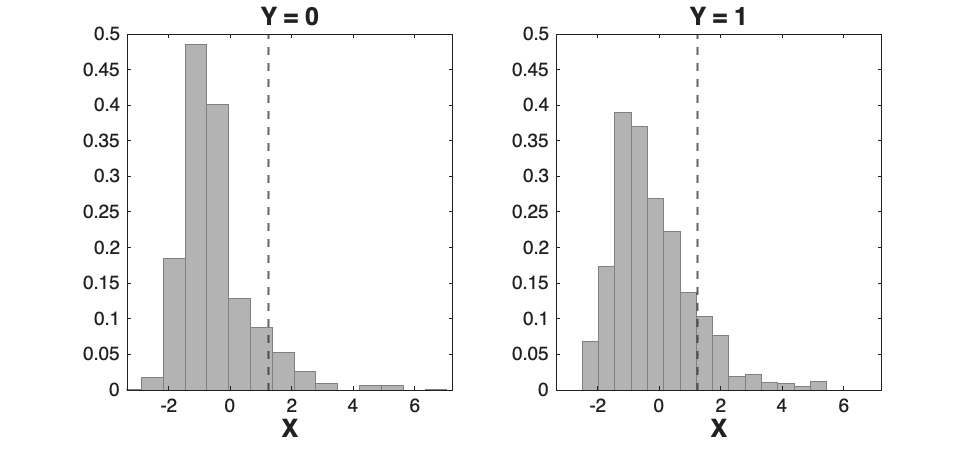}
\end{center}
{\footnotesize {\em Notes:} %Baseline sample: RRE, forecasting period = 2009Q4, $T=40$. 
Black dashed lines mark the 90th percentile of the $X_{it}$ distribution and indicate the tail region.}\setlength{\baselineskip}{4mm}
\end{figure}

\begin{table}[h]
\caption{Sample Statistics for Charge-Off Rates}
\label{tab:app-stat}
\begin{center}
\scalebox{1}{
\begin{tabular}{llllrrrrrrr} \hline \hline
& & & & \multicolumn{3}{c}{Sample Sizes} & \multicolumn{4}{c}{Sample Statistics} \\
\small Loan & \small $t_{start}$ & \small $t_{fcst}$ & \small $T$ & \small $N$ & \small $N_{e}^{\dagger}$ & \small $N_{f}^{\dagger}$ & \scriptsize $\%(Y=1)$ & \small $SD_X$ & \small $Skew_X$ & \small $Kurt_X$ \\ 
\hline
CRE & 2001Q3 & 2009Q3 & 32 & 7882 & 1286 & 1086 & 3.27 & 0.94 & 0.70 & 12.26 \\
CRE & 2000Q3 & 2009Q3 & 36 & 8224 & 1408 & 1194 & 2.90 & 0.95 & 0.73 & 12.37 \\
CRE & 1999Q3 & 2009Q3 & 40 & 8588 & 1451 & 1212 & 2.60 & 0.99 & 0.67 & 11.96 \\
CRE & 1998Q3 & 2009Q3 & 44 & 9047 & 1506 & 1260 & 2.36 & 1.00 & 0.66 & 12.28 \\
CRE & 1997Q3 & 2009Q3 & 48 & 9495 & 1546 & 1303 & 2.14 & 1.02 & 0.63 & 12.28 \\
CRE & 2001Q4 & 2009Q4 & 32 & 7798 & 1602 & 1367 & 3.91 & 0.92 & 0.73 & 10.87 \\
CRE & 2000Q4 & 2009Q4 & 36 & 8128 & 1726 & 1476 & 3.46 & 0.93 & 0.79 & 11.18 \\
CRE & 1999Q4 & 2009Q4 & 40 & 8510 & 1806 & 1533 & 3.10 & 0.95 & 0.73 & 10.96 \\
CRE & 1998Q4 & 2009Q4 & 44 & 8899 & 1909 & 1607 & 2.81 & 0.97 & 0.74 & 11.28 \\
CRE & 1997Q4 & 2009Q4 & 48 & 9382 & 1959 & 1654 & 2.56 & 0.99 & 0.72 & 11.43 \\
RRE & 2001Q3 & 2009Q3 & 32 & 7909 & 2104 & 1693 & 18.11 & 0.94 & 0.68 & 12.16 \\
RRE & 2000Q3 & 2009Q3 & 36 & 8244 & 2348 & 1868 & 16.03 & 0.95 & 0.70 & 12.26 \\
RRE & 1999Q3 & 2009Q3 & 40 & 8614 & 2497 & 1958 & 14.32 & 0.98 & 0.65 & 11.84 \\
RRE & 1998Q3 & 2009Q3 & 44 & 9085 & 2604 & 2038 & 12.94 & 1.00 & 0.64 & 12.15 \\
RRE & 1997Q3 & 2009Q3 & 48 & 9531 & 2719 & 2143 & 11.74 & 1.02 & 0.60 & 12.14 \\
RRE & 2001Q4 & 2009Q4 & 32 & 7824 & 2248 & 1813 & 19.24 & 0.92 & 0.71 & 10.80 \\
RRE & 2000Q4 & 2009Q4 & 36 & 8148 & 2470 & 1986 & 17.01 & 0.92 & 0.77 & 11.10 \\
RRE & 1999Q4 & 2009Q4 & 40 & 8538 & 2642 & 2098 & 15.18 & 0.95 & 0.71 & 10.87 \\
RRE & 1998Q4 & 2009Q4 & 44 & 8936 & 2849 & 2264 & 13.74 & 0.97 & 0.72 & 11.19 \\
RRE & 1997Q4 & 2009Q4 & 48 & 9418 & 2974 & 2366 & 12.47 & 0.98 & 0.70 & 11.32 \\
\hline
\end{tabular} 
} 
\end{center}
{\footnotesize {\em Notes:} $N$ is the total number of banks in each sample. $N_{e}^{\dagger}$ is the number of banks in the tail and contributing to the likelihood, characterized by $X_{it}$ being above the 90th percentile of the estimation sample and $Y_{it}$ switching values across $t$ for these tail observations. $N_{f}^{\dagger}$ is the number of banks that meet the criteria for $N_{e}^{\dagger}$ and, additionally, have $X_{i,T+1}$ in the tail during the forecasting period. Sample statistics are calculated for each $N$-by-$T$ panel. The panels are unbalanced due to missing values.}\setlength{\baselineskip}{4mm}
\end{table} 

\begin{table}[h]
\caption{Forecast evaluation and parameter estimation - banking application, all samples}
\label{tab:app-lps-ape-all}
\begin{center}
\scalebox{.85}{
\begin{tabular}{llrrrrr} \hline \hline 
& & $T=32$ & $T=36$& $T=40$& $T=44$& $T=48$ \\ \hline

\multicolumn{7}{c}{RRE, prediction period = 2009Q3} \\ \hline
\multirow{3}{*}{$\text{LPS}\cdot N_{f}^{\dagger}$} & \it Tail & \it -1112.60 \phantom{***} & \it -1217.91 \phantom{***} & \it -1291.65 \phantom{***} & \it -1355.10 \phantom{***} & \it -1443.19 \phantom{***} \\ \cline{2-7}
& Logit, tail $X$ & -7.65 *** & -10.01 *** & -10.75 *** & -12.46 *** & -12.84 *** \\ 
& Logit, all $X$ & -87.83 *** & -85.66 *** & -132.49 *** & -170.68 *** & -207.60 *** \\ \hline 
\multirow{4}{*}{Tail} & $\hat\theta^*_{\log X}$ & 0.96 ***& 1.00 ***& 0.97 ***& 1.00 ***& 1.01 ***\\

& $\hat\theta^*_{Z\log X}$ & 0.89 *** & 0.61 *** & 0.78 *** & 0.67 *** & 0.58 *** \\

& APE & 0.13 *** & 0.13 *** & 0.12 *** & 0.12 *** & 0.12 *** \\ \cline{2-7}
& $\hat V(\tilde A_i)$ &0.99 \phantom{***} &1.00 \phantom{***} &0.98 \phantom{***} &0.98 \phantom{***} &0.98 \phantom{***} \\

\hline\hline

\multicolumn{7}{c}{RRE, prediction period = 2009Q4} \\ \hline
\multirow{3}{*}{$\text{LPS}\cdot N_{f}^{\dagger}$} &\it Tail & \it -1232.46 \phantom{***} & \it -1347.38 \phantom{***} & \it -1433.72 \phantom{***} & \it -1559.47 \phantom{***} & \it -1649.76 \phantom{***} \\ \cline{2-7}
& Logit, tail $X$ & -14.12 *** & -17.95 *** & -18.66 *** & -21.85 *** & -22.37 *** \\ 
& Logit, all $X$ & -154.94 *** & -144.03 *** & -178.91 *** & -206.57 *** & -232.10 *** \\ \hline
\multirow{4}{*}{Tail} & $\hat\theta^*_{\log X}$ & 1.11 ***& 1.14 ***& 1.09 ***& 1.11 ***& 1.12 ***\\

& $\hat\theta^*_{Z\log X}$ & 0.52 **\phantom{*} & 0.54 **\phantom{*} & 0.63 *** & 0.53 *** & 0.42 **\phantom{*} \\

& APE & 0.16 *** & 0.16 *** & 0.15 *** & 0.15 *** & 0.14 *** \\ \cline{2-7}
& $\hat V(\tilde A_i)$ &1.06 \phantom{***} &1.06 \phantom{***} &1.04 \phantom{***} &1.03 \phantom{***} &1.03 \phantom{***} \\
\hline\hline

\multicolumn{7}{c}{CRE, prediction period = 2009Q3} \\ \hline
\multirow{3}{*}{$\text{LPS}\cdot N_{f}^{\dagger}$} & \it Tail
& \it -777.51 \phantom{***} & \it -853.64 \phantom{***} & \it -878.00 \phantom{***} & \it -910.85 \phantom{***} & \it -953.34 \phantom{***} \\ \cline{2-7}
& Logit, tail $X$ & -8.93 *** & -9.51 *** & -10.38 *** & -11.71 *** & -11.83 *** \\ 
& Logit, all $X$ & -15.95 *** & -12.40 *** & -21.61 *** & -29.59 *** & -45.66 *** \\ \hline
\multirow{4}{*}{Tail} & $\hat\theta^*_{\log X}$ & 1.20 ***& 1.25 ***& 1.26 ***& 1.29 ***& 1.29 ***\\

& $\hat\theta^*_{Z\log X}$ & 0.46 \phantom{***} & 0.30 \phantom{***} & 0.30 \phantom{***} & 0.26 \phantom{***} & 0.23 \phantom{***} \\

& APE & 0.11 *** & 0.11 *** & 0.11 *** & 0.10 *** & 0.10 *** \\ \cline{2-7}
& $\hat V(\tilde A_i)$ &1.00 \phantom{***} &1.00 \phantom{***} &1.01 \phantom{***} &1.03 \phantom{***} &1.04 \phantom{***} \\

\hline\hline
\multicolumn{7}{c}{CRE, prediction period = 2009Q4} \\ \hline
\multirow{3}{*}{$\text{LPS}\cdot N_{f}^{\dagger}$} & \it Tail & \it -1017.50 \phantom{***} & \it -1088.70 \phantom{***} & \it -1134.48 \phantom{***} & \it -1189.45 \phantom{***} & \it -1237.07 \phantom{***} \\ \cline{2-7}
& Logit, tail $X$ & -17.18 *** & -22.72 *** & -23.49 *** & -28.68 *** & -29.12 *** \\ 
& Logit, all $X$ & -50.76 *** & -68.91 *** & -77.75 *** & -95.20 *** & -106.86 *** \\ \hline 
\multirow{4}{*}{Tail} & $\hat\theta^*_{\log X}$ & 1.31 ***& 1.30 ***& 1.31 ***& 1.37 ***& 1.42 *** \\ 
& $\hat\theta^*_{Z\log X}$ & 0.07 \phantom{***} & 0.20 \phantom{***} & 0.04 \phantom{***} & -0.02 \phantom{***} & -0.10 \phantom{***} \\

& APE & 0.14 *** & 0.14 *** & 0.13 *** & 0.13 *** & 0.13 *** \\ \cline{2-7}
& $\hat V(\tilde A_i)$ &1.03 \phantom{***} &1.05 \phantom{***} &1.06 \phantom{***} &1.08 \phantom{***} &1.08 \phantom{***} \\
\hline

\end{tabular} 
} 
\end{center}
{\footnotesize {\em Notes:} For the tail estimator, the table reports the exact values of $\text{LPS}\cdot N_{f}^{\dagger}$. For other estimators, the table reports their differences from the tail estimator. The tests compare other estimators with the tail estimator. The last four rows in each panel report the estimated parameters and APEs as well as the sample variances of the estimated $\tilde A_i$'s using the tail estimator. All significance levels are indicated by *: 10\%, **: 5\%, and ***: 1\%.}\setlength{\baselineskip}{4mm}
\end{table}

\begin{table}[h]
\caption{Forecast evaluation and parameter estimation - banking application, varying $c$}
\label{tab:app-lps-ape-c}
\begin{center}
\scalebox{1}{
\begin{tabular}{llrrrr} \hline \hline 
& & $c=0$ & $c=0.01$& $c=0.02$& $c=0.05$ \\ \hline
\multirow{3}{*}{$\text{LPS}\cdot N_{f}^{\dagger}$} & \it Tail & \it -1433.72 \phantom{***} & \it -1453.03 \phantom{***} & \it -1462.00 \phantom{***} & \it -1447.46 \phantom{***} \\ \cline{2-6}
& Logit, tail $X$ & -18.66 *** & -18.71 *** & -18.21 *** & -18.63 *** \\ 
& Logit, all $X$ & -178.91 *** & -174.92 *** & -175.49 *** & -180.03 *** \\ \hline 
\multirow{4}{*}{Tail} & $\hat\theta^*_{\log X}$ & 1.09 ***& 1.08 ***& 1.08 ***& 1.09 ***\\ 
& $\hat\theta^*_{Z\log X}$ & 0.63 {***} & 0.61 {***} & 0.60 {***} & 0.58 {***} \\

& APE & 0.15 *** & 0.15 *** & 0.15 *** & 0.15 *** \\ \cline{2-6} 
& $\hat V(\tilde A_i)$ &1.06 \phantom{***} &1.05 \phantom{***} &1.04 \phantom{***} &1.04 \phantom{***} \\ \hline\hline

& & $c=0.1$ & $c=0.2$& $c=0.5$& $c=1$ \\ \hline
\multirow{3}{*}{$\text{LPS}\cdot N_{f}^{\dagger}$} & \it Tail & \it -1411.56 \phantom{***} & \it -1343.85 \phantom{***} & \it -1131.54 \phantom{***} & \it -833.62 \phantom{***} \\ \cline{2-6}
& Logit, tail $X$ & -18.40 *** & -15.77 *** & -11.74 *** & -9.10 *** \\ 
& Logit, all $X$ & -178.44 *** & -180.10 *** & -166.46 *** & -95.35 *** \\ \hline 
\multirow{4}{*}{Tail} & $\hat\theta^*_{\log X}$ & 1.14 ***& 1.12 ***& 1.11 ***& 1.22 ***\\ 
& $\hat\theta^*_{Z\log X}$ & 0.52 **\phantom{*} & 0.52 **\phantom{*} & 0.44 *\phantom{**} & 0.16 \phantom{***} \\

& APE & 0.15 *** & 0.14 *** & 0.11 *** & 0.10 *** \\ \cline{2-6} 
& $\hat V(\tilde A_i)$ &1.04 \phantom{***} &1.01 \phantom{***} &0.99 \phantom{***} &0.96 \phantom{***} \\ \hline

\end{tabular} 
} 
\end{center}
{\footnotesize {\em Notes:} Baseline sample: RRE, forecasting period = 2009Q4, $T=40$. For the tail estimator, the table reports the exact values of $\text{LPS}\cdot N_{f}^{\dagger}$. For other estimators, the table reports their differences from the tail estimator. The tests compare other estimators with the tail estimator. The last four rows in each panel report the estimated parameters and APEs as well as the sample variances of the estimated $\tilde A_i$'s using the tail estimator. All significance levels are indicated by *: 10\%, **: 5\%, and ***: 1\%.}\setlength{\baselineskip}{4mm}
\end{table}

%% file: TailBinaryPanel_header.tex
\usepackage{xr}
\usepackage{amssymb}
\usepackage{amsfonts}
\usepackage{amsmath}
\usepackage{theorem}
\usepackage{graphicx}
\usepackage{xcolor}
\usepackage[longnamesfirst]{natbib}

\newtheorem{assumption}{Assumption}[section]
\newtheorem{theorem}{Theorem}[section]

\newtheorem{proposition}{Proposition}[section]
\newtheorem{remark}{Remark}[section]

\allowdisplaybreaks
\usepackage{multirow}
\usepackage{rotating}
\usepackage{caption}
\usepackage{enumitem}
\usepackage[hidelinks]{hyperref}

\newcommand{\x}{{\underline{x}}}

\newcommand{\E}{\mathbb{E}}
\newcommand{\V}{\mathbb{V}}
\renewcommand{\P}{\mathbb{P}}
\newcommand{\I}{\mathcal{I}}
\newcommand{\C}{\mathcal{C}}
\renewcommand{\L}{\mathcal{L}}
\newcommand{\1}{\mathbf{1}}
\newcommand{\sumi}{\sum_{i=1}^N}

\newcommand{\intr}{{\text{int}}}

\newcommand{\py}{{(y)}}
\newcommand{\po}{{(0)}}
\newcommand{\pl}{{(1)}}

\makeatletter
\renewcommand\footnoterule{%
\kern-3\p@
\hrule\@width.4\columnwidth
\kern2.6\p@}
\makeatother

\makeatletter
\newcommand{\printbaselinestretch}{%
  \begingroup
    \edef\@tempa{Current \string\baselinestretch\space value: \baselinestretch}%
    \par\noindent\@tempa\par
  \endgroup
}
\makeatother

\usepackage{floatpag}

\usepackage[margin=1in]{geometry}

\usepackage{setspace}
\newif\ifsubmission